\documentclass[11pt]{article}

\usepackage[textsize=footnotesize]{todonotes}
\usepackage{epsfig}
\usepackage{amssymb,amsmath,amsthm,url}
\usepackage{amsfonts}
\usepackage{multirow}
\usepackage{booktabs}
\usepackage{natbib}
\usepackage{setspace}
\usepackage{dsfont}
\usepackage{mathrsfs}
\usepackage{eufrak}
\usepackage{mathabx}
\usepackage{rotating}
\usepackage{enumitem}

\usepackage[section]{placeins}
\usepackage{subcaption}

\usepackage{times,amsfonts,eucal}
\usepackage{algorithm}
\usepackage{algorithmic}
\usepackage{graphicx,ifthen}
\usepackage{wrapfig}
\usepackage{latexsym}
\usepackage{color}
\usepackage{mathrsfs}
\usepackage{bm}
\usepackage{bbm}
\usepackage{srctex}
\usepackage{verbatim}
\usepackage{xcolor}
\usepackage{xr}
\graphicspath{{Mrule_plots/}}

\DeclareMathAlphabet{\mathpzc}{OT1}{pzc}{m}{it}

\newcounter{alphcount}
{\begin{list}{{\upshape(}\alph{alphcount}\/{\upshape)\ }}%
             {\usecounter{alphcount}\labelwidth1.5em%
              \leftmargin2em\labelsep0.5em\topsep0.25em plus 0.5ex%
              \itemsep0.25em plus 0.5ex\parsep0em}}{\end{list}}
{\begin{list}{{\upshape#1\arabic{alphcount}}}%
             {\usecounter{alphcount}\labelwidth2.5em%
              \leftmargin2.5em\labelsep0em\topsep0.25em plus 0.5ex%
              \itemsep0.25em plus 0.5ex\parsep0em}}{\end{list}}
{\begin{list}{{\upshape\arabic{alphcount}.\ }}%
             {\usecounter{alphcount}\labelwidth1.5em%
              \leftmargin2em\labelsep0.5em\topsep0.25em plus 0.5ex%
              \itemsep0.25em plus 0.5ex\parsep0em}}{\end{list}}

\usepackage{bm}
\usepackage{natbib}
\newtheorem*{assumption**}{\assumptionnumber}
\providecommand{\assumptionnumber}{}
\makeatletter
\newenvironment{assump2}[3]
 {%
  \renewcommand{\assumptionnumber}{Assumption \upshape{I}\,(#2,#3)}%
  \begin{assumption**}%
  \protected@edef\@currentlabel{Assumption \upshape{I}}%
 }
 {%
  \end{assumption**}
 }
\makeatother

\newtheorem*{assumption***}{\assumptionnumber}
\providecommand{\assumptionnumber}{}
\makeatletter
\newenvironment{assump3}[3]
 {%
  \renewcommand{\assumptionnumber}{Assumption \upshape{I*}\,(#2,#3)}%
  \begin{assumption***}%
  \protected@edef\@currentlabel{Assumption \upshape{I*}}%
 }
 {%
  \end{assumption***}
 }
\makeatother

\newtheorem{theorem}{Theorem}[section]
\newtheorem{lemma}{Lemma}[section]
\newtheorem{pro}{Proposition}[section]
\newtheorem{cor}{Corollary}[section]
\newtheorem{definition}{Definition}

\newtheorem{assumption}{Assumption}[section]
\numberwithin{equation}{section}
\numberwithin{theorem}{section}
\numberwithin{lemma}{section}
\numberwithin{pro}{section}
\numberwithin{cor}{section}
\numberwithin{definition}{section}
\numberwithin{cons}{section}
\numberwithin{rem}{section}
\numberwithin{exa}{section}
\numberwithin{table}{section}
\numberwithin{figure}{section}

\newcommand{\mR}{\mathbb{R}}
\newcommand{\mZ}{\mathbb{Z}}
\newcommand{\mT}{\mathbb{T}}
\newcommand{\mC}{\mathbb{C}}
\newcommand{\mE}{\mathbb{E}}
\newcommand{\mP}{\mathbb{P}}
\newcommand{\y}{\tilde{Y}}
\newcommand{\p}{\partial}
\newcommand{\ackname}{Acknowledgements}
\newcommand{\vare}{\varepsilon}

\newcommand{\A}{\mathbf{A}}
\renewcommand{\a}{\mathbf{a}}
\newcommand{\C}{\mathbf{C}}
\newcommand{\D}{\mathbf{D}}
\newcommand{\I}{\mathbf{I}}
\renewcommand{\S}{\mathbf{S}}
\newcommand{\U}{\mathbf{U}}
\newcommand{\X}{\mathbf{X}}
\newcommand{\Z}{\mathbf{Z}}
\newcommand{\al}{\mathbf{\alpha}}
\newcommand{\Y}{\tilde{Y}}
\newcommand{\by}{\mathbf{y}}
\newcommand{\bs}{\mathbf{s}}
\newcommand{\bd}{\mathbf{d}}
\newcommand{\bA}{\mathbf{A}}
\newcommand{\bB}{\mathbf{B}}
\newcommand{\bI}{\mathbf{I}}
\newcommand{\la}{\langle}
\newcommand{\ra}{\rangle}
\newcommand{\varp}{\varphi}
\newcommand{\btheta}{\mathbf{\theta}}
\newcommand{\bTheta}{\mathbf{\Theta}}
\newcommand{\ubA}{\underline{\bA}}
\newcommand{\ubB}{\underline{\bB}}
\newcommand{\ubd}{\underline{\bd}}
\newcommand{\ubtheta}{\underline{\btheta}}

\newcommand{\prob}{\mathbb{P}}
\newcommand{\rnum}{\mathbb{R}}
\newcommand{\cnum}{\mathbb{C}}
\newcommand{\E}{\mathbb{E}}
\newcommand{\znum}{\mathbb{Z}}
\newcommand{\nnum}{\mathbb{N}}

\newcommand{\bigsnorm}[1]{\big\vvvert{#1}\big\vvvert}
\newcommand{\Bigsnorm}[1]{\Big\vvvert{#1}\Big\vvvert}
\newcommand{\biggsnorm}[1]{\bigg\vvvert{#1}\bigg\vvvert}
\newcommand{\biginprod}[2]{\big\langle #1, #2 \big\rangle}
\newcommand{\Biginprod}[2]{\Big\langle #1, #2 \Big\rangle}
\newcommand{\cum}{\mathrm{cum}}
\newcommand{\cumker}[1]{\mathcal{C}_{t_1,\ldots,t_{{#1}-1}}(\tau_1,\ldots,\tau_{#1})}
\newcommand{\cumk}[1]{\mathcal{C}_{t_1,\ldots,t_{{#1}-1}}}
\newcommand{\cumloc}[2]{\mathcal{C}_{#1;t_1,\ldots,t_{{#2}-1}}
}
\newcommand{\cumbp}[2]{\kappa_{{#1};t_1,\ldots,t_{{#2}-1}}}

\newcommand{\ET}{\hat E^{(T)}}

\newcommand{\floc}[2]{\F_{#1;\omega_1,\ldots,\omega_{{#2}-1}}}

\newcommand{\flocatFour}[3]{f_{#1;\omega_{#3_1},\ldots,\omega_{#3_{{#2}-1}}}}
\newcommand{\Floc}[3]{\mathcal{F}_{#1;\omega_{#3_1},\ldots,\omega_{#3_{{#2}-1}}}}
\newcommand{\Fourlocs}[3]{\tilde{\mathcal{F}}_{#3;\omega_{#2_1},\ldots,\omega_{#2_{{#1}-1}} } }
\newcommand{\Fourloc}[2]{\tilde{\mathcal{F}}_{{#2_1}+\ldots+{#2_k};\omega_{#2_1},\ldots,\omega_{#2_{{#1}-1}} } }
\newcommand{\Fopkatfour}[2]{\mathcal{F}_{\omega_{#2_1},\ldots,\omega_{#2_{{#1}-1}}}}
\newcommand{\fdft}[2]{ D^{({l_{#1}})}_{\omega_{#2}}}
\newcommand{\fdftc}[2]{ D^{({l_{#1}})}_{-\omega_{#2}}}
\newcommand{\Fop}[1]{\hat{\mathcal{F}}^{(T)}_{\omega_{#1}}}
\newcommand{\F}{\mathcal{F}}
\newcommand{\tprojk}[3]{\gamma^{(l_#1,l_#2)}_{#3}}

\newcommand{\fdftinn}[3]{\langle D_{\omega_{#1}}, \phi^{\omega_{#2}}_{l_{#3}}\rangle }

\newcommand{\im}{\mathrm{i}}
\newcommand{\inprod}[2]{\langle #1, #2 \rangle}	
\newcommand{\snorm}[1]{\vvvert{#1}\vvvert}
\newcommand{\var}{\mathrm{Var}}
\newcommand{\Var}{\var}
\newcommand{\Cov}{\mathrm{Cov}}
\newcommand{\cov}{\mathrm{Cov}}
\newcommand{\Hspace}{\mathbb{H}}
\newcommand{\tproj}{{\bm{\beta}}^{(T)}}
\newcommand{\tprojes}{\hat{\bm{\beta}}^{(T)}}

\newcommand{\WT}{{w}^{(T)}_{h}}
\newcommand{\Wr}[1]{{w}^{(T)}_{h_{#1}}}

\newcommand{\Xu}[1]{X^{(u)}_{#1}}
\newcommand{\Xuu}[2]{X^{(#1)}_{t_{#2}}}
\newcommand{\XT}[1]{X^{(T)}_{#1}}

\newcommand{\T}{\top}
\newcommand\cC{{\mathcal C}}
\newcommand{\cN}{{\mathcal N}}

\newcommand{\intt}{\int\hspace{-.2cm}\int}
\newcommand{\inttt}{\int\hspace{-.2cm}\int\hspace{-.2cm}\int }

\newcommand*\tageq{\refstepcounter{equation}\tag{\theequation}}
\newcommand{\lpr}{{l^\prime}}
\newcommand{\mpr}{{m^\prime}}
\def\ol{\overline}
\def\ic{\mathbf{i}}
\def\tr{\mathrm{Tr}}

\def\beq{\begin{equation}}
\def\eeq{\end{equation}}
\def\bals{\begin{align*}}
\def\eals{\end{align*}}

\def\bal{\begin{align}}
\def\eal{\end{align}}

\allowdisplaybreaks

\begin{document}

\title{Testing for stationarity of functional time series \\ in the frequency domain
\footnote{AA was partially supported by NSF grants DMS 1305858 and DMS 1407530.  AvD was partially supported by Maastricht University, the contract ``Projet d'Actions de Recherche Concert{\'e}es'' No. 12/17-045 of the ``Communaut{\'e} fran\c{c}aise de Belgique'' and by the Collaborative Research Center 
``Statistical modeling of nonlinear dynamic processes'' (SFB 823, Project A1, C1, A7) of the German
Research Foundation (DFG).}}

\author{
Alexander Aue\footnote{Department of Statistics, University of California, Davis, CA 95616, USA, email: \tt{aaue@ucdavis.edu}}
\and Anne van Delft\footnote{Ruhr-Universit{\"a}t Bochum, Fakult{\"a}t f{\"u}r Mathematik, 44780 Bochum, Germany, email: \tt{Anne.vanDelft@rub.de}
}}
\date{\today}
\maketitle

\begin{abstract}
\setlength{\baselineskip}{1.665em}
Interest in functional time series has spiked in the recent past with papers covering both methodology and applications being published at a much increased pace. This article contributes to the research in this area by proposing a new stationarity test for functional time series based on frequency domain methods. \textcolor{black}{The proposed test statistics is based on joint dimension reduction via functional principal components analysis across the spectral density operators at all Fourier frequencies, explicitly allowing for frequency-dependent levels of truncation to adapt to the dynamics of the underlying functional time series.}
The properties of the test are derived both under the null hypothesis of stationary functional time series and under the smooth alternative of locally stationary functional time series. The methodology is theoretically justified through asymptotic results. Evidence from simulation studies and an application to annual temperature curves suggests that the test works well in finite samples.  

\end{abstract}

\setlength{\baselineskip}{1.665em}

\section{Introduction}
\label{sec:intro}
The aim of this paper is to provide a new stationarity test for functional time series based on frequency domain methods. Particular attention is given to taking into account alternatives allowing for smooth variation as a source of non-stationarity, even though non-smooth alternatives \textcolor{black}{are covered within the simulation study}. Functional data analysis has seen an upsurge in research contributions for at least one decade. This is reflected in the growing number of monographs in the area. Readers interested in the current state of statistical inference procedures may consult \cite{b00}, \cite{fv}, \cite{hk12}, \cite{he15} and \cite{rs10}. 

Papers on functional time series have come into the focus more recently and constitute now an active area of research. \cite{hk} introduced a general weak dependence concept for stationary functional time series, while \cite{vde16} provided a framework for locally stationary functional time series. \cite{as03}, \cite{adh} and \cite{bcs00} constructed prediction methodology that may find application across many areas of science, economics and finance. With the exception of \cite{vde16}, the above contributions are concerned with procedures in the time domain. Complementing methodology in the frequency domain has been developed in parallel. One should mention \cite{pt13}, who provided results concerning the Fourier analysis of time series in function spaces, and \cite{hkh15}, who addressed the problem of dimension reduction for functional time series using dynamic principal components. 

The methodology proposed in this paper provides a new frequency domain inference procedure for functional time series. More precisely, tests for second-order stationarity are developed. In the univariate case, such tests have a long history, going back at least to the seminal paper \cite{psr69}, who based their method on the evaluation of evolutionary spectra of a given time series. Other contributions building on this work include \cite{vsn99}, who used local periodograms and wavelet analysis, and \cite{p09}, whose test is based on comparing a local estimate of the spectral density to a global estimate. \citet{dpv11} and \cite{dpv13} developed methods to derive both a measure of and a test for stationarity in locally stationary time series, the latter authors basing their method on empirical process theory. In all papers, interest is in smoothly varying alternatives. The same tests, however, also \textcolor{black}{tend to} have power against non-smooth alternatives such as structural breaks or change-points. A recent review discussing methodology for structural breaks in time series is \cite{ah}, while \citet{ars17+} is a recent contribution to structural breaks in functional time series. 

The proposed test for second-order stationarity of functional time series seeks to exploit that the Discrete Fourier Transform (DFT) of a functional time series evaluated at distinct Fourier frequencies are asymptotically uncorrelated if and only if the series is second-order stationary. The proposed method is therefore related to the initial work of \citet{dsr16}, who put forth similar tests in a univariate framework. Their method has since been generalized to multivariate time series in \citet{jsr15} as well as to spatial and spatio-temporal data by \citet{bsr17} and \citet{bjsr17}, respectively. A different version of functional stationarity tests, based on time domain methodology involving cumulative sum statistics \citep{ah}, was given in \cite{hkr14}. 

\textcolor{black}{The intrinsic variation of a functional time series is always larger than any sample size, and standard results known from univariate and multivariate time series analysis do not directly apply. From a practical perspective this brings to the fore the question of how to compress this infinite-dimensional variation to finite dimension in a meaningful way, as there is a complex interplay between dynamics occurring across frequencies and the function space. This means that dimension reduction has to be done jointly across estimated spectral density operators at all Fourier frequencies, yet separately as the exact level of dimension reduction has to be decided per frequency. The proposed test statistics collect these different sets of projections, obtained via functional principal components analysis, into a quadratic form encapsulating the second-order dynamics. 
To derive the large-sample behavior of this statistic under both the null hypothesis of a stationary time series and the alternative of a locally stationary functional time series requires new, and perhaps independently interesting, results on distributional convergence of a cross-periodogram operator in function space, where verifying existence of the limit process and tightness are nontrivial tasks. The subsequent proofs of distributional convergence of the test statistics which require taking into account the pecularities of fPCA estimators, are also complex and new. The main results are derived under the assumption that the curves are observed in their entirety, corresponding to a setting in which functions are sampled on a dense grid rather than a sparse grid. Differences for these two cases have been worked out in \citet{lh10}.}

The remainder of the paper is organized as follows. Section \ref{sec:setting} provides background, gives requisite notations, introduces properties of functional version of the DFT and gives intuition for the test. The exact form of the hypothesis test, model assumptions and the test statistics are introduced in Section \ref{sec:test}. The large-sample behavior under the null hypothesis of second-order stationarity and the alternative of local stationarity is established in Sections \ref{sec:large}. Empirical aspects are highlighted in Section \ref{sec:empirical}. The proofs are technical and relegated to the Appendix. Several further auxiliary results are proved in the supplementary document \citet{avd16+}, henceforth referred to simply as the Online Supplement.

\section{Notation and setup}
\label{sec:setting}

\label{subsec:setting:func_space_op}

A functional time series $(X_t\colon t\in\mathbb{Z})$ will be viewed as a sequence of random elements on a probability space $(\Omega,{\cal A},P)$ with paths in a separable Hilbert space. Without loss of generality, we shall focus on processes taking values in $H_\rnum=L^2_{\rnum}([0,1])$, the space of equivalence classes of real-valued, square integrable functions on the unit interval $[0,1]$. Because the methodology introduced in this paper is based on a frequency domain approach, we shall make extensive use of the complex Hilbert space $H=L^2_\cnum([0,1])$. We briefly introduce notation and relevant properties of this space and associated operators. The complex conjugate of $z \in \cnum$ is denoted by $\overline{z}$ and the imaginary number by $\im$. For $f,g \in H$, the inner product and the induced $L_2$-norm on $H$ are respectively given by 
\begin{align}
\label{eq:innerprod_norm}
\langle f, g \rangle = \int_{0}^{1} f(\tau)\overline{g(\tau)}d\tau
\qquad\text{and}\qquad
\|f\|_2 =  \sqrt{\langle f, f\rangle}.
\end{align}
Two elements of $H$ 
are understood to be equal if their difference has vanishing $L_2$-norm. More generally, for measurable functions $g\colon [0,1]^k \to \cnum$, the $L^p$-norm shall be denoted by $\|g\|_p$
and the supremum norm by $\|g\|_{\infty} = \sup_{\boldsymbol{\tau} \in [0,1]^k  }|g(\boldsymbol{\tau})|$.

Next, some properties of linear operators on $H$ are stated. Denote by $S_{\infty}(H)$ the Banach space of bounded linear operators $A\colon H \to H$ equipped with the operator norm $\snorm{A}_\infty=\sup_{\|g\|_2\le1}\|Ag\|_2$. For all $f,g \in H$, the adjoint operator of $A$, denoted by $A^{\dagger}$, is defined by $\inprod{Af}{g}=\inprod{f}{A^{\dagger}g} $ and the conjugate operator of $A$ is given by $\overline{A}g=\overline{(A\overline{g})}$. An operator $A$ is called self-adjoint if $\inprod{Af}{g} = \inprod{f}{A g}$ for all $f,g \in H$ and non-negative definite if $\inprod{Ag}{g}\ge 0$ for all $g \in H$. For $v\in H$, define the tensor product $f \otimes g\colon H \otimes H \to H$ as the bounded linear operator $(f \otimes g)v =\inprod{v}{g}f$.
A compact operator $A$ admits a {\em singular value decomposition}
\begin{align}
\label{eq:svd}
A= \sum_{n=1}^\infty s_n(A)\,\psi_n \otimes \phi_n,
\end{align}
where $(s_n(A)\colon n\in\nnum)$, are the {\em singular values} of $A$, $(\phi_n\colon n\in\nnum)$ and $(\psi_n\colon n\in\nnum)$ orthonormal bases of $H$. 
The singular values are ordered to form a monotonically decreasing sequence of non-negative numbers. 
A compact operator $A$ is said to belong to the {\em Schatten $p$-class} $S_p(H)$ if and only if the sequence $s(A)=(s_n(A)\colon n\in\nnum)$ of singular values of $A$ belongs to the sequence space $\ell^p$, so if and only if 
$\snorm{A}_p= {
(\sum_{n=1}^\infty s_n^p(A) 
)}^{1/p} < \infty$,
where $\snorm{A}_p$ is referred to as the {\em Schatten $p$-norm}.
Relevant here are $S_1(H)$, the space of trace-class operators, and particularly $S_2(H)$, the space of Hilbert--Schmidt operators. The latter is also a Hilbert space with inner product $\inprod{A}{B}_{S} = \sum_{i =1}^{\infty}\inprod{A\psi_i}{B\psi_i}$ where $A, B \in S_2(H)$ and $(\psi_n\colon n\in\nnum)$ is an ONB of $H$. The mapping  $\mathcal{T}\colon H \otimes H \to S_2(H)$ defined by the linear extension of 
$
\mathcal{T}(f \otimes g) = f \otimes \overline{g}$ is an isometric isomorphism and defines a Hilbert--Schmidt operator with kernel in $H \times H$ given by $(f\otimes g)(\tau,\sigma)=f(\tau)\overline{g}(\sigma)$, $\tau, \sigma \in [0,1]$. As a consequence, $A \in S_2(H)$ if and only if there exists $a \in H \times H$ such that $\snorm{A}_{2}=\|a\|_2$. Further useful properties needed in the proofs of the various statements of this paper are relegated to the Appendix and the Online Supplement.



\subsection{Dependence structure on the function space} 
\label{subsec:setting:dep_struc}

Let $L^2_{\cnum}(\Omega)$ be the Hilbert space with elements satisfiying $\E[\|X\|^2_2] < \infty$ and denote by $\E [X]$ the mean function of $X$, where the expectation should be viewed in the sense of a Bochner integral. For $X, Y \in L^2_{\cnum}(\Omega)$, the covariance operator $\mathcal{C}_{X,Y}\colon H \otimes H \to H$ is defined as $\mathcal{C}_{X,Y} =\E[(X -\E[ X]) \otimes (Y-\E[ Y])]$ and belongs to $S_2(H)$.  A functional time series $X=(X_t\colon t\in\znum)$ is called strictly stationary if, for all finite sets of indices $J \subset \mathbb{Z}$, the joint distribution of $(X_{t+j}\colon j \in J)$ does not depend on $t\in\znum$. Similarly, $X$ is weakly stationary if its first- and second-order moments exist and are invariant under translation in time. Without loss of generality, it is assumed throughout that 
$\E [X_t] =0$ and that $X_t \in L^2_\rnum(\Omega)$ for all $t \in \mathbb{Z}$. The lag-$h$ covariance operator between $X_t$ and $X_{t+h}$ is denoted by
\[\mathcal{C}_{t,h} = \E[X_{t+h} \otimes X_{t}]\]
which reduces to $\mathcal{C}_h=  \E[X_h \otimes X_{0}]$ in case of weak stationarity. Note that this object is a non-negative definite element of $S_1(H_\rnum)$ for $h=0$. The covariance operator $\mathcal{C}_h$ can be shown to form a Fourier pair with a non-negative Hermitian element of $S_p(H)$. 
Provided sufficiently fast decay of the second-order structure, the {\em spectral density operator} $\mathcal{F}_{\omega}$ is well-defined and given by the Fourier transform of $\mathcal{C}_h$,
\begin{align}
\label{eq:Fomega}
\mathcal{F}_{\omega}
= \frac{1}{2\pi}\sum_{h \in \mathbb{Z}} \mathcal{C}_h\,e^{-\im\omega h}.
\end{align}
A sufficient condition for the existence of $\mathcal{F}_{\omega}$ in $S_p(H)$ is $\sum_{h\in \mathbb{Z}} \snorm{\mathcal{C}_h}_p < \infty$.\medskip

Higher-order dependence among the functional observations is defined through cumulant mixing conditions \citep{b81,br67}. For this, the notion of higher-order cumulant tensors is required; see Appendix \ref{subsec:proof:cumulants} for their definition and a discussion of their properties for nonstationary functional time series.

\subsection{The functional discrete Fourier transform}
\label{sec:fDFT}

The starting point of this paper is the following proposition that characterizes second-order stationary behavior of a functional time series in terms of a spectral representation.  Its proof is in Appendix \ref{subsec:proof:cramer}.
\begin{pro}\label{start}
A zero-mean, $H$-valued stochastic process $(X_t\colon t \in\znum)$ whose spectral measure is trace class 
admits the representation
\begin{align}\label{cramrep}
X_t = \int_{-\pi}^{\pi} e^{\im t \omega} d Z_{\omega} \qquad \text{a.s.,}
\end{align}
where $(Z_{\omega}\colon\omega\in(-\pi,\pi])$ is a right-continuous functional orthogonal-increment process, if and only if 
it is weakly stationary. 
\end{pro}
If the process is not weakly stationary, then a representation in the frequency domain is not necessarily well-defined and certainly not with respect to complex exponential basis functions. However, a time-dependent functional Cram\'er representation exists if the characteristics of the process are captured by a Bochner-measurable mapping that is an evolutionary operator-valued mapping in time direction \citep{vde16}.
Assume that the functions $X_1,\ldots,X_{T}$ have been observed. If the process is weakly stationary, the {\em functional Discrete Fourier Transform} (fDFT) evaluated at frequency $\omega$, given by
\begin{align} \label{eq:fDFT}
D^{(T)}_{\omega} =\frac{1}{\sqrt{2 \pi T}} \sum_{t=1}^{T} X_{t} e^{-\im \omega t},
\end{align}  
can be seen as an estimate of the increment process $Z_{\omega}$ and exists almost surely as an element of $H$. The functional time series itself can then be represented through the inverse fDFT as
\begin{align}\label{eq:invfDFT}
 X_t= \sqrt{\frac{2 \pi}{T}} \sum_{j=1}^{T}D^{(T)}_{\omega_{j}} e^{\im \omega_{j} t}.
\end{align}
Under regularity conditions, a set of fDFTs evaluated at distinct frequencies yield asymptotically independent Gaussian random elements in $H$ and, for fixed $\omega$, one has $\Var(D^{(T)}_{\omega}) \to \F_{\omega}$ \citep{pt13}. The fDFT sequence of a Hilbertian-valued stationary process is in particular asymptotically uncorrelated at the canonical frequencies $\omega_{j} = {2\pi j}/{T}$. Consequently, provided the series is weakly stationary, for $j \ne j^\prime$ or $j \ne T-j^\prime$, we have 
$\snorm{\cov(D^{(T)}_{\omega_{j}},D^{(T)}_{\omega_{j^\prime}})}_2= O({1}/{T})$. In other words, the lag-$h$ covariance operator of the fDFT converges in norm and hence weak operator topology to the zero operator as $T \to \infty$. Similar to the above, the reverse argument (uncorrelatedness of the functional DFT sequence implies weak stationarity) can be shown by means of the inverse fDFT. 
Using expression \eqref{eq:fDFT}, the covariance operator $\mathcal{C}_{t,h}$ of $X_{t+h}$ and $X_{t}$ can be written in terms of the fDFT sequence as
\begin{align*}
\mathcal{C}_{t,h} 
=\frac{2 \pi}{T} \sum_{j,j^\prime=1}^{T} \E [ D^{(T)}_{\omega_{j}} \otimes D^{(T)}_{\omega_{j^\prime}}]e^{\im \omega_{j}h} 
=\frac{2 \pi}{T} \sum_{j=1}^{T} \E [ D^{(T)}_{\omega_{j}} \otimes D^{(T)}_{\omega_{j}}]e^{\im \omega_{j}h} 
= \mathcal{C}_h, 
\end{align*}
where the equality holds in an $L^2$-sense. This demonstrates that the autocovariance kernel of a second-order stationary functional time series is obtained and, hence, that an uncorrelated fDFT sequence implies second-order stationarity up to lag $T-1$. \textcolor{black}{The fDFT thus captures exactly the defining property of a weakly stationary process and provides a natural starting point for a test of stationarity. It is, however, a nontrivial task to construct a test statistic that optimally extracts the information contained in the infinite-dimensional process to finite dimensions. Not only can the dependence structure and the resulting dynamics of a functional time series be of a complicated nature (see Figure~\ref{fig:eig} and the example given in Section \ref{func_not_mult} of the Online Supplement), but the process will vary along both frequency and functional directions. To construct a powerful test it is therefore crucial to understand how the fDFT's behave when weak stationarity is violated. In accordance with aforementioned time series literature, the theoretical behavior of the fDFT sequence under smooth alternatives is studied. These properties will then be exploited to verify large-sample results for a testing framework for functional stationarity.}

\section{The functional stationarity testing framework}
\label{sec:test}

This section gives precise formulations of the hypotheses of interest, states the main assumptions of the paper and introduces the test statistics. Throughout, interest is in testing the null hypothesis 
\[
H_0\colon (X_t\colon t\in\mathbb{Z})~\mbox{is a weakly stationary functional time series \phantom{locally.}}
\]
versus the alternative
\[
H_A\colon (X_t\colon t\in\mathbb{Z})~\mbox{is a locally stationary functional time series},
\]
where 
locally stationary functional time series are defined as follows. 
\begin{definition}
\label{lsp}
A stochastic process $(X_t\colon t\in\mathbb{Z})$ taking values in $H_\rnum$ is said to be {\it locally stationary} if 
\begin{enumerate}
\itemsep-.4ex
\item[(1)] $X_t=X_t^{(T)}$ for $t=1,\ldots,T$ and $T\in\mathbb{N}$; and
\item[(2)] for any rescaled time $u\in[0,1]$, there is a strictly stationary process $(X_t^{(u)}\colon t\in\mathbb{Z})$ such that 
\[
\big\|X_t^{(T)}-X_t^{(u)}\big\|_2\leq\Big(\Big|\tfrac tT-u\Big|+\tfrac 1T\Big)P_{t,T}^{(u)}
\qquad\mbox{a.s.},
\]
where $P_{t,T}^{(u)}$ is a positive, real-valued triangular array of random variables such that, for some $\rho>0$, $\mathbb{E}[|P_{t,T}^{(u)}|^\rho]<\infty$ for all $t$ and $T$, uniformly in $u\in[0,1]$.
\end{enumerate}
\end{definition}

Note that, under $H_A$, the process constitutes a triangular array of functions. Inference methods are then based on in-fill asymptotics as popularized in \citet{d97} for univariate time series. The process is then considered to be observed on a finer grid as $T$ increases such that more observations are available at a local level. A rigorous statistical framework for locally stationary functional time series was recently provided in \citet{vde16}. Note that weakly stationary processes are included in Definition \ref{lsp}, which then reduces to standard asymptotics.

Based on the observations in Section \ref{sec:fDFT}, a test for weak stationarity can be set up exploiting the uncorrelatedness of the elements in the sequence $(D^{(T)}_{\omega_{j}}\colon j=1,\ldots,T)$. \textcolor{black}{This could be done considering the lag-$h$ sample covariance operator $T^{-1}\sum_{j=1}^TD^{(T)}_{\omega_j}\otimes D^{(T)}_{\omega_{j+h}}$ which should be centered at the zero operator in $S_2$ for all $h=1,\ldots,T-1$. 
Here, two statistics based on the coefficients in the Karhunen--Lo\`eve decomposition of the fDFTs are considered.  
For $j=1,\ldots,T$, let $(\phi_{l}^{\omega_j}\colon l\in\mathbb{N})$ be the orthonormal basis of eigenfunctions of $\mathcal{F}_{\omega_j}$ and observe that for this choice of basis $\Var(\inprod{D_{\omega_j}}{\phi^{\omega_j}_l}) = \inprod{\F_{\omega_j}(\phi^{\omega_j}_l)}{\phi^{\omega_j}_l} = \lambda^{\omega_j}_l$, where $(\lambda^{\omega_j}_l\colon l \in \nnum) \in \rnum_{+}$ are the eigenvalues of $\F_{\omega_j}$. Then, for any $j,j^\prime$, $(\phi_l^{\omega_j}\otimes \phi_{l^\prime}^{\omega_{j^\prime}}\colon l,l^\prime \in\mathbb{N})$ is an orthonormal basis of $L^2_{\cnum}([0,1]^2)$ and, by definition of the Hilbert--Schmidt inner product on the algebraic tensor product space $H\otimes H$, 
\begin{align*} 
\frac 1T\sum_{j=1}^TD_{\omega_j}^{(T)}\otimes D_{\omega_{j+h}}^{(T)}
&=\frac 1T\sum_{j=1}^T\sum_{l=1}^\infty\sum_{l^\prime=1}^\infty
\biginprod{ D_{\omega_j}^{(T)}\otimes D_{\omega_{j+h}}^{(T)}}{
\phi_l^{\omega_j}\otimes \phi_{l^\prime}^{\omega_{j+h}}}_{S} \,
\phi_l^{\omega_j}\otimes \phi_{l^\prime}^{\omega_{j+h}} \tageq\label{eq:CovDsh}\\
&\approx\frac 1T\sum_{j=1}^T\sum_{l=1}^{L}\sum_{l^\prime=1}^{L^\prime}
\langle D_{\omega_j}^{(T)},\phi_l^{\omega_j}\rangle
\overline{\langle D_{\omega_{j+h}}^{(T)}, \phi_{l^\prime}^{\omega_{j+h}}\rangle}
\phi_l^{\omega_j}\otimes \phi_{l^\prime}^{\omega_{j+h}}
\end{align*}
for sufficiently large $L$ and $L^\prime$. The foregoing motivates to set up tests based on the 
score products
\begin{align} 
\label{eq:testproj_egb}
\gamma^{(T)}_{j,h}(l,l^\prime)=
\langle D^{(T)}_{\omega_{j}}, \phi^{\omega_j}_{l} \rangle 
\overline{\langle D^{(T)}_{\omega_{j+h}}, \phi^{\omega_{j+h}}_{l^\prime}}\rangle
\end{align} 
or on the standardized score products 
\begin{align} 
\label{eq:testproj_egb_st}
\rho^{(T)}_{j,h}(l,l^\prime)=\frac{\gamma^{(T)}_{j,h}(l,l^\prime)}{\sqrt{ {\lambda^{\omega_j}_{l}}\lambda^{\omega_{j+h}}_{l^\prime}}}.
\end{align} 
In practice, the unknown spectral density operators $\mathcal{F}_{\omega_j}$ and $\mathcal{F}_{\omega_{j+h}}$ are to be replaced with consistent estimators $\Fop{j}$ and $\Fop{j+h}$, which will then yield 
 respective sample eigenvalues $\hat\lambda_l^{\omega_j}$ and eigenfunctions $\hat\phi_l^{\omega_j}$. 
The estimated quantities corresponding to \eqref{eq:testproj_egb} and \eqref{eq:testproj_egb_st} will be denoted by $\hat{\gamma}^{(T)}_{j,h}(l,l^\prime)$ and $\hat\rho^{(T)}_{j,h}(l,l^\prime)$, respectively.} 
As an estimator of ${\mathcal{F}}_{\omega}$, take
\begin{align} \label{eq:estimatestat}
\hat{\mathcal{F}}^{(T)}_{\omega}= \frac{2\pi}{T}\sum_{j=1}^{T}K_{b}(\omega-\omega_j){\big(D^{(T)}_{\omega_j} \otimes D^{(T)}_{\omega_j}\big)},
\end{align}
where $K_b(\cdot)$ is a kernel with bandwidth $b$ satisfying the following conditions. 
\begin{assumption} \label{windowfunction}
{\em (a) Let $K\colon[-\frac{1}{2},\frac{1}{2}] \to \mathbb{R}_+$ be 
symmetric 
with $\int K(x) dx=1$ and $\int K(x)^2 dx< \infty$. 

(b) Let $b=b_T$ be a bandwidth such that $T^{-1/2}\ll b_T\ll T^{-1/4}$. 

(c) Let $K_b(x)=b^{-1}K((2\pi b)^{-1}x)$ and 
and extend the kernel periodically such that $K_b(x) =K_b(x\pm 2\pi)$} in order to include estimates for frequencies around $\pm \pi$.
\end{assumption}
%
%

\textcolor{black}{To set up the test statistics, it now appears reasonable to extract information across a range of directions $l=1,\ldots, L_j$ and $l^\prime=1,\ldots,L_{j+h}$ as well as a selection of lags $h=1,\ldots,\bar{h}$, where $\bar{h}$ denotes an upper limit. The truncation parameters $L_j=L(\omega_j)$ and $L_{j+h}=L(\omega_{j+h})$ are explicitly allowed to depend on the $j$-th and $(j+h)$-th Fourier frequencies in order to accommodate heterogeneity in the Karhunen--Lo\`{e}ve decompositions across the spectral domain. Set 
\begin{equation}
\label{eq:weighted}
\hat{{\beta}}_{h,u}^{(T)}=
\frac{1}{T} \sum_{j=1}^T \sum_{l=1}^{L_j}\sum_{l^\prime=1}^{L_{j+h}} \hat\gamma_{j,h}^{(T)}(l,l^\prime)
\qquad\mbox{and}\qquad
\hat{{\beta}}_{h,s}^{(T)}=
\frac{1}{T} \sum_{j=1}^T \sum_{l=1}^{L_j}\sum_{l^\prime=1}^{L_{j+h}} \hat\rho_{j,h}^{(T)}(l,l^\prime),
\end{equation}
where the subscripts $u$ and $s$ refer to the un-standardized and standardized forms, respectively. In the following, the subscript $x$ will be used to refer to any of these two versions when no confusion can arise.}

\textcolor{black}{Choose next a collection $h_1,\ldots,h_M$ of lags each of which is upper bounded by $\bar{h}$ to pool information across a number of autocovariances and build the vectors
\[
\hat{\bm{b}}_{M,x}^{(T)}
=
\big(\Re\hat{{\beta}}_{h_1,x}^{(T)},\ldots,\Re\hat{{\beta}}_{h_M,x}^{(T)},
\Im\hat{{\beta}}_{h_1,x}^{(T)},\ldots,\Im\hat{{\beta}}_{h_M,x}^{(T)}\big)^\top,
\]
where $\Re$ and $\Im$ denote real and imaginary part, respectively. Finally, set up the quadratic forms 
\begin{equation}
\label{eq:quad_form_test}
\hat{Q}_{M,x}^{(T)}=T(\hat{\bm{b}}_{M,x}^{(T)})^\top\hat{\Sigma}_{M,x}^{-1}\hat{\bm{b}}_{M,x}^{(T)},
\end{equation}
where $\hat\Sigma_{M,x}$ is an estimator of the asymptotic covariance matrix of the vectors ${\bm{b}}_{M,x}^{(T)}$ which are defined by replacing $\hat\gamma_{j,h}^{(T)}(l,l^\prime)$ and $\hat\rho_{j,h}^{(T)}(l,l^\prime)$ with $\gamma_{j,h}^{(T)}(l,l^\prime)$ and $\rho_{j,h}^{(T)}(l,l^\prime)$ in \eqref{eq:weighted} and then using the resulting $\beta_{h,x}^{(T)}$ in place of $\hat\beta_{h,x}^{(T)}$ in the definition of $\hat{\bm{b}}_{M,x}^{(T)}$. The foregoing provides the two test statistics $\hat{Q}_{M,u}^{(T)}$ and $\hat{Q}_{M,s}^{(T)}$ that will be used to test the null of stationarity against the alternative of local stationarity. Note that both quadratic forms depend on the tuning parameters $L_j$, $L_{j+h}$ and $M$, the selection of which will be evaluated empirically in Section \ref{sec:empirical}.} 

\textcolor{black}{To facilitate the derivation of large-sample results, the following assumptions are made: for the un-standardized respectively standardized test require 
\begin{itemize}\itemsep-.3ex
\item [] Condition $C_u$: Let $L_j\sim\log T$ and\/ $\lim_l\inf_\omega\lambda_l^{\omega}>0$;
\item [] Condition $C_s$: Let $\inf_{\omega}\lambda^{\omega}_{\bar{L}} >0$ for some $\bar{L}\geq\sup_jL_j$.
\end{itemize}
In keeping with the above arrangement, the respective conditions will be referred to as $C_x$ if no confusion arises. Condition $C_u$ for the un-standardized test allows to send the truncation levels $L_j$ to infinity in a coordinated manner as long as the divergence is slow (here, logarithmic) compared to $T$; see \cite{fhks14}. Condition $C_s$ for the standardized test on the other hand requires a finite truncation level, to ensure that the smallest eigenvalues of the compact operators $\mathcal{F}_{\omega_j}$ are bounded away from zero as these show up in the denominator of \eqref{eq:testproj_egb_st}.
}

\section{Large-sample results}
\label{sec:large}


\subsection{Assumptions}
\label{sec:large:assumptions}

\textcolor{black}{The following gives the main requirements under both stationarity and local stationarity in terms of cumulant tensors of the functional time series (Appendix \ref{subsec:proof:cumulants}) that are needed to establish the asymptotic behavior of the test statistics under both hypotheses. Note that the null hypothesis is nested within the alternative. Because of this basic fact, we start with the general assumptions under local stationarity before specializing to the stationary case. 
\textcolor{black}
{\begin{assump2}{2}{$k$}{\,$\ell$} \label{cumglsp} 
Assume $(\XT{t}\colon t\leq T, T\in\mathbb{N})$ and $(\Xu{t}\colon t\in\mathbb{Z})$ are as in Definition \ref{lsp}. Suppose $\sup_{t}\E[\|X_t\|^{\min(k,12)}_2] < \infty$ and that there exists a a positive sequence $\cumbp{k}{k}$ in $L^2_\mathbb{R}([0,1]^k)$, independent of\/ $T$ such that, for all $j=1,\ldots,k-1$ and some $\ell\in\mathbb{N}$,
\begin{align} 
\label{eq:kapmix}
\sum_{t_1,\ldots,t_{k-1} \in \mathbb{Z}} (1+|t_j|^\ell)\|\cumbp{k}{k}\|_2 <\infty.
\end{align}
Suppose furthermore that there exist representations
\begin{align} \label{eq:repstatap}
\XT{t}- \Xuu{{t}/{T}}{}= Y^{(T)}_{t} 
\qquad\mbox{and}\qquad  
\Xu{t}- X^{(v)}_{t}=(u-v) Y_{t}^{(u,v)}, 
\end{align}
for some processes $(Y^{(T)}_{t}\colon t\leq T, T\in\mathbb{N})$ and $(Y_{t}^{(u,v)}\colon t\in\mathbb{Z})$ taking values in $H_\rnum$ whose $k$-th order joint cumulants satisfy
\begin{enumerate}\itemsep-.3ex
\item[(i)] $\|\cum(\XT{t_1},\ldots,\XT{t_{k-1}},Y^{(T)}_{t_k}) \|_2 \le \frac{1}{T}\|\kappa_{k;t_1-t_k,\ldots,t_{k-1}-t_k}\|_2 $,
\item[(ii)] $\|\cum(\Xuu{u_1}{1},\ldots,\Xuu{u_{k-1}}{k-1},Y_{t_{k}}^{(u_k,v)}) \|_2 \le \|\kappa_{k;t_1-t_k,\ldots,t_{k-1}-t_k}\|_2 $,
\item[(iii)] $\sup_u \|\cum(\Xuu{u}{1},\ldots,\Xuu{u}{k-1},\Xuu{u}{k}) \|_2 \le \|\kappa_{k;t_1-t_k,\ldots,t_{k-1}-t_k}\|_2$,
\item[(iv)] $\sup_u \|\frac{\partial^\ell}{\partial u^\ell} \cum(\Xuu{u}{1},\ldots,\Xuu{u}{k-1},\Xuu{u}{k}) \|_2 \le \|\kappa_{k;t_1-t_k,\ldots,t_{k-1}-t_k}\|_2.$
\end{enumerate}
\end{assump2}}
\textcolor{black}{
\ref{cumglsp} provides Lipschitz conditions that are generalizations of those in \citet{lsr16}, who investigated the properties of quadratic forms of stochastic processes in a finite-dimensional setting. The above conditions enable to express the behavior of the fDFT's of a $k$-th order locally stationary process in terms of $k$-th order time-varying spectral density tensors (Lemma \ref{cumboundglsp}). This is convenient in order to derive explicit expressions of the distributional properties under the alternative and to understand departures from stationarity. 
Under $H_A$, we can uniquely characterize the second-order stucture of the stochastic process $(\XT{t}\colon t\leq T,T \in \nnum)$ via the {\em time-varying spectral density operator}
\begin{align} 
\mathcal{F}_{u,\omega} = \frac{1}{2 \pi}\sum_{h \in \mathbb{Z}} \mathcal{C}_{u,h}e^{-\im \omega h},
\end{align}
where $ \mathcal{C}_{u,h}=\cum(X^{(u)}_h, X^{(u)}_0)$ denotes the local cumulant tensor at fixed time $u$ of the stationary approximating process $(X^{(u)}_t\colon t\in\znum)$. Note that the parameter $\ell$ and (iii)-(iv) in \ref{cumglsp}, influence the smoothness of the operator-valued mapping $(u,\omega) \mapsto \F_{u,\omega}$.
Under \ref{cumglsp}(2,2), derivative maps are well-defined elements of $S_2(H)$ and $\omega \mapsto \F_{u;\omega}$ is uniformly continuous in $\omega$ with respect to $\snorm{\cdot}_2$. We refer to Lemma \ref{tvsdo} for details.
More generally, under $k$-th order local stationarity, these properties carry over to the {\em local $k$-th order cumulant spectral density tensor} 
\begin{align}\label{eq:tvsdoker} 
\floc{u}{k} = \frac{1}{(2 \pi)^{k-1}}\sum_{t_1,\ldots,t_{k-1} \in \mathbb{Z}}  \cumloc{u}{k}e^{-\im \sum_{j=1}^{k-1} \omega_{j} t_j},
\end{align}
where $\omega_1,\ldots, \omega_{k-1} \in (-\pi, \pi]$ and $
\cumloc{u}{k}=\cum\big(\Xuu{u}{1},\ldots,\Xuu{u}{k-1},\Xuu{u}{0}\big)$
is the corresponding local cumulant kernel tensor of order $k$ at time $u_0$. Observe that, for $k>1$, \eqref{eq:tvsdoker} can be viewed as an element of $S_2(H^{\otimes^{\lfloor (k+1)/2 \rfloor}}, H^{\otimes^{\lfloor k/2 \rfloor}})$. Under $k$-th order stationarity the above objects become independent of local time $u$, so that $\floc{u}{k}\equiv \F_{\omega_1,\ldots,\omega_{k-1}}$, and \ref{cumglsp} specializes to the following.}
\textcolor{black}{\begin{assump3}{2}{$k$}{$\,\ell$}
\label{Statcase}Let $(X_t\colon t\in\mathbb{Z})$ be a $k$-th order stationary functional time series with values in $H_\rnum$ such that
(i) $\E[\|X_0\|^{\min(k,12)}_2] < \infty$ and 
(ii) $\sum_{t_1,\ldots, t_{k-1}= -\infty}^{\infty} (1+|t_j|^{\ell}) \|
\mathcal{C}_{t_1,\ldots,t_{k-1}}\|_2 < \infty$ for all $1 \le j \le k-1$.
\end{assump3}}
Because the test statistics require estimators of the eigenelements of $\F_{\omega}$, it is of importance to consider the properties of the estimator \eqref{eq:estimatestat} for both null and alternative hypotheses. The next theorem shows that it is a consistent estimator of the integrated (in a Bochner sense) time-varying spectral density operator 
\[
{G}_{\omega} = \int_{0}^1 \mathcal{F}_{u,\omega }du,
\]
where the convergence is uniform in $\omega \in [-\pi,\pi]$ with respect to $\snorm{\cdot}_2$. This therefore becomes an operator-valued function in $\omega$ that acts on $H$ and is independent of rescaled time $u$. Under $H_0$, $G_\omega$ thus reduces to $\mathcal{F}_\omega$. }
\begin{theorem}[{\textbf{Consistency and uniform convergence}}]
\label{UnifCon}
\textcolor{black}{Suppose $(\XT{t}\colon t \leq T, T \in \nnum)$ satisfies \ref{cumglsp}$(4,2)$. Consider the estimator $\hat{\mathcal{F}}^{(T)}_{\omega}$ in \eqref{eq:estimatestat} with smoothing kernel $K$ fulfilling Assumption \ref{windowfunction}(a) and (c). Then, 
\begin{enumerate}\itemsep-.3ex
\item[(a)] $\E[\snorm{\hat{\mathcal{F}}^{(T)}_{\omega}- {G}_{\omega}}^2_2] = O((bT)^{-1}+b^4)$, uniformly in $\omega \in [-\pi,\pi]$.
\item[(b)] If, in addition, Assumption \ref{windowfunction}(b) holds and $K$ has bounded derivative on $(-1/2,1/2)$ then,\\
 $\sup_{\omega \in [-\pi,\pi] }\snorm{\hat{\mathcal{F}}^{(T)}_{\omega}- {G}_{\omega}}_{2} \overset{p}{\to} 0. \label{eq:unif}$
\end{enumerate}}
\end{theorem}

The proof of Theorem \ref{UnifCon} is given in Section \ref{subsec:proof:alt} of the Appendix. Since the theorem shows consistency of $\hat{\F}_{\omega}$, a self-adjoint element of $S_2(H)$, it follows from \cite{mm03} that the sample eigenelements $(\hat{\lambda}^{\omega}_l, \hat{\phi}^{\omega}_l \colon l  \in \mathbb{N})$ of $\hat{\F}_{\omega}$ provide consistent estimators for the eigenelements $(\tilde{\lambda}^{\omega}_l, \tilde{\phi}^{\omega}_l\colon l  \in \mathbb{N})$ of $G_{\omega}$. If $H_0$ is satisfied, then the stated consistency holds for the eigenelements $(\lambda^{\omega}_l,\phi^{\omega}_l\colon l\in\mathbb{N})$ of $\F_{\omega}$.


\subsection{Properties under the null of stationarity}
\label{sec:large:null}

The asymptic results under $H_0$ are collected in this section. 
The first theorem establishes that the scaled difference between $\beta_{h,x}^{(T)}$ and $\hat{\beta}_{h,x}^{(T)}$ is negligible in large samples. \textcolor{black}{Note that the assumptions here and for other theorems in this section are formulated imposing stationarity on certain moments for the null hypothesis via \ref{Statcase}. To verify the results, typically further assumptions on higher-order cumulants are required. \textcolor{black}{These are controlled via  \ref{cumglsp}.}}
\begin{theorem}
\label{boundstatcase}
\textcolor{black}{Let Assumption \ref{windowfunction},  \ref{cumglsp}(12,2)  and $C_x$ hold. Then, under $H_0$, for any fixed $h$, }
\begin{align*}
{\color{black}{\sqrt{T}\big| \hat{\beta}_{h,x}^{(T)}-{\beta}^{(T)}_{h,x} \big| 
= O_p\bigg(\frac{1}{{bT}}+b^2\bigg)
\qquad (T\to\infty).
}}
\end{align*}
\end{theorem}
The proof is given in Section \ref{sec:bounderror} of the Appendix. In view of Assumption \ref{windowfunction}, Theorem \ref{boundstatcase} shows that the distributional properties of $\hat{\beta}_{h,x}^{(T)}$ are asymptotically the same as those of $\beta^{(T)}_{h,x}$. Note that these rates are necessary for the estimator in \eqref{eq:estimatestat} to be consistent, as is seen from part (a) of Theorem \ref{UnifCon}, which reduces to the stationary case if the process does not depend on $u$. They hence do not impose an additional constraint under $H_0$. 

\textcolor{black}{
The next theorem derives that, under the additional assumption of fourth-order stationarity, the asymptotic variance is uncorrelated for all lags $h$ and that there is no correlation between the real and imaginary parts. For $n\in\mathbb{N}$, set $[n]=\{1,\ldots, n\}$.
\begin{theorem}
\label{momentstatcase}
Let Assumption \ref{windowfunction} and $C_x$ hold. Suppose further that \ref{Statcase}(4,2) is satisfied.
Then, for $h_1=h_2=h$,
\begin{align*} 
(a)\quad &T\,\cov\Big( \Re \hat{\beta}_{h,u}^{(T)}, \Re\hat{\beta}_{h,u}^{(T)}\Big)  
=T\, \cov\Big(\Im \hat{\beta}_{h,u}^{(T)}, \Im\hat{\beta}_{h,u}^{(T)}\Big)   \\[.1cm]
&\to
\frac{1}{4 \pi}\int \int \sum_{(\bm{l},\bm{l}^\prime)\in{\mathcal{L}}\times{\mathcal{L}}^\prime}
\langle \F_{\omega,-\omega-\omega_{h},-\omega^{\prime}} (\phi^{\omega^{\prime}}_{l_{1}^\prime} \otimes \phi^{\omega^{\prime}+\omega^{\prime}_{h}}_{l_2^\prime}), \phi^{\omega}_{l_{1}} \otimes \phi^{\omega+\omega_h}_{l_2}\rangle d\omega d\omega^{\prime}
+\frac{1}{2\pi}\int\sum_{\bm{l}\in\mathcal{L}}\lambda_{l_1}^\omega\lambda_{l_2}^{\omega+\omega_h}d\omega,~~~~~~~~~~\phantom{x}\\[.2cm]
(b)\quad &T\,\cov\Big( \Re \hat{\beta}_{h,s}^{(T)}, \Re\hat{\beta}_{h,s}^{(T)}\Big)  
=T\, \cov\Big(\Im \hat{\beta}_{h,s}^{(T)}, \Im\hat{\beta}_{h,s}^{(T)}\Big)   \\[.1cm]
&\to
\frac{1}{4 \pi}\int \int \sum_{(\bm{l},\bm{l}^\prime)\in{\mathcal{L}}\times{\mathcal{L}}^\prime}
\frac{ \langle \F_{\omega,-\omega-\omega_{h},-\omega^{\prime}} (\phi^{\omega^{\prime}}_{l_{1}^\prime} \otimes \phi^{\omega^{\prime}+\omega^{\prime}_{h}}_{l_2^\prime}), \phi^{\omega}_{l_{1}} \otimes \phi^{\omega+\omega_h}_{l_2}\rangle }{{\sqrt{\lambda^{\omega}_{l_1}\lambda^{{\omega+\omega_{h}}}_{l_2}\lambda^{\omega^{\prime}}_{l_1^\prime} \lambda_{l_2^\prime}^{\omega^{\prime}+\omega^{\prime}_{h}}}}}d\omega d\omega^{\prime}
+\frac{1}{2\pi}\int\sum_{\bm{l}\in\mathcal{L}}\delta_{l_1,l_2}d\omega,
\end{align*} 
where $\bm{l}=(l_1,l_2)$, $\bm{l}^\prime=(l_1^\prime,l_2^\prime)$, $\mathcal{L}=[L(\omega)]\times [L(\omega+\omega_h)]$, $\mathcal{L}^\prime=[L(\omega^\prime)]\times [L(\omega^\prime+\omega^\prime_h)]$, and $\delta_{i,j} = 1$ if $i=j$ and $0$ otherwise. If $h_1\neq h_2$, $T\,\cov( \Re \hat{\beta}_{h_1,x}^{(T)}, \Re\hat{\beta}_{h_2,x}^{(T)})\to 0$,
$T\, \cov(\Im \hat{\beta}_{h_1,x}^{(T)}, \Im\hat{\beta}_{h_2,x}^{(T)})\to 0$ and $T\,\cov( \Re \hat{\beta}_{h_1,x}^{(T)}, \Im\hat{\beta}_{h_2,x}^{(T)})\to 0$.
\end{theorem}
}

The proof of Theorem \ref{momentstatcase} is given in Appendix \ref{subsec:cov-struc}. 
Observe that the results in part (b) imply that the standardized test statistics is pivotal if the data is Gaussian. 
Note also that the results in the theorem use at various instances the fact that the $k$-th order spectral density operator at frequency $\bm{\omega}=(\omega_1,\ldots,\omega_k)^T\in\mathbb{R}^k$ is equal to the $k$-th order spectral density operator at frequency $-\bm{\omega}$ in the manifold $\sum_{j=1}^k\omega_j\!\mod 2\pi$. 

With the previous results in place, the large-sample behavior of the quadratic form statistics $\hat{Q}_{M,x}^{(T)}$ defined in \eqref{eq:quad_form_test} can be derived. This is done in the following theorem.

\begin{theorem}
\label{th:test_null}
\textcolor{black}{Let Assumption \ref{windowfunction} and $C_x$ hold. Suppose further that \ref{cumglsp}($k$,\,2) is satisfied for all $k\ge3$. Then, under $H_0$,}
\vspace{-.2cm}
\begin{enumerate}
\itemsep-.3ex
\item[(a)] For any collection $h_1,\ldots,h_M$ bounded by $\bar{h}$,
\[
\sqrt{T}\hat{\bm{b}}^{(T)}_{M,x}\stackrel{\cal D}{\to}\mathcal{N}_{2M}(\bm{0},\Sigma_{0,x})
\qquad(T\to\infty),
\]
where $\stackrel{\mathcal{D}}{\to}$ denotes convergence in distribution. Under the additional assumption of fourth-order stationarity, $\mathcal{N}_{2M}(\bm{0},\Sigma_{0,x})$ is a $2M$-dimensional normal distribution with mean $\bm{0}$ and diagonal covariance matrix $\Sigma_{0,x}=\mathrm{diag}(\sigma^2_{0,m,x}\colon m=1,\ldots,2M)$ whose elements are
\[
\sigma_{0,m,x}^2
=\lim_{T\to\infty} 
T\cov\big(\Re\hat{\beta}_{h_m,x},\Re\hat{\beta}_{h_m,x}\big),
\qquad m=1,\ldots, M,
\]
and $\sigma_{0,M+m,x}^2=\sigma_{0,m,x}^2$. The explicit form of the limit is determined by Theorem \ref{momentstatcase}. \textcolor{black}{If fourth-order stationarity is violated, then the limiting normal distribution has a non-diagonal covariance structure.} 
\item[(b)] Using the result in (a), it follows that for the statistic defined in \eqref{eq:quad_form_test}
\[
\hat{Q}_{M,x}^{(T)}\stackrel{\mathcal{D}}{\to}\chi_{2M}^2
\qquad(T\to\infty),
\]
where $\chi_{2M}^2$ is a $\chi^2$-distributed random variable with $2M$ degrees of freedom. 
\end{enumerate}
\end{theorem}

The proof of Theorem \ref{th:test_null} is provided in Appendix \ref{sec:proofs:weak_conv}. Part (b) of the theorem can now be used to construct tests with asymptotic level $\alpha$. 
Note that the application of the test requires an estimator of $\hat\Sigma_{M,x}$. This will be discussed in Section \ref{subsec:empirical:fourth}. 

\textcolor{black}{
To explicitly compute the limiting covariance structure in part (a) of Theorem \ref{th:test_null} under second-order stationarity but fourth-order nonstationarity, \textcolor{black}{the source of nonstationarity needs to be specified. For example, the results put forward in the next two sections 
allow for the computation of $\Sigma_{0,x}$ if the process is fourth-order locally stationary.}
Then, in the covariance structure of the covariance operator of the fDFT's, the fourth-order cumulant tensor component will, for $h_1\neq h_2$, (quadratically) decay in norm as the distance $|h_1 - h_2|$ increases (see Lemma~\ref{cumboundglsp}, Corollary~\ref{cumbounds}\,(ii) and equation \eqref{eq:covHa}). As a consequence of this term being present in the covariance structure, the real and imaginary part of the
projections are no longer uncorrelated but the correlation decays with increasing distance $|h_1 -h_2|$. In this scenario, a small loss of power is to be expected when the test statistic is built under the assumption of a diagonal covariance structure.
}


\subsection{Properties under the alternative}
\label{sec:large:alternative}

This section contains a generalization of the results in Section \ref{sec:large:null} to locally stationary functional time series. 
The following theorem is the counterpart to Theorem \ref{boundstatcase} under the null hypothesis.

\begin{theorem}
\label{boundlocstatcase}
\textcolor{black}{Let Assumption \ref{windowfunction}, \ref{cumglsp}(12,2) and $C_x$ hold. Then, under $H_A$,
\begin{align*}
\sqrt{T}\mathbb{E}\Big[\big|\hat{\beta}_{h,x}^{(T)}-\beta_{h,x}^{(T)}-
\mathcal{B}_{h,x}^{(T)}\big|\Big]
= O\bigg(\frac{1}{bT}+b^2+\frac{1}{b\sqrt{T}}+b^2\sqrt{T}\bigg)
\qquad (T\to\infty),
\end{align*}
where 
\[
\mathcal{B}_{h,x}^{(T)}
=\frac 1T\sum_{j=1}^T\sum_{\bm{l}\in{\mathcal{L}}}\zeta_{\bm{l},x}
\big\langle\mathbb{E}[ D_{\omega_j}\otimes D_{\omega_{j+h}}\big],
\mathbb{E}[\hat\phi_l^{\omega_j}\otimes\hat\phi_{l^\prime}^{\omega_j+h}]
-\hat\phi_l^{\omega_j}\otimes\hat\phi_{l^\prime}^{\omega_j+h}\big\rangle_{S}
\]
is a stochastic bias term satisfying $\sqrt{T}\mathcal{B}_{h,x}^{(T)}=O_P(1)$, and $\zeta_{\bm{l},u}=1$ and $\zeta_{\bm{l},s}=(\tilde\lambda_l^\omega,\tilde\lambda_{l^\prime}^{\omega+\omega_h})^{-1/2}$.}
\end{theorem}
The proof of Theorem \ref{boundlocstatcase} is given in Section \ref{sec:bounderror} of the Appendix.
\textcolor{black}{In view of Assumption \ref{windowfunction}, the theorem shows that $\hat\beta_{h,x}^{(T)}$ has the same asymptotic sampling properties as $\beta_{h,x}^{(T)}$ up to a stochastically bounded bias term (after scaling with $\sqrt{T}$). Note that $|\hat{\beta}_{h,x}^{(T)}-\mathbb{E}[\beta_{h,x}^{(T)}]|\stackrel{P}{\to}0$, where 
\begin{equation}
\mathbb{E}[\beta_{h,x}^{(T)}]
\to\frac{1}{2\pi}\int_0^{2\pi}\int_0^1\sum_{\bm{l}\in\mathcal{L}}
\zeta_{\bm{l},x}\langle\mathcal{F}_{u;\omega}e^{-\i 2\pi uh},\tilde{\phi}_l^\omega\otimes\tilde{\phi}_{l^\prime}^{\omega+\omega_h}\rangle_{S}dud\omega
=\mu_{h,x}
\label{eq:muhx}
\end{equation}
is an noncentrality parameter (see Appendix \ref{sec:proof:first}) that will have to enter the limit distribution of $\hat{Q}_{M,x}^{(T)}$ as a consequence of the violation of weak stationarity. We discuss this term in some more detail below.}

A precise formulation of the asymptotic properties under $H_A$ is given in the next theorem.
\begin{theorem}
\label{th:test_alt}
Let Assumption \ref{windowfunction} and $C_x$ hold. Suppose further that \ref{cumglsp}($k$,\,2) is satisfied for all $k\ge2$. Then, under $H_A$, 
\vspace{-.2cm}
\begin{enumerate}
\itemsep-.3ex
\item[(a)] For any collection $h_1,\ldots,h_M$ bounded by $\bar{h}$, 
\[
\sqrt{T}\hat{\bm{b}}_{M,x}^{(T)}\stackrel{\mathcal{D}}{\to}\mathcal{N}_{2M}(\bm{\mu}_x,\Sigma_{A,x})
\qquad(T\to\infty),
\]
where $\mathcal{N}_{2M}(\bm{\mu}_x,\Sigma_{A,x})$ denotes a $2M$-dimensional normal distribution with mean vector $\bm{\mu}_x$ 
whose first $M$ components are $\Re{\mu}_{h_m,x}$ and last $M$ components are $\Im{\mu}_{h_m,x}$, where $\mu_{h_m,x}$ is defined through \eqref{eq:muhx}, and non-diagonal block covariance matrix
\[
\Sigma_{A,x}=
\begin{pmatrix}
\displaystyle\Sigma_{A,x}^{(11)}& \Sigma_{A,x}^{(12)} \\[.2cm]
\displaystyle\Sigma_{A,x}^{(21)}& \Sigma_{A,x}^{(22)}
\end{pmatrix}
\]
whose $M\times M$ blocks are determined by the results in Appendix \ref{sec:HAdist} and Section \ref{sec:covCLT} of the Online Supplement. 
\item[(b)] Using the result in (a), it follows that for the statistic defined in \eqref{eq:quad_form_test}
\[
\hat{Q}_{M,x}^{(T)}\stackrel{\mathcal{D}}{\to}\chi^2_{\mu_x,2M},
\qquad(T\to\infty),
\]
where $\chi^2_{\mu_x,2M}$ denotes a generalized noncentral $\chi^2$-distributed random variable with noncentrality parameter $\mu_x=\|\bm{\mu}_x\|_2^2$ and $2M$ degrees of freedom.
\end{enumerate}
\end{theorem}

The proof of Theorem \ref{th:test_alt} can be found in Appendix \ref{sec:HAdist}. \textcolor{black}{Observe that the limiting noncentrality parameter $\mu_x$ of the statistic $\hat{Q}_{M,x}^{(T)}$ measures the aggregation of the functions in \eqref{eq:muhx}. Under $H_A$, the operator in \eqref{eq:CovDsh} no longer converges in norm to the zero operator but instead to the operator $\frac{1}{2\pi} \int_{0}^{2\pi} \int_0^1 \F_{u,\omega}e^{-\im 2\pi uh}du d\omega$. The properties of the latter, which are  extracted to finite dimension via $\mu_{h,x}$, carry some meaningful information on the behavior of the test under the alternative.}
Firstly, denote a general term in the limiting expansion of $\mu_{h,x}$ by
\[
\mu_{h,x}(\bm{l})=
 \frac{1}{2\pi}\int_0^{2\pi} \int_{0}^{1} \zeta_{\bm{l},x}\langle\mathcal{F}_{u;\omega}e^{-\im 2\pi uh},\tilde{\phi}_l^\omega\otimes\tilde{\phi}_{l^\prime}^{\omega+\omega_h}\rangle_{S}dud\omega.
 \]
 For fixed directions $\bm{l}=(l,\lpr)$, this function can be seen to approximate the $(h,0)$-th Fourier coefficients of the function $(u,\omega) \mapsto \zeta_{\bm{l},x}\inprod{\F_{u,\omega} (\tilde{\phi}^{\omega+\omega_{h}}_\lpr)}{  \tilde{\phi}^{\omega}_l  }$ , i.e., for small $h$ and $T \to \infty$ they approximate
\begin{align*}
\vartheta_{h,j,x}(\bm{l})
&=\frac{1}{2\pi}\int_0^{2\pi}\int_0^1
\zeta_{\bm{l},x}\inprod{\F_{u,\omega} \tilde{\phi}^{\omega+\omega_{h}}_\lpr}{  \tilde{\phi}^{\omega}_l  }
e^{\im 2\pi uh-\im j\omega}dud\omega
\end{align*}
with $j=0$. In other words, $\mu_{h,x}(\bm{l})\approx\vartheta_{h,0,x}(\bm{l})$. If the process is weakly stationary then 
the integrand of the coefficient does not depend on $u$ and all Fourier coefficients are zero except $\vartheta_{0,j,x}(\bm{l})$. In particular, $\vartheta_{0,0,s}(\bm{l})=1$. Following \cite{p09} and \cite{dsr16}, the mean functions can thus be seen to reveal long-term non-stationary behavior. Unlike testing methods based on segments in the time domain, the proposed method is therefore able to detect smoothly changing behavior in the temporal dependence structure.

\textcolor{black}{
Secondly, the operator $\int_0^1 \F_{u,\omega}e^{-\im 2\pi uh}du$ can be viewed as the $h$-th Fourier coefficient of the operator-valued function $(u)\mapsto \F_{u,\omega}$ for fixed $\omega$ (Lemma \ref{cumboundglsp}), which exhibits a quadratic decay in norm as a function of $h$ such that the sum of the norms of these coefficients is finite (Corollary \ref{cumbounds}). Since this behavior carries over to the projections, the contribution to $\mu_x$ of the functions $\mu_{h,x}$ in \eqref{eq:muhx} for larger values of $h$ will become negligible. Intuitively, utilizing large values of $M$ in the statistic $\hat{Q}_{M,x}^{(T)}$ is hence expected to increase the likelihood of a type II error; see also Section \ref{sec:empirical}. }

\textcolor{black}{
 The results in this and the previous section require an understanding of the estimator $\hat\Sigma_{M,x}$ used in the definition of the test statistics $\hat Q_{M,x}^{(T)}$ in \eqref{eq:quad_form_test}. The corresponding results are part of the next subsection.}

\subsection{Estimating the fourth-order spectrum}
\label{subsec:empirical:fourth}

\textcolor{black}{The estimation of the matrix $\Sigma_M$ is a necessary ingredient in the application of the proposed stationarity test. Generally, the estimation of the sample (co)variance can influence the power of tests, as has been observed in a number of previous works set in similar albeit nonfunctional contexts. Among the contributions more closely related to this paper are \citet{p09} 
who used the spectral density of the squares, 
\citet{dsr16}, who focused on Gaussianity of the observations, and \citet{jsr15}, who employed a stationary bootstrap procedure. A different idea was put forward by \citet{bsr17} and \citet{bjsr17}. These authors utilized the notion of orthogonal samples to estimate the variance, falling back on a general estimation strategy developed in \citet{sr16}.}

\textcolor{black}{In order to utilize the results of Theorem \ref{th:test_null}, we require an estimator of the tri-spectral density operator $\F_{\omega,-\omega-\omega_{h},-\omega^\prime}$, \textcolor{black}{which can then subsequently be projected onto the (standardized) empirical eigenfunctions and integrated over $\omega, \omega^\prime$.} As an estimator, consider 
\begin{align}\label{eq:fourthorderest}
\hat{\F}_{\omega_{j_1},\ldots, \omega_{j_4}}= \frac{(2\pi)^{3}}{{(b_4 T)}^{3}}\sum_{k_1,k_2, k_3}  K_4\Big(\frac{\omega_{j_1}-{\omega_{k_1}}}{b_4},\ldots,\frac{\omega_{j_4}-{\omega}_{k_4}}{b_4}\Big) \Phi({\omega_{k_1}},\ldots,{\omega}_{k_4})  I^{(T)}_{\omega_{k_1},\ldots, \omega_{k_4}},
\end{align}
where  
\begin{align*}
I^{(T)}_{\omega_{k_1},\omega_{k_2},\omega_{k_3},\omega_{k_4}}
= \frac{T}{2\pi} D_{\omega_{k_1}} \otimes D_{\omega_{k_2}} \otimes D_{\omega_{k_3}}\otimes D_{\omega_{k_4}}
\end{align*}
denotes the tri-periodogram tensor and where $K_4(x_1,\ldots,x_4)$ is a smoothing kernel with compact support on $\rnum^4$ and where 
\[\Phi(\alpha_1,\alpha_2,\alpha_3,\alpha_4) = 1\]
 if $\sum_{k=1}^{4} \alpha_{k} \equiv 0 \mod 2\pi$ such that $\sum_{k \in {J}} \alpha_{k} \not \equiv 0 \mod 2\pi$ where $J$ is any non-empty subset of $\{1,2,3,4\}$ and equals $0$ otherwise. This function therefore controls that we are only working with those combinations of frequencies that lie on the principal manifold but do not lie in any proper submanifold. The reason for this is that, for $k >2$, the expectation of $k$-th order periodogram tensors evaluated at such submanifolds possibly diverges \citep[see also][for the Euclidean case]{br67}. As the next theorem shows, the estimator in \eqref{eq:fourthorderest} can be shown to be consistent if the bandwidth $b_4$ satisfies $b_4\to 0$ but $b_4^{-3}T\to\infty$ as $T\to\infty$.
\begin{theorem}\label{th:4th_null}
\textcolor{black}{Suppose \ref{Statcase}(4,2) and \ref{cumglsp}(8,2) hold.}
Then the estimator \eqref{eq:fourthorderest} of the tri-spectral density operator satisfies
\begin{align}
\E \bigg[\biggsnorm{ \frac{(2\pi)^2}{T^2}\sum_{j_1,j_2=1}^{T}\hat{\F}_{\omega_{j_1},-\omega_{j_1+h},-\omega_{j_2}} - \int \int  {\F}_{\omega,-\omega +\omega_{h},-\omega^{\prime}}d\omega d\omega^{\prime}}^2_2 \bigg]
=O\bigg(\frac{1}{b_4^3T} +b_4^4\bigg).  \label{eq:intF4}
\end{align}
\end{theorem}
The section is rounded out with large-sample behavior under the alternative. 
\begin{theorem}\label{thm:intF4loc}
\textcolor{black}{Suppose \ref{cumglsp}(8,2) holds}. Then,
\begin{align*}
(a)\quad& \biggsnorm{\E \bigg[\frac{(2\pi)^2}{T^2}\sum_{j_1,j_2=1}^{T}\hat{\F}_{\omega_{j_1},-\omega_{j_1+h},-\omega_{j_2}}\bigg] - \int \int  {G}_{\omega,-\omega +\omega_{h},-\omega^{\prime}}d\omega d\omega^{\prime}-\mathcal{Z}_h}_2 
=O\bigg(\frac{1}{b_4 T}+b_4\bigg), 
 \\
 (b)\quad& \bigsnorm{\Cov( \hat{\F}_{\omega_{j_1},\omega_{j_2},\omega_{j_3}},\hat{\F}_{\omega_{j_1},\omega_{j_2},\omega_{j_3}}) }^2_2
 =O\bigg(\frac{1}{b_4^3 T}\bigg),
\end{align*}
where $G_{\omega,-\omega +\omega_{h},-\omega^{\prime}}$ denotes the time-integrated tri-spectral operator and where $\mathcal{Z}_h \in S_2(H \otimes H)$ \textcolor{black}{is a bias term of order $O(\snorm{\mathcal{Z}_h}_2)=1$}.
\end{theorem}
The proofs of Theorems \ref{th:4th_null} and \ref{thm:intF4loc} are given in Section \ref{sec:F4proof} of the Online Supplement. Using continuity of the inner-product, Theorem \ref{UnifCon}(a) and the continuous mapping theorem imply  projecting onto the empirical eigenfunctions will not affect the rates.}

\section{Empirical results}
\label{sec:empirical}

This section reports the results of an illustrative simulation study designed to verify that the large-sample theory is useful for applications to finite samples. The test is subsequently applied to annual temperature curves data. The findings provide guidelines for a further fine-tuning of the test procedures to be investigated in future research.

\subsection{Simulation setting}
\label{sec:empirical:setting}

To generate functional time series, the general strategy applied, for example, in \citet{adh} and \citet{hkh15}, is utilized. For this simulation study, all processes are built on a Fourier basis representation on the unit interval $[0,1]$ with basis functions $\psi_1,\ldots,\psi_{15}$. Note that the $l$th Fourier coefficient of a $p$th-order functional autoregressive, FAR($p$), process $(X_t\colon t\in\mathbb{Z})$ satisfies
\begin{align}
\langle X_t,\psi_l\rangle
&=\sum_{l^\prime=1}^\infty\sum_{t^\prime=1}^p\langle X_{t-t^\prime},\psi_l\rangle
\langle A_{t^\prime}(\psi_l), \psi_{l'}\rangle + \langle \varepsilon_t, \psi_l  \rangle  \nonumber\\
& \approx 
\sum_{l' = 1}^{L_{\rm max}}\sum_{t^\prime=1}^{p} \langle X_{t-t^\prime}, \psi_l \rangle \langle A_{t^\prime}(\psi_l), \psi_{l'}\rangle + \langle \varepsilon_t, \psi_l  \rangle, \label{eq:farp}
\end{align}
the quality of the approximation depending on the choice of $L_{\rm max}$. The vector of the first $L_{\rm max}$ Fourier coefficients $\mathbf{X}_{t}=(\langle X_t, \psi_1 \rangle ,\ldots, \langle X_t, \psi_{L_{\rm max}} \rangle )^{\top}$ can thus be generated using the $p$th-order vector autoregressive, VAR($p$), equations
\begin{align*}
\mathbf{X}_{t} = \sum_{t^\prime=1}^{p} \mathbf{A}_{t^\prime} \mathbf{X}_{t-t^\prime} + \boldsymbol{\varepsilon}_t,
\end{align*}
where the $(l,l^\prime)$ element of $\mathbf{A}_{t^\prime}$ is given by $\langle A_{t^\prime}(\psi_l), \psi_{l'}\rangle$ and $\boldsymbol{\varepsilon}_t=(\langle \varepsilon_t, \psi_1 \rangle ,\ldots, \langle \varepsilon_t, \psi_{L_{\rm max}} \rangle )^{\top}$. The entries of the matrices $\mathbf{A}_{t^\prime}$ are generated as $\mathcal{N}(0,\nu^{(t^\prime)}_{l,l'})$ random variables, with the specifications of $\nu_{l,l^\prime}$ given below. To ensure stationarity or the existence of a causal solution, the norms $\kappa_{t^\prime}$ of $\boldsymbol{A}_{t^\prime}$ are required to satisfy certain conditions,  for example, $\sum_{t^\prime=1}^{p}\snorm{\boldsymbol{A}_{t^\prime}}_{\infty} <1$, which might be of more complicated nature \citep[see][for the stationary and locally stationary case, respectively]{b00,vde16}.
The functional white noise, FWN, process is included in \eqref{eq:farp} setting $p=0$. All simulation experiments were implement in {\tt R} and any result reported in the remainder of this section is based on 1{,}000 simulation runs.

\subsection{Specification of tuning parameters}
\label{subsec:empirical:tuning}

\textcolor{black}{The test statistics in \eqref{eq:quad_form_test} depends on the tuning parameters $L_j=L(\omega_j)$, determining the dimension of the projection spaces, and $M$, the number of frequency lags to be included in the procedure. In the following, a criterion will be set up to choose $L_j$, while for $M$ only a limited number of values were entertained because the selection is less critical for the performance as long as it is not chosen too large. 
Figure \ref{fig:eig} shows that it can well be of interest in practice to choose $L_j$ in a frequency-dependent way, as the eigenvalue decay might vary significantly between different $\omega_j$. The left part of the figure shows the situation for a functional white noise sequence. The spectral density operators are constant operator-valued functions of frequency and consequently their spectral decompositions coincide, producing relatively straight lines in the sample eigenvalue plots. In this case, one would not necessarily have to resort to determining the various truncation levels $L_j$ individually. However, the right part of the figure shows a time series with a significant level of dependence, in fact DGP (b) introduced in Section \ref{subsec:empirical:null} below. The functional variation of this second-order autoregressive process receives drastically different contributions from different frequency bands, yielding large differences also in the spectral decompositions: sample eigenvalues plotted against frequency are far from constant. Note also how the plot of the top sample eigenvalue resembles the univariate spectral density of a scalar second-order autoregression with levels of dependence determined by the operator norm $\snorm{\cdot}_{\infty}$. Both plots taken together highlight that some flexibility in choosing the $L_j$ is desirable.
\begin{figure}
\centering
\includegraphics[width=0.45\textwidth]{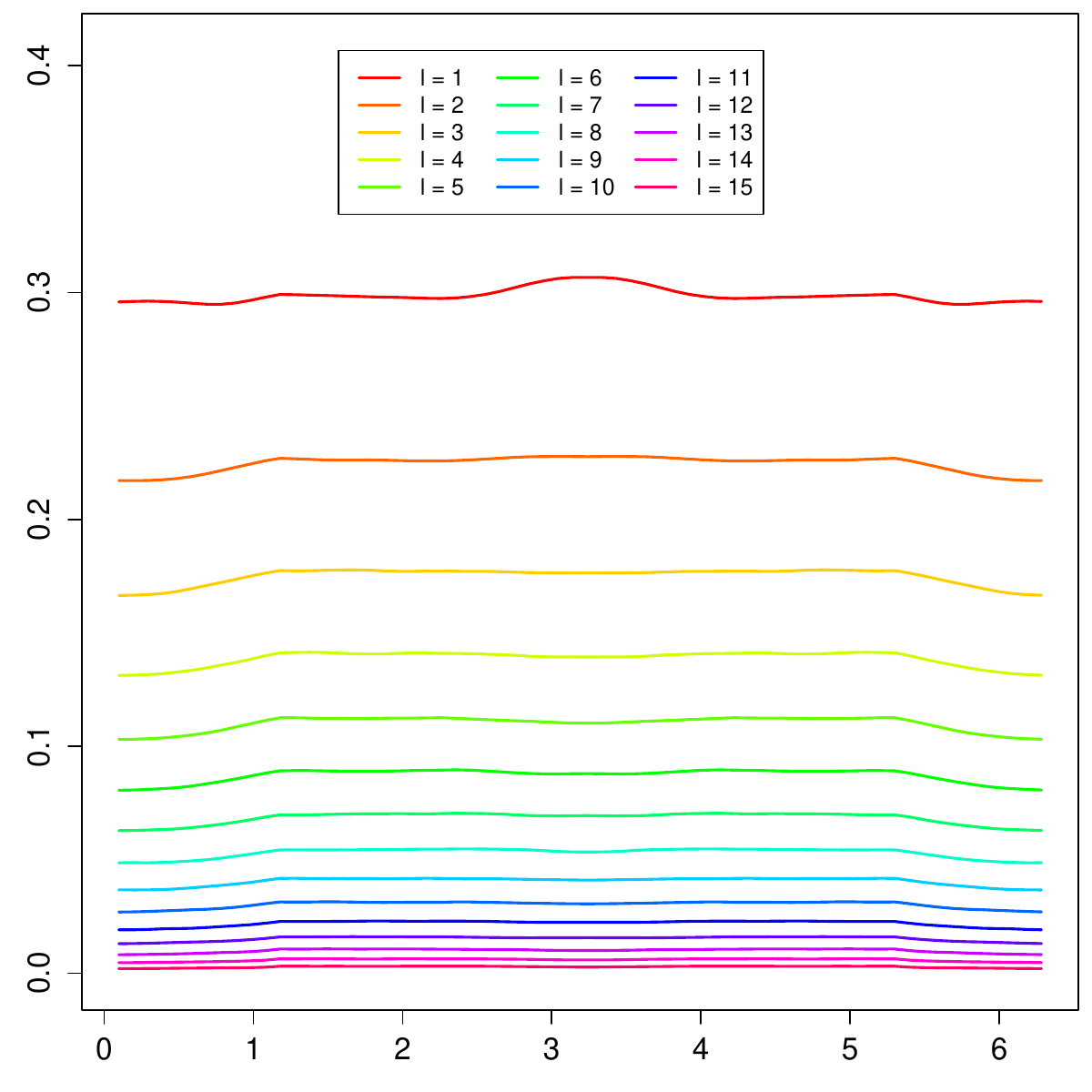}\quad
\includegraphics[width=0.45\textwidth]{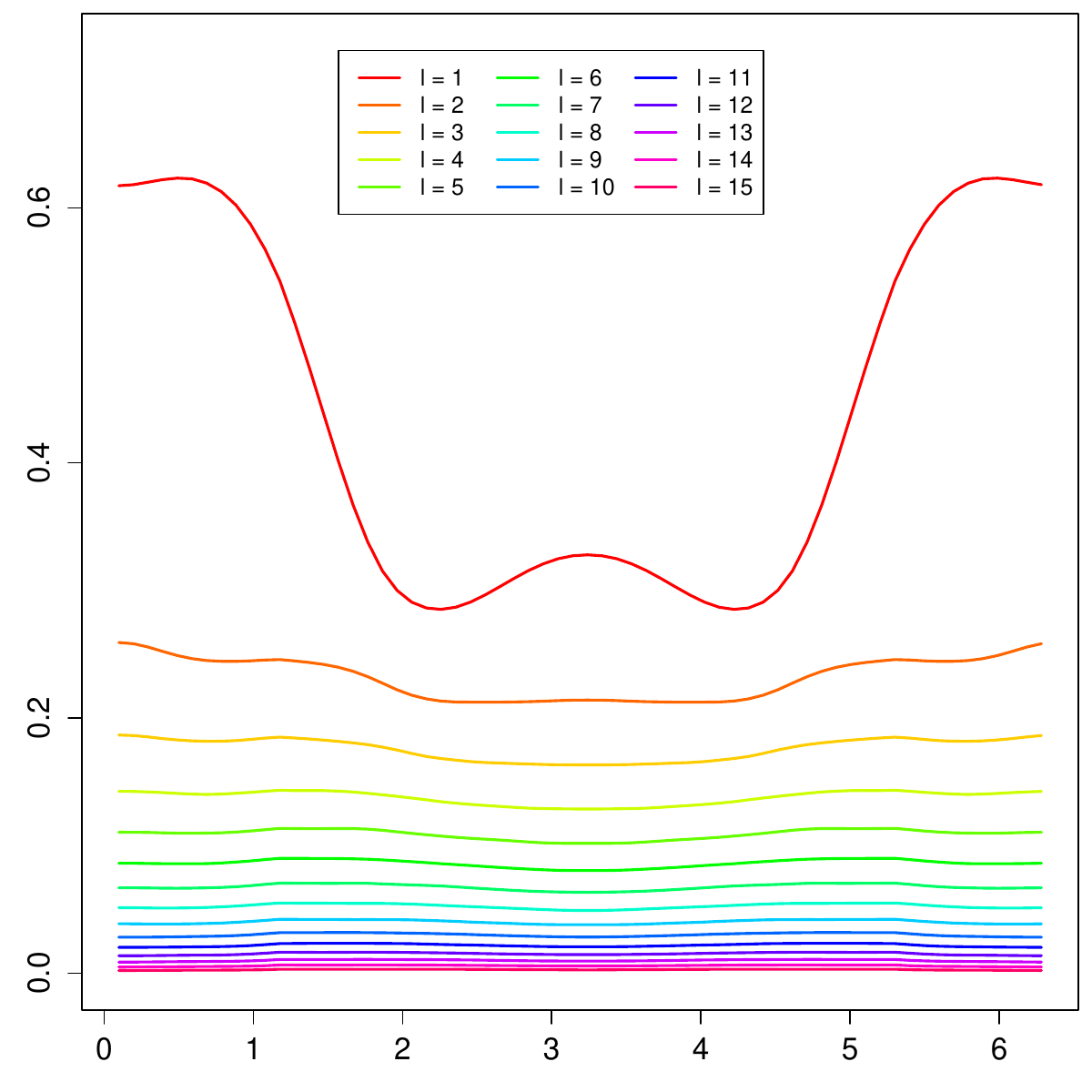}  
\caption{Plot of sample eigenvalues $\lambda_l^{\omega_j}$ across Fourier frequencies $\omega_j$ for $l=1,\ldots,15$ for a functional white noise process (left) and a second-order functional autoregression (right). }
\label{fig:eig}
\end{figure}}

\textcolor{black}{To accommodate the previous observation, the following arrangements were made for the standardized test based on $\hat Q_{M,s}^{(T)}$. In the first part, a reasonable level of variation explained at each frequency $\omega_j$ is ensured through requiring that $0.5<\mathrm{TVE}_j<0.9$ for all $j$. In the second part, the procedure adapts to different eigenvalue decays by choosing 
\[
L_j=\max\bigg\{l\colon \frac{\lambda_l^{\omega_j}}{\lambda_1^{\omega_j}}>.2-\frac{1}{\sqrt{bT}}\bigg\}
\]
subject to the TVE criterion being satisfied. If no such $L_j$ exists, choose $L_j=1$. The unstandardized test statistics is very stable in practice and does not require the specification of tuning parameters.} 

Estimation of the spectral density operator and its eigenelements, needed to compute the two statistics, was achieved using \eqref{eq:estimatestat} with the concave smoothing kernel $K(x)=6(0.25-x^2)$ with compact support on $x\in[-1/2,1/2]$ and bandwidth $b=T^{-.26}$. The fourth-order estimation is done with $K_4(x_1,\ldots,x_4)=\prod_{j=1}^4K(x_j)$, where $K$ is same as before, 
and bandwidth $b_4= T^{-1/5}$. It should be noted that the outcomes were not overly sensitive with respect to bandwidth choices for $b$ respecting Assumption \ref{windowfunction}. It is worthwhile to mention that the computational complexity of the fourth-order estimator is considerable for larger sample sizes. The implementation was therefore partially done with the compiler language {\tt {C++}} and the {\tt{Rcpp}}-package in {\tt{R}}.

\subsection{Finite sample performance under the null}
\label{subsec:empirical:null}

Under the null hypothesis of stationarity the following data generating processes, DGPs, were studied:
\begin{itemize}\itemsep-.2ex
\item[(a)] The Gaussian FWN variables $\varepsilon_1,\ldots,\varepsilon_T$ with coefficient variances $\mathrm{Var}(\langle\varepsilon_t,\psi_l\rangle)=\exp({\color{black}{-}}(l-1)/10)$;
\item[(b)] The FAR(2) variables $X_1,\ldots,X_T$ with operators specified through the respective variances $\nu^{(1)}_{l,l^\prime}=\exp(-l-l^\prime)$ and $\nu^{(2)}_{l,l^\prime}=1/(l+{l^\prime}^{3/2})$ and operator norms $\kappa_1=0.75$ and $\kappa_2=-0.4$, and innovations $\varepsilon_1,\ldots,\varepsilon_T$ as in (a);
\item[(c)] The FAR(2) variables $X_1,\ldots,X_T$ as in (b) but with operator norms $\kappa_1=0.4$ and $\kappa_2=0.45$. 
\end{itemize}
The sample sizes under consideration are $T=2^n$ for $n=6,\ldots,10$, so that the smallest sample size consists of $64$ functions and the largest of $1024$. The processes in (a)--(c) comprise a range of stationary scenarios. DGP (a) is the simplest model, specifying an independent FWN process. DGPs (b) and (c) exhibit significant second-order autoregressive dynamics of different persistence.

\begin{table}[h!]
 \begin{center}
\resizebox{\textwidth}{!}{
\begin{tabular}{c r @{\qquad} c ll @{\qquad} c ll @{\qquad} c ll @{\qquad} c ll}
\hline
&& &   \multicolumn{2}{l  }{\% level}  & \multicolumn{2}{l  }{\% level}  &  &\multicolumn{2}{l }{\% level}&  & \multicolumn{2}{l }{\% level}\\ 
&$T$~~ & $\hat{Q}_{1,u}^{(T)}  $ & 5 &1 & $\hat{Q}_{5,u}^{(T)}  $ & 5 &1 & $ \hat{Q}_{1,s}^{(T)} $ & 5&  1 & $\hat{Q}_{5,s}^{(T)}  $  & 5&  1\\ \hline
(a) 
&64&  1.33 & 5.80 & 1.40 & 8.93 & 9.10 & 2.60 & 1.29 & 4.30 & 1.50 & 8.26 & 7.80 & 2.70 \\
&128& 1.41 & 5.90 & 1.20 & 9.03 & 7.20 & 2.10 & 1.36 & 5.70 & 1.00 & 8.96 & 5.70 & 2.10 \\
&256& 1.26 & 5.10 & 0.90 & 9.15 & 5.30 & 1.70 & 1.27 & 5.20 & 1.40 & 9.02 & 5.10 & 1.00 \\
&512 &1.37 & 4.80 & 1.30 & 9.27 & 6.80 & 1.40 & 1.40 & 4.60 & 1.30 & 9.16 & 6.30 & 1.30 \\ 
&1024 & 1.32 & 4.70 & 1.20 & 9.19 & 5.20 & 1.50 & 1.33 & 5.40 & 0.60 & 9.33 & 4.60 & 1.10 \\ 
\hline 
(b)
& 64 & 1.58 & 6.00 & 1.50 & 9.50 & 9.40 & 3.50 & 1.35 & 5.70 & 1.40 & 8.65 & 6.10 & 2.70\\ 
& 128 & 1.44 & 5.70 & 1.60 & 9.35 & 8.90 & 2.80 & 1.30 & 4.70 & 1.50 & 8.72 & 6.30 & 1.70 \\ 
& 256 & 1.28 & 4.20 & 0.90 & 9.11 & 6.20 & 2.30 & 1.32 & 4.70 & 0.60 & 8.78 & 7.00 & 1.70 \\ 
& 512 & 1.32 & 5.00 & 1.70 & 9.42 & 6.70 & 1.90 & 1.26 & 4.70 & 0.90 & 9.11 & 6.10 & 0.90 \\ 
& 1024 & 1.44 & 4.40 & 0.80 & 9.26 & 5.40 & 1.10 & 1.32 & 4.70 & 0.50 & 8.87 & 4.80 & 0.90 \\ 
\hline 
(c) 
& 64 & 1.42 & 5.60 & 1.90 & 8.50 & 7.60 & 3.30 & 1.20 & 5.70 & 0.90 & 8.36 & 8.20 & 2.60 \\ 
& 124 & 1.31 & 5.20 & 1.00 & 9.05 & 6.20 & 2.50 & 1.29 & 4.00 & 0.50 & 8.77 & 5.70 & 2.00 \\ 
& 256 & 1.48 & 6.10 & 1.20 & 9.19 & 6.70 & 1.90 & 1.42 & 5.20 & 1.70 & 8.90 & 6.10 & 1.30 \\ 
& 512 & 1.35 & 5.60 & 0.70 & 9.48 & 4.90 & 1.00 & 1.41 & 4.50 & 0.60 & 8.99 & 5.30 & 1.40\\ 
& 1024 & 1.34 & 6.90 & 1.60 & 9.26 & 5.70 & 1.30 & 1.35 & 4.60 & 1.10 & 9.10 & 4.40 & 0.90 \\ \hline 
\end{tabular}
}
\caption{Median of test statistic values and rejection rates of $\hat{Q}_{M,u}^{(T)}$ and $\hat{Q}_{M,s}^{(T)}$ at the 1\% and 5\% asymptotic level for the processes (a)--(c) for various choices of $M$ and $T$. All table entries are generated from 1000 repetitions.}
\label{tab:null}
\end{center}
\end{table}

The empirical rejection levels for the processes (a)--(c) can be found in Table \ref{tab:null}. It can be seen that the empirical levels for both statistics with $M=1$ are generally well adjusted with slight deviations in a few cases.  
The performance of the statistics with $M=5$ is similar, although the empirical rejection levels tend towards the nominal ones with increasing sample size. 
Some evidence on closeness between empirical and limit densities for the statistics $\hat Q_{5,u}^{(T)}$ and $\hat Q_{5,s}^{(T)}$ are provided in Figure \ref{fig:null}. 

\textcolor{black}{Figure \ref{fig:Lj} shows the average choices of $L$ over the 1000 repetitions for the various DGPs for the sample sizes $T=64$ and $T=1024$. First, one can see that the average $L$ increases with the sample size, as more degrees of freedom become available. For the small sample size $T=64$, choices of $L_j$ under the null hypothesis are more similar both across frequencies and across the three DGPs because the form of dependence is not yet entirely evident. With increasing sample size, the average $L_j$ increases uniformly for DGP (a), while for DGPs (b) and (c) $L_j$ in certain frequency bands are accentuated while others are attenuated according to their contributions to the spectral analysis of variance of the underlying functional time series. For DGP (b) the shape of the  curve $\omega_j\mapsto L_j$ might also be compared to the shape of the curve $\omega_j\mapsto\lambda_1^{\omega_j}$ in the right panel of Figure \ref{fig:eig}.
}

\begin{figure}
\centering
\includegraphics[width=0.4\textwidth]{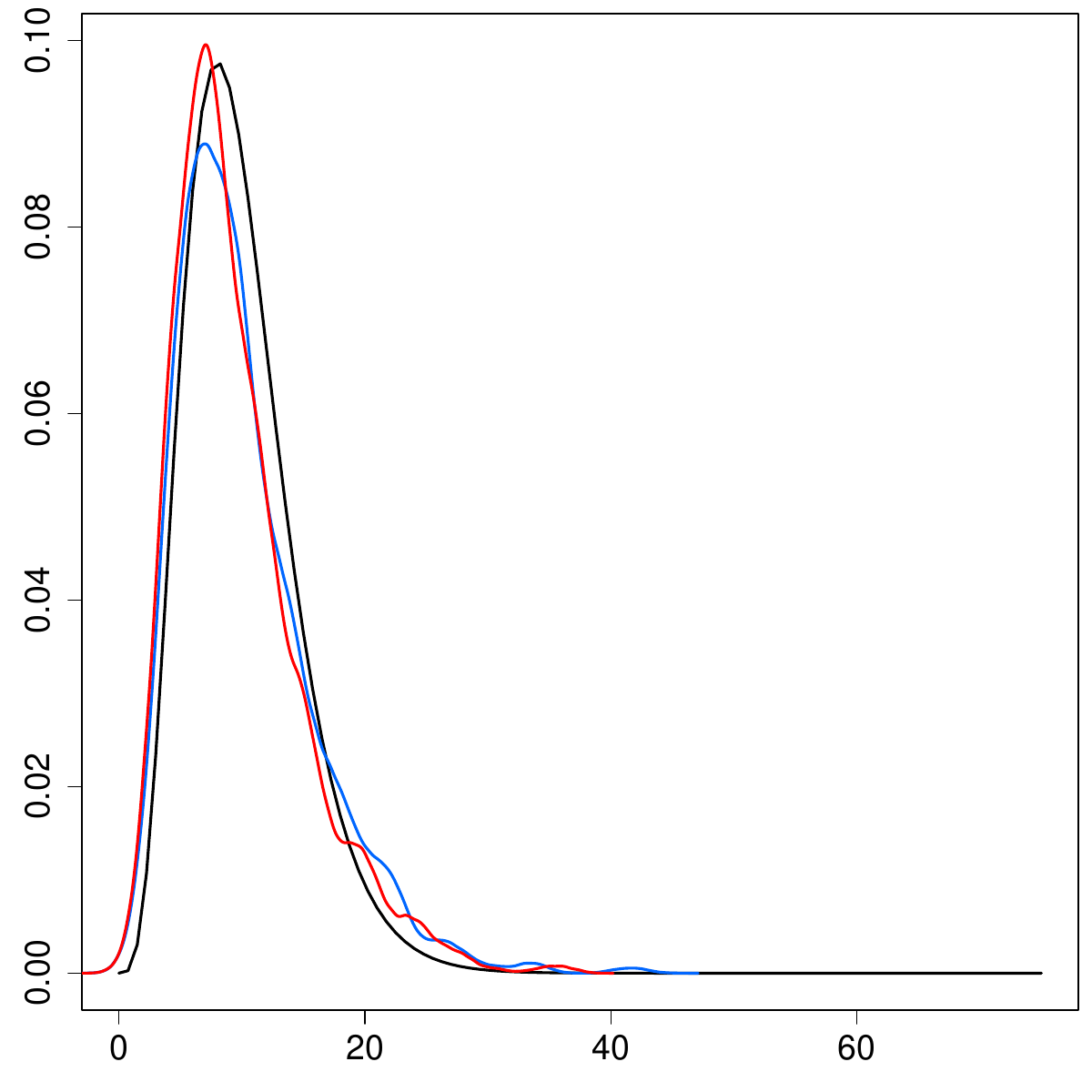}  
\includegraphics[width=0.4\textwidth]{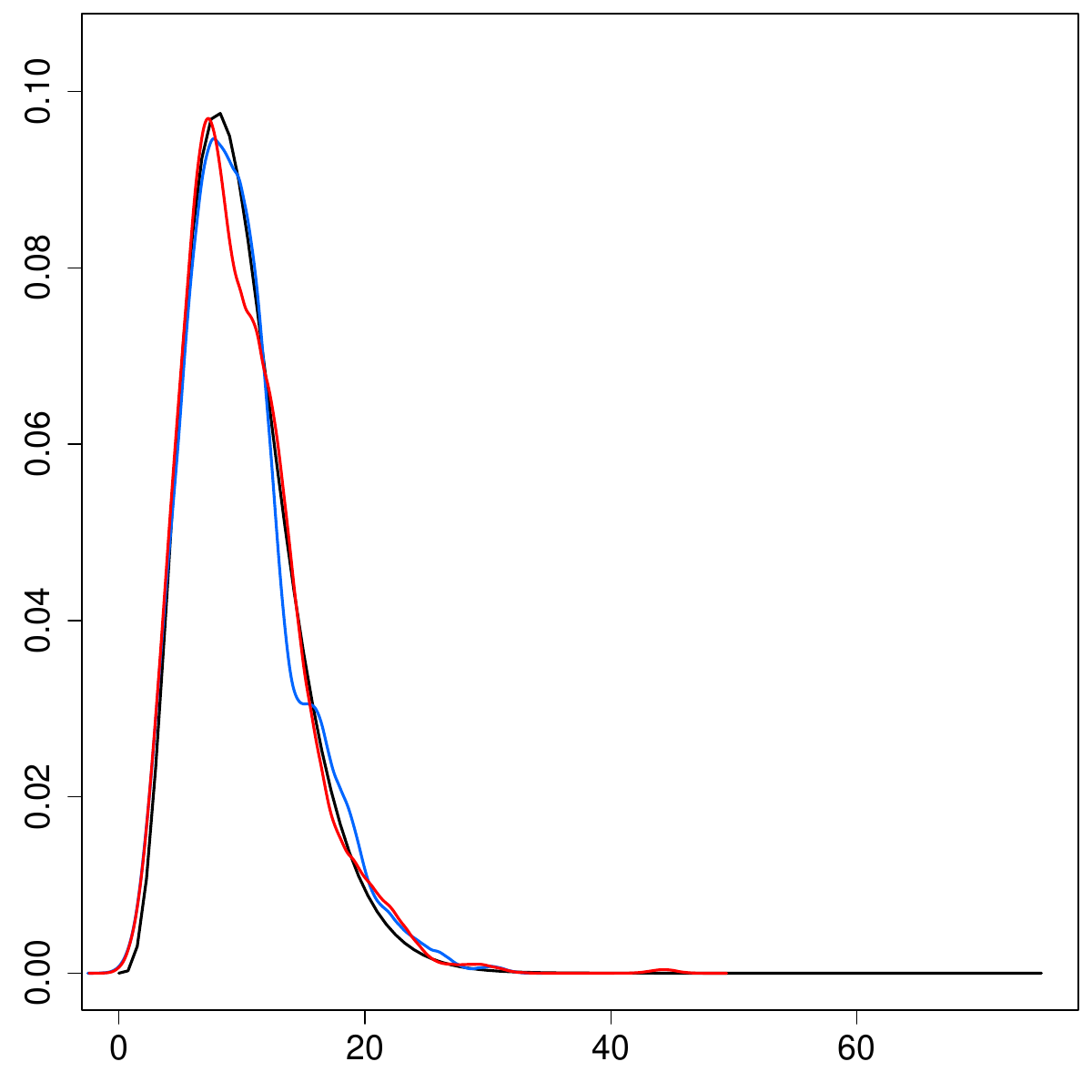}
\includegraphics[width=0.4\textwidth]{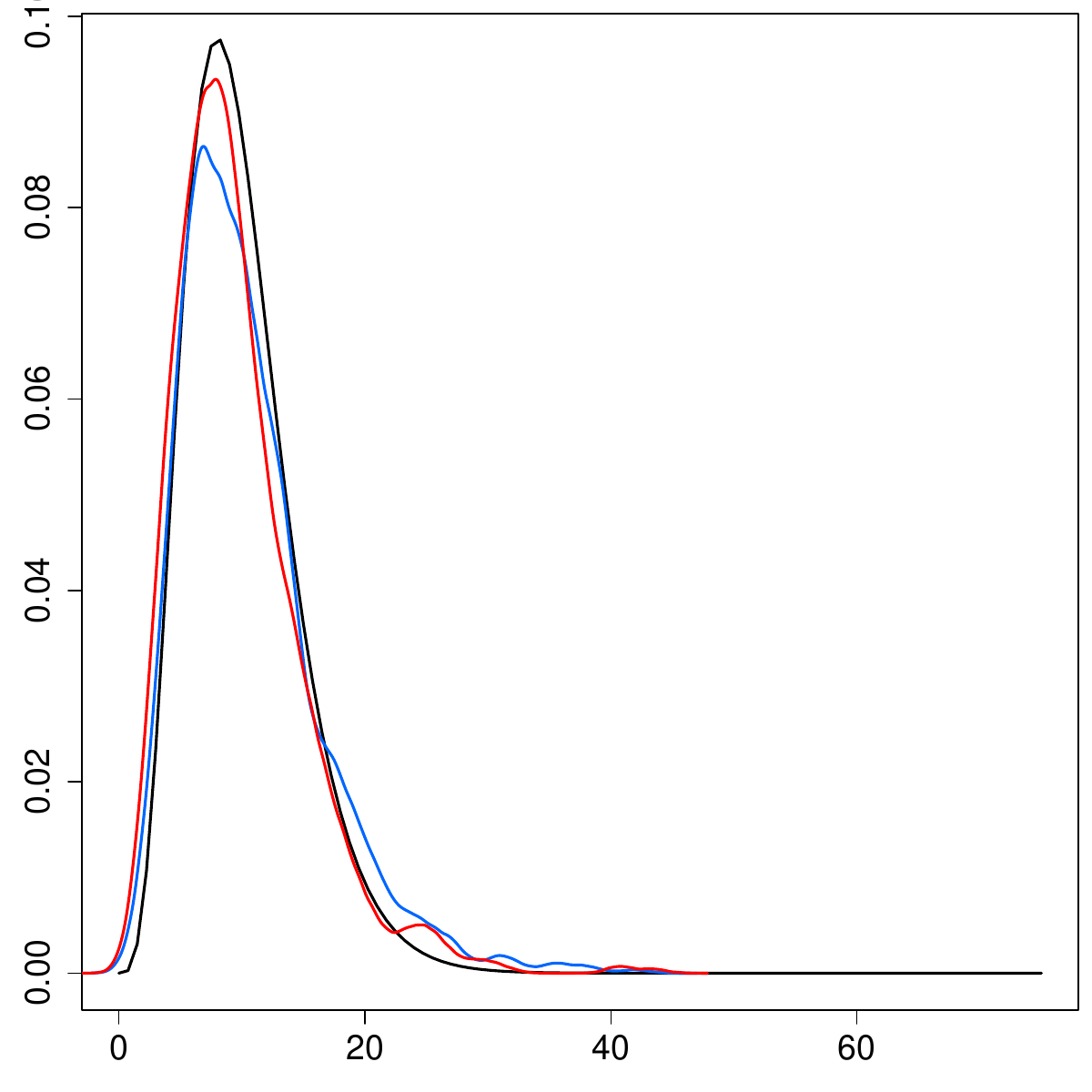}  
\includegraphics[width=0.4\textwidth]{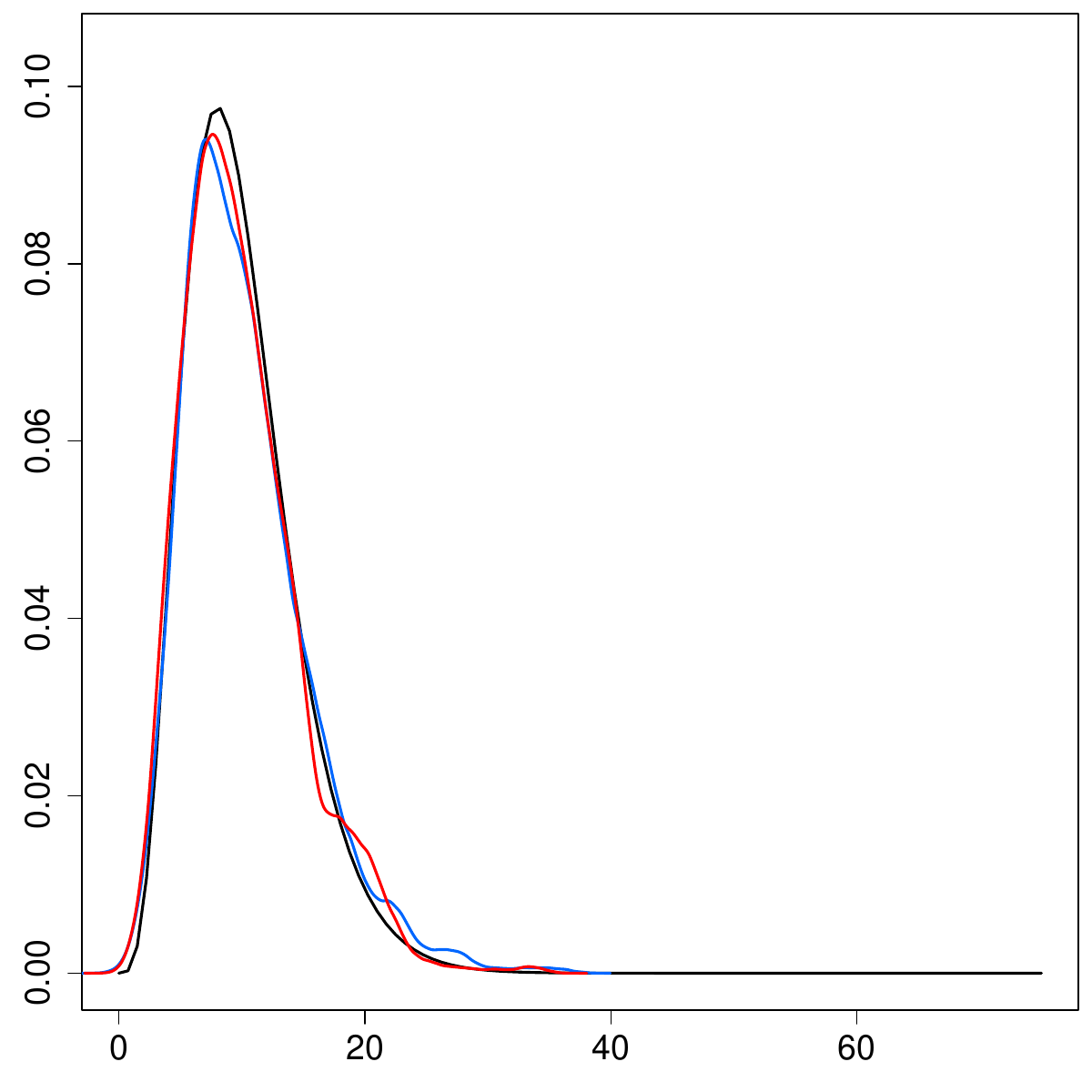}
\includegraphics[width=0.4\textwidth]{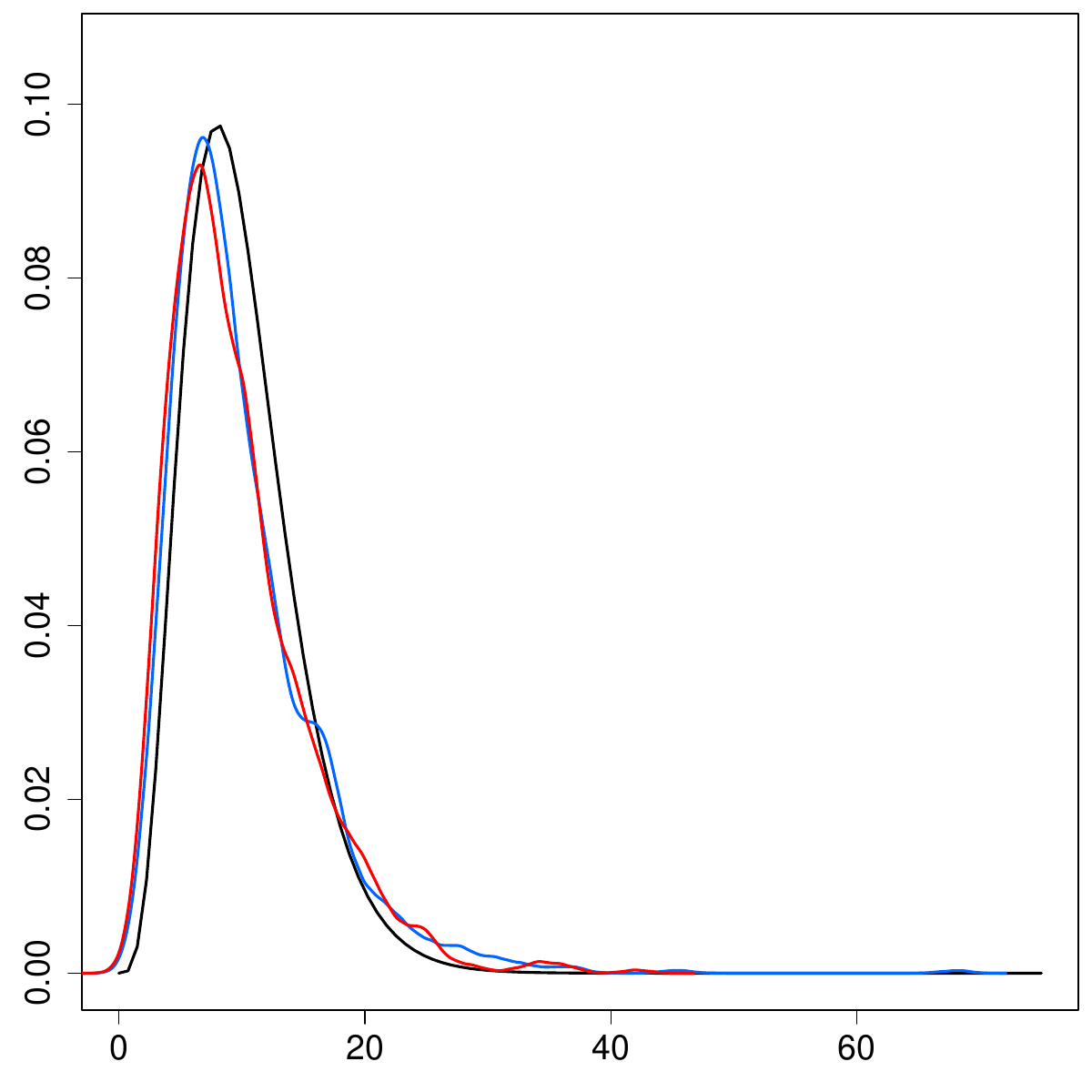}  
\includegraphics[width=0.4\textwidth]{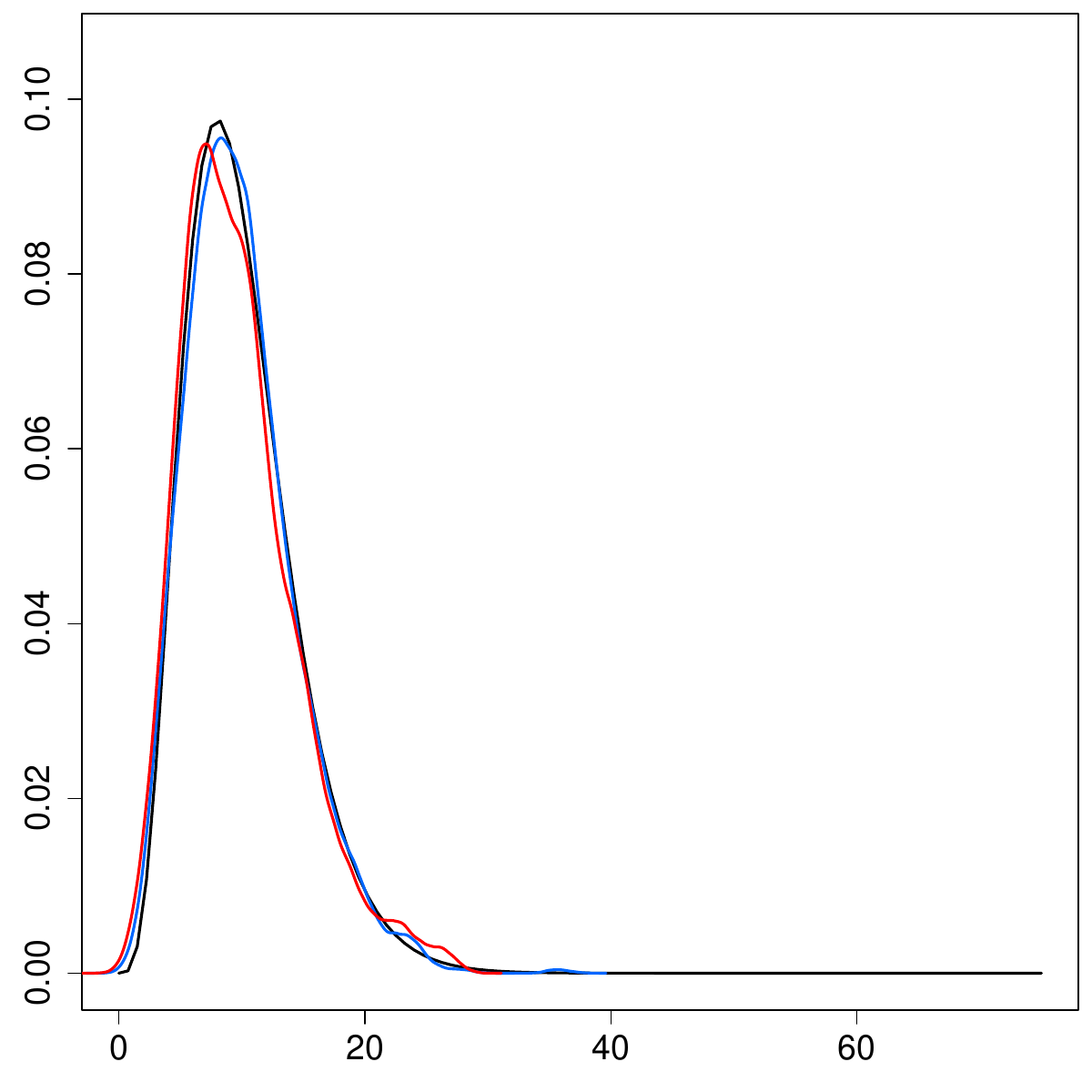}
\caption{Empirical density of $\hat Q_{5,u}^{(T)}$ (black) and $\hat Q_{5,s}^{(T)}$ (blue) for $T=64$ (left panel) and $T=512$ (right panel) for DGPs (a)--(c) (top to bottom). Red: The corresponding chi-squared densities predicted under the null.}
\label{fig:null}
\end{figure}

\begin{figure}
\centering
\includegraphics[width=0.4\textwidth]{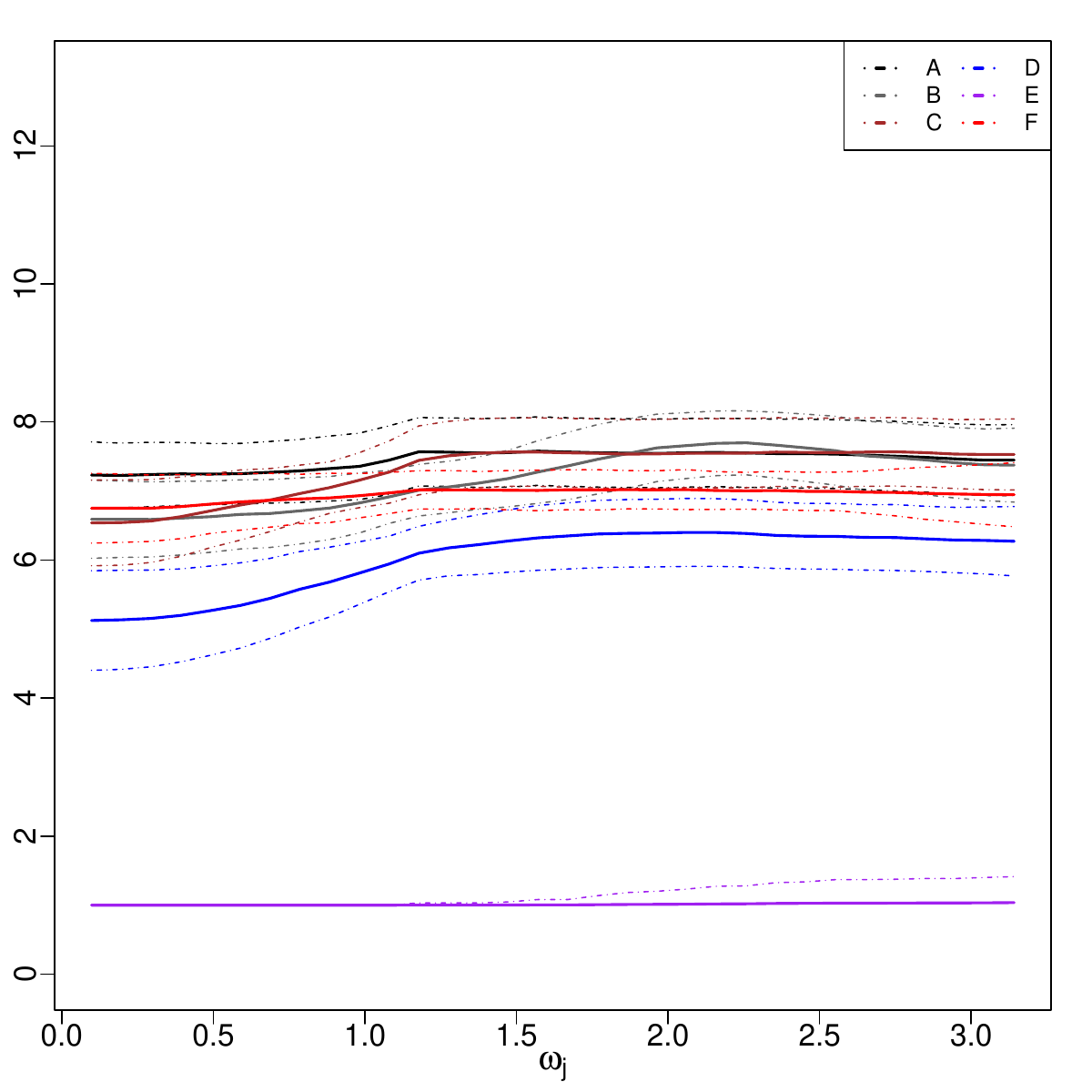}  
\includegraphics[width=0.4\textwidth]{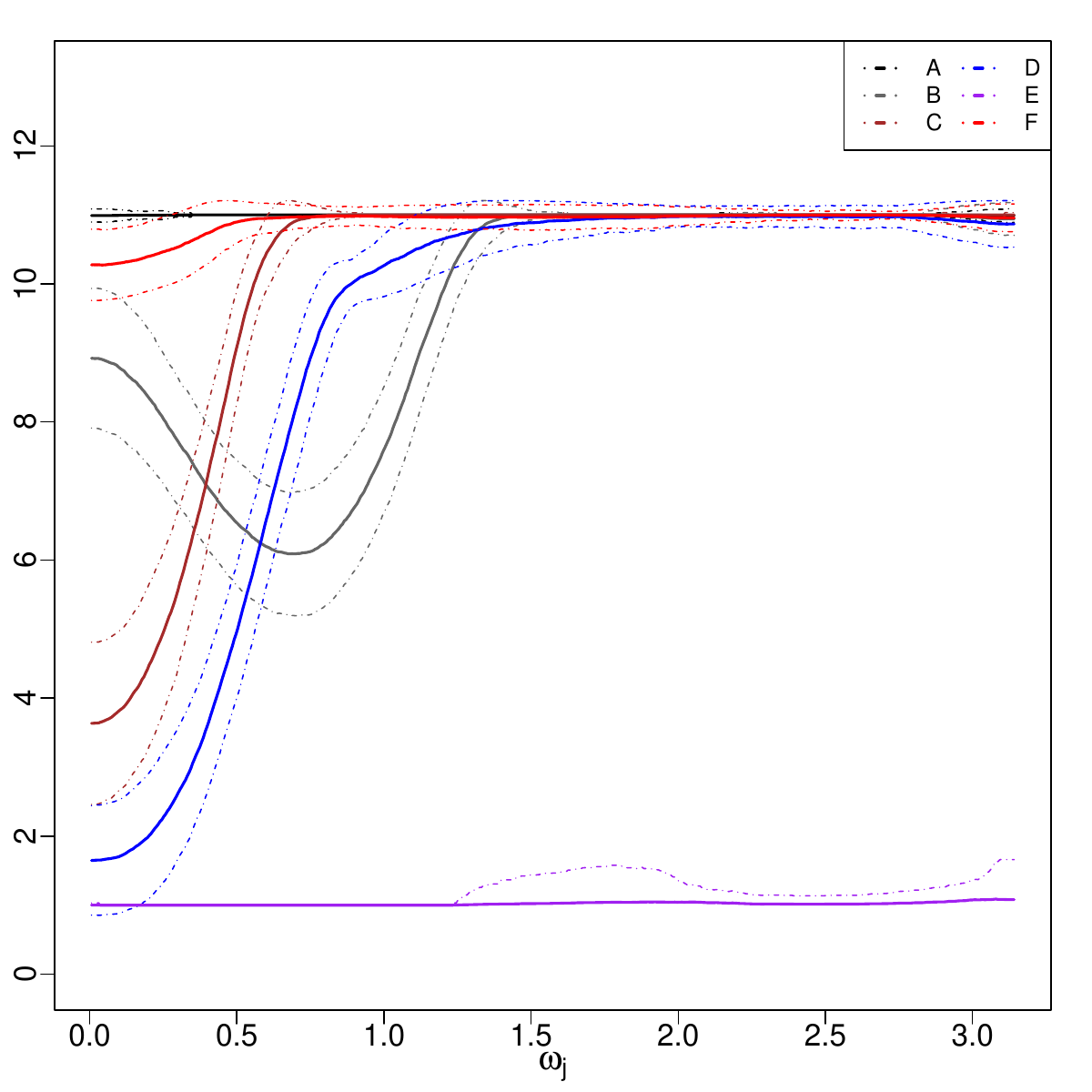}
\caption{Average choice of truncation level $L_j$ against frequency $\omega_j$ for the six DGPs (solid lines) with respective standard deviations (dashed lines) for $T=64$ (left) and $T=1024$ (right).}
\label{fig:Lj}
\end{figure}

\subsection{Finite sample performance under the alternative}
\label{subsec:empirical:alt}

Under the alternative, the following data generating processes are considered:
\begin{itemize}\itemsep-.2ex
\item[(d)] The tvFAR(1) variables $X_1,\ldots,X_T$ with operator specified through the variances $\nu^{(1)}_{l,l^\prime}=\exp(-l-l^\prime)$ and operator norm $\kappa_1=0.8$, and innovations given by (a) with added multiplicative time-varying variance
\[
{\color{black}{
\sigma^2(t)=
\cos\bigg(\frac 12+\cos\bigg(\frac{2\pi t}{T}\bigg)+0.3\sin\bigg(\frac{2\pi t}{T}\bigg)\bigg);
}}
\]
\item[(e)] The tvFAR(2) variables $X_1,\ldots,X_T$ with both operators as in (d) but with time-varying operator norm
\[
\kappa_{1,t}=1.8\cos\bigg(1.5-\cos\bigg(\frac{4\pi t}{T}\bigg)\bigg),
\]
constant operator norm $\kappa_2=-0.81$, and innovations as in (a);
\item[(f)] The structural break FAR(2) variables $X_1,\ldots,X_T$ given in the following way. 
\begin{itemize}
\item For $t\leq 3T/8$, the operators are as in (b) but with operator norms $\kappa_1=0.7$ and $\kappa_2=0.2$, and innovations as in (a);
\item For $t> 3T/8$, the operators are as in (b) but with operator norms $\kappa_1=0$ and $\kappa_2=-0.2$, and innovations as in (a) but with variances $\mathrm{Var}(\langle\varepsilon_t,\psi_l\rangle)=2\exp({\color{black}{-}}(l-1)/10)$.
\end{itemize}
\end{itemize}
All other aspects of the simulations are as in Section~\ref{subsec:empirical:null}. The processes studied under the alternative provide intuition for the behavior of the proposed tests under different deviations from the null hypothesis. DGP (d) is time-varying only through the innovation structure, in the form of a slowly varying variance component. The first-order autoregressive structure is independent of time. DGP (e) is a time-varying second-order FAR process for which the first autoregressive operator varies with time. The final DGP in (f) models a structural break, a different type of alternative. Here, the process is not locally stationary as prescribed under the alternative in this paper, but piecewise stationary with the two pieces being specified as two distinct FAR(2) processes. 

\begin{table}[h!]
\vspace*{.3cm} 
 \begin{center}
\resizebox{\textwidth}{!}{
\begin{tabular}{c r @{\quad} rrr @{\quad} rrr @{\quad} rrr @{\quad} rrr }
\hline
&& &   \multicolumn{2}{c  }{\% level}  &  &\multicolumn{2}{c }{\% level}&  & \multicolumn{2}{c }{\% level} & &\multicolumn{2}{c }{\% level}\\ 
&$T$~~ & $\hat{Q}_{1,u}^{(T)}  $ & 5\phantom{.00} & 1\phantom{.00} & $\hat{Q}_{5,u}^{(T)}  $ & 5\phantom{.00} & 1\phantom{.00} & $ \hat{Q}_{1,s}^{(T)} $ & 5\phantom{.00} & 1\phantom{.00} & $\hat{Q}_{5,s}^{(T)}$ & 5\phantom{.00} & 1\phantom{.00} \\ \hline
%
(d)%
&64& 9.84 & 77.80 & 54.30 & 20.33 & 57.30 & 39.80 & 8.61 & 71.30 & 46.70 & 17.92 & 48.80 & 30.70 \\ 
&128& 19.55 & 99.00 & 94.40 & 33.34 & 94.10 & 84.10 & 18.26 & 98.20 & 91.40 & 30.44 & 90.20 & 76.30 \\ 
&256 &  36.70 & 100.00 & 100.00 & 54.07 & 99.90 & 99.70 & 34.40 & 100.00 & 100.00 & 50.27 & 99.80 & 99.40 \\ 
&512 & 69.49 & 100.00 & 100.00 & 94.47 & 100.00 & 100.00 & 62.90 & 100.00 & 100.00 & 84.75 & 100.00 & 100.00  \\ 
&1024 &  140.53 & 100.00 & 100.00 & 179.75 & 100.00 & 100.00 & 118.18 & 100.00 & 100.00 & 152.12 & 100.00 & 100.00\\ 
\hline
(e)%
&64 & 33.38 & 100.00 & 100.00 & 131.80 & 100.00 & 98.10 & 33.46 & 99.50 & 99.20 & 100.13 & 99.30 & 99.20 \\ 
&128 & 49.04 & 100.00 & 100.00 & 118.13 & 100.00 & 100.00 & 66.48 & 99.70 & 99.30 & 172.30 & 99.80 & 99.70\\ 
&256 & 98.43 & 100.00 & 100.00 & 393.65 & 100.00 & 100.00 & 151.44 & 99.70 & 99.60 & 568.55 & 99.90 & 99.80 \\  
&512 & 173.35 & 100.00 & 100.00 & 763.11 & 100.00 & 100.00 & 302.51 & 100.00 & 100.00 & 1257.93 & 100.00 & 100.00 \\ 
&1024 & 286.54 & 99.90 & 99.90 & 1311.08 & 100.00 & 100.00 & 579.00 & 99.80 & 99.80 & 2484.54  & 100.00 & 99.90  \\ 
\hline 
(f) %
&64 & 5.64 & 46.50 & 25.40 & 15.02 & 33.70 & 19.90 & 4.38 & 35.20 & 16.50 & 12.36 & 24.40 & 12.70\\ 
&128 & 10.90 & 82.80 & 60.90 & 21.65 & 64.30 & 43.00 & 8.93 & 83.10 & 48.40 & 18.37 & 50.40 & 29.30 \\ 
&256 &18.29 & 98.20 & 90.50 & 30.40 & 90.00 & 77.50 & 15.71 & 95.70 & 85.20 & 27.03 & 84.70 & 66.00 \\  
&512 & 31.81 & 100.00 & 100.00 & 47.49 & 99.90 & 99.20 & 30.71 & 99.90 & 99.80 & 45.71 & 99.80 & 98.50 \\
&1024 &  62.72 & 100.00 & 100.00  & 83.82 & 100.00 & 100.00 & 62.29 & 100.00 & 100.00 & 83.18 & 100.00 & 100.00 \\ 
\hline
\end{tabular}
}
\caption{Median of test statistic values and rejection rates of $\hat{Q}_{M,u}^{(T)}$ and $\hat{Q}_{M,s}^{(T)}$ at the 1\% and 5\% asymptotic level for the processes (d)--(f) for various choices of $M$ and $T$. All table entries are generated from 1000 repetitions.}
\label{tab:alt}
\end{center}
\end{table}

The empirical power of the various test statistics for the processes in (d)--(f) are in Table~\ref{tab:alt}. Power results are roughly similar across the selected values of $M$ for both statistics. For DGP (f) and to some extend for DGP (d), power is low for the small sample sizes $T=64$. It reaches 100\% for all $T$ larger or equal to 256 for all DGPs but (f), where close to perfect detection is reached for $T=512$.
Generally, the standardized statistics is slightly more unstable than its unstandardized counterpart for DGP (e), while both statistics behave remarkably similar for the other processes.  
The results for DGP (f) indicate that the proposed statistics have power against structural break alternatives. \textcolor{black}{This is intuitive since the second-order structure is in this case not invariant under translations of time and hence induces a non-zero mean in the test statistics.}

Figure \ref{fig:alt} exhibits exemplary the empirical densities for DGP (d). It can be seen that the deviation from the chi-squared distribution predicted under the null hypothesis grows with increasing sample size. Figure \ref{fig:Lj} contains the average choice of $L_j$ for DGPs (d)-(f) under the alternative. While processes (d) and (f) display behavior more similar to the null DGPs, process (e) is significantly different, as almost always only one principal component is chosen at each frequency for both the small and the large sample size.

\begin{figure}[h!]
\centering
\vspace{.5cm}
\includegraphics[width=0.4\textwidth]{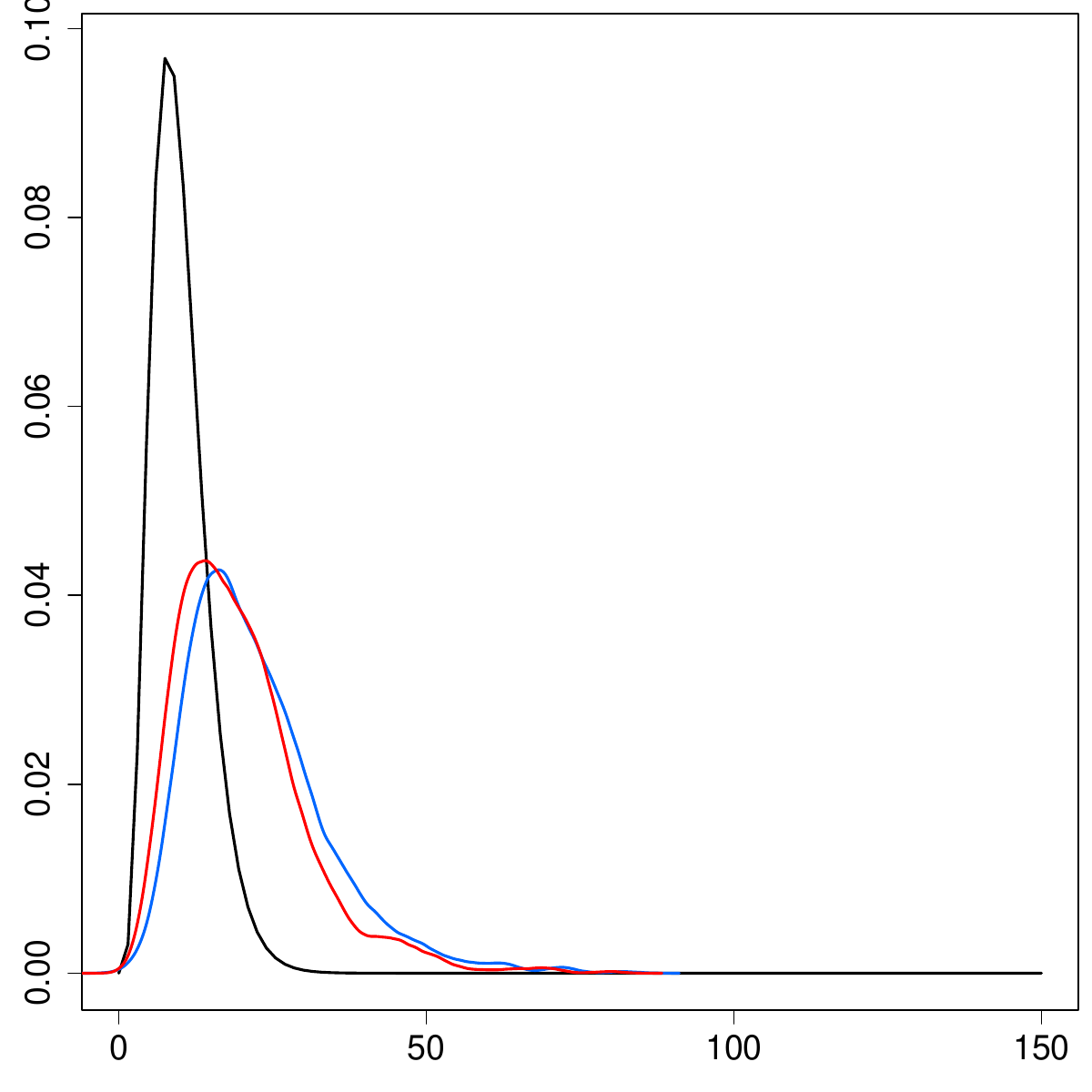}  
\includegraphics[width=0.4\textwidth]{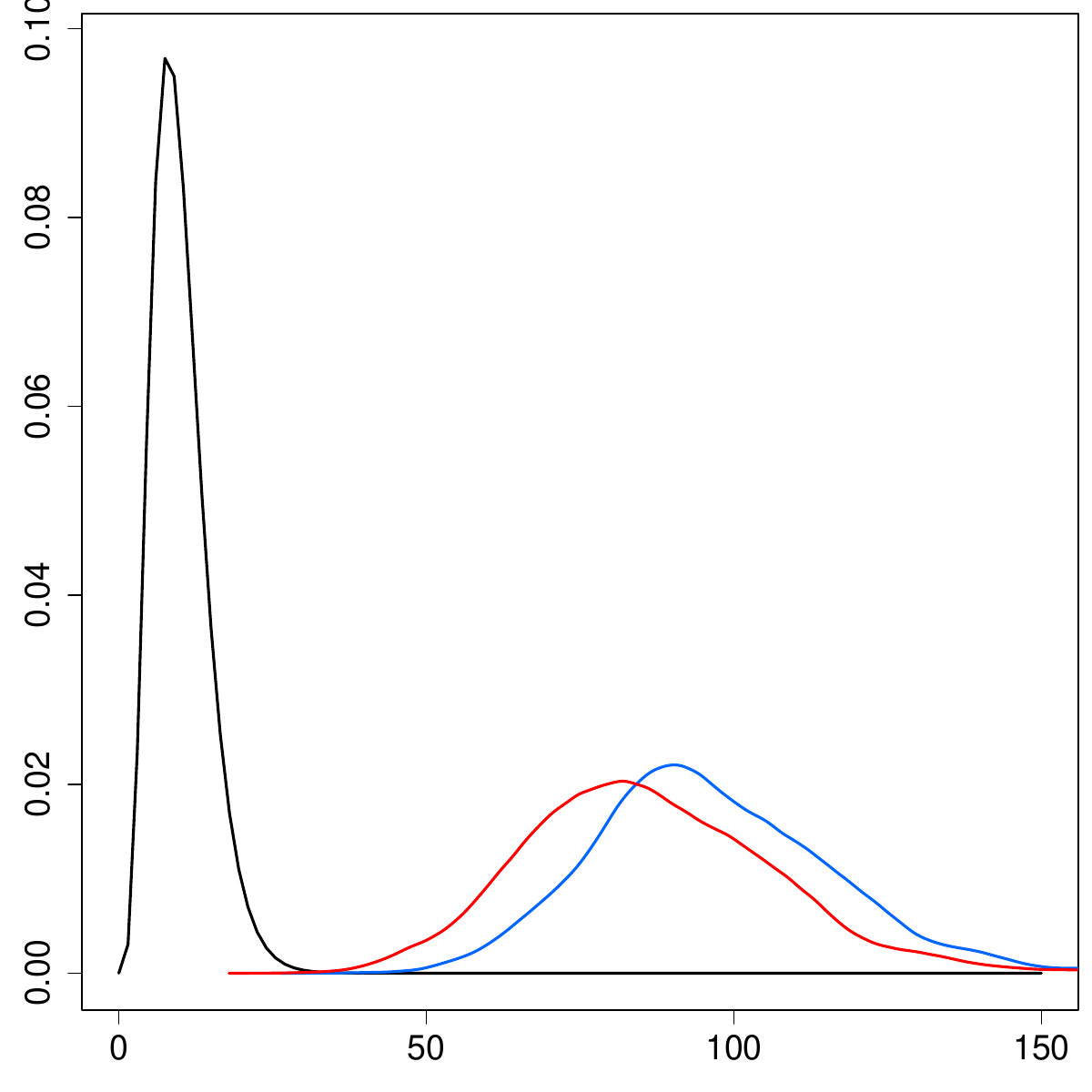}
\caption{Empirical density of $\hat Q_{5,u}^{(T)}$ (black) and $\hat Q_{5,s}^{(T)}$ (blue) for $T=64$ (left panel) and $T=512$ (right panel) for DGP (d). 
}
\label{fig:alt}
\end{figure}

\subsection{Finite sample performance under non-Gaussian observations}
\label{subsec:empirical:non-Gauss}

\textcolor{black}{In this section, the behavior of the eigenbased test under non-Gaussianity is further investigated through the following processes:
\begin{itemize}\itemsep-.2ex
\item[(g)] The FAR(2) variables $X_1,\ldots,X_T$ as in (b) but with both independent $t_{19}$-distributed FWN and independent $\beta(6,6)$-distributed FWN;
\item[(h)] The tvFAR(1) variables $X_1,\ldots,X_T$ as in (d) but with independent $t_{\color{black}{19}}$-distributed FWN and independent $\beta(6,6)$-distributed FWN. 
\end{itemize}
For direct comparison, both $t_{19}$- and $\beta(6,6)$-distributions were standardized to conform to zero mean and unit variance as the standard normal. All other aspects are as detailed in Section~\ref{subsec:empirical:null}. The additional simulations were designed to shed further light on the effect of estimating the fourth-order spectrum in situations deviating from the standard Gaussian setting. Note in particular that the $t_{19}$-distribution serves as an example for leptokurtosis (the excess kurtosis is $0.4$) and the $\beta(6,6)$ distribution for platykurtosis (the excess kurtosis is $-0.4$). Process (g) showcases the behavior under the null, while process (h) highlights the performance under the alternative. The corresponding results are given in Table \ref{tab:fourth} and can be readily compared with corresponding outcomes for the Gaussian processes (b) and (d) in Tables \ref{tab:null}--.\ref{tab:alt}.}

\begin{table}[h!]
\vspace*{.3cm} 
 \begin{center}
\resizebox{\textwidth}{!}{
\begin{tabular}{c r @{\quad} rrr @{\quad} rrr @{\quad} rrr @{\quad} rrr }
\hline
&& &   \multicolumn{2}{c  }{\% level}  &  &\multicolumn{2}{c }{\% level}&  & \multicolumn{2}{c }{\% level} & &\multicolumn{2}{c }{\% level}\\ 
&$T$~~ & $\hat{Q}_{1,u}^{(T)}  $ & 5\phantom{.00} & 1\phantom{.00} & $\hat{Q}_{5,u}^{(T)}  $ & 5\phantom{.00} & 1\phantom{.00} & $ \hat{Q}_{1,s}^{(T)} $ & 5\phantom{.00} & 1\phantom{.00} & $\hat{Q}_{5,s}^{(T)}$ & 5\phantom{.00} & 1\phantom{.00} \\ \hline
%
(g), $t$%
& 64 & 1.58 & 5.60 & 0.70 & 9.46 & 7.70 & 2.70 & 1.40 & 2.90 & 0.60 & 8.65 & 6.20 & 1.40 \\ 
& 128 & 1.42 & 4.50 & 1.30 & 9.37 & 6.90 & 2.00 & 1.30 & 3.60 & 0.40 & 8.81 & 4.80 & 1.30 \\ 
& 256 & 1.40 & 4.40 & 0.80 & 9.17 & 5.50 & 1.70 & 1.29 & 4.90 & 0.90 & 8.89 & 5.70 & 1.20 \\ 
& 512 & 1.47 & 4.70 & 0.70 & 9.33 & 4.70 & 1.30 & 1.44 & 4.10 & 0.50 & 9.32 & 4.30 & 1.20 \\ 
& 1024 & 1.53 & 5.90 & 0.40 & 9.52 & 5.00 & 1.00 & 1.43 & 5.40 & 0.90 & 8.92 & 4.70 & 1.10 \\ 
\hline
(g), $\beta$%
&64& 1.31 & 3.10 & 0.80 & 9.00 & 6.70 & 1.60 & 1.29 & 3.40 & 0.80 & 8.66 & 5.50 & 1.10 \\ 
 & 128 & 1.37 & 4.80 & 1.10 & 9.13 & 6.10 & 1.90 & 1.25 & 3.70 & 0.60 & 8.89 & 4.50 & 0.90 \\ 
 & 256 & 1.39 & 4.70 & 1.00 & 9.13 & 4.10 & 1.30 & 1.32 & 4.70 & 1.10 & 8.57 & 3.40 & 0.60 \\ 
 & 512 & 1.30 & 3.90 & 0.70 & 9.22 & 4.50 & 1.00 & 1.32 & 4.50 & 0.90 & 9.13 & 5.80 & 1.40 \\ 
 & 1024 & 1.43 & 5.00 & 0.90 & 9.57 & 4.40 & 0.80 & 1.34 & 4.20 & 1.20 & 9.37 & 4.90 & 0.70 \\ 
 \hline
(h), $t$%
& 64 & 9.07 & 77.10 & 49.00 & 18.95 & 53.10 & 29.10 & 8.16 & 69.30 & 41.30 & 16.82 & 43.50 & 22.10 \\ 
& 128 & 17.21 & 98.30 & 92.80 & 28.78 & 91.10 & 74.30 & 16.47 & 98.00 & 89.70 & 26.52 & 87.00 & 66.70 \\ 
& 256 & 31.12 & 100.00 & 99.90 & 45.94 & 100.00 & 99.70 & 30.12 & 100.00 & 99.70 & 43.54 & 99.70 & 98.60 \\ 
& 512 & 57.81 & 100.00 & 100.00 & 78.95 & 100.00 & 100.00 & 53.14 & 100.00 & 100.00 & 71.69 & 100.00 & 100.00 \\ 
& 1024 & 112.95 & 100.00 & 100.00 & 146.21 & 100.00 & 100.00 & 98.88 & 100.00 & 100.00 & 127.16 & 100.00 & 100.00 \\ 
\hline 
(h), $\beta$ %
&64 & 9.17 & 77.80 & 49.60 & 18.30 & 50.00 & 28.10 & 8.20 & 69.40 & 40.20 & 16.86 & 42.60 & 21.60 \\ 
& 128 & 17.49 & 98.40 & 91.40 & 29.06 & 91.30 & 74.10 & 16.31 & 97.20 & 88.50 & 27.21 & 87.00 & 67.70 \\ 
& 256 & 31.05 & 100.00 & 100.00 & 46.75 & 99.90 & 99.30 & 29.58 & 100.00 & 100.00 & 44.06 & 99.90 & 98.60 \\ 
& 512 & 57.90 & 100.00 & 100.00 & 78.97 & 100.00 & 100.00 & 52.97 & 100.00 & 100.00 & 71.40 & 100.00 & 100.00 \\ 
 &1024 & 114.13 & 100.00 & 100.00 & 146.75 & 100.00 & 100.00 & 100.95 & 100.00 & 100.00 & 128.45 & 100.00 & 100.00 \\ 
\hline
\end{tabular}
}
\caption{Median of test statistic values and rejection rates of $\hat{Q}_{M,u}^{(T)}$ and $\hat{Q}_{M,s}^{(T)}$ at the 1\% and 5\% asymptotic level for the processes (g) and (h), where $t$ and $\beta$ indicate $t_{19}$- and $\beta(6,6)$-distributed innovations, respectively. 
 All table entries are generated from 1000 repetitions.}
\label{tab:fourth}
\end{center}
\end{table}
It can be seen from the results in Table \ref{tab:fourth} that the proposed procedures perform roughly as expected. First, under the null hypothesis for levels for both types of innovations, both sets of tests and both choices of $M$ are well adjusted and observe similar patterns as their normal counterparts in DGP (b) in Table \ref{tab:null}. Second, under the alternative for process (h), powers align roughly as for the Gaussian case in Table \ref{tab:alt}. 
Overall, the simulation results reveal that the estimation of the fourth-order spectrum does not lead to a marked decay in performance.

\subsection{Application to annual temperature curves}
\label{subsec:emirical:data}

To give an instructive data example, the proposed method was applied to annual temperature curves recorded at several measuring stations across Australia over the last century and a half. The exact locations and lengths of the functional time series are reported in Table \ref{tab:data}, and the annual temperature profiles recorded at the Gayndah station are displayed for illustration in the left panel of Figure \ref{fig:data0}. To test whether these annual temperature profiles constitute stationary functional time series or not, the proposed testing method was utilized, using specifications similar to those in the simulation study. 
\textcolor{black}{To get an idea of the spectral structure of these different temperature curves, the left-hand side of Figure \ref{fig:dataL} shows the averaged eigenvalue decay standardized with respect to the largest eigenvalue at each frequency. More precisely, $\frac 1T\sum_{j=1}^T\lambda_l^{\omega_j}/\lambda_1^{\omega_j}$ is plotted against $l$. Figure \ref{fig:eig_l_freq} displays in addition the plots of the $l$-largest sample eigenvalues $\lambda_l^{\omega_j}$ against $j$ for $l=1,\ldots,15$. It can be seen that frequency-specific contributions are heterogeneous for each of the four stations. There are also substantial differences in the eigenvalue plots across different stations. The choices of $L_j$ across frequency $\omega_j$ as used by the standardized test procedure are shown in the right-hand side of Figure \ref{fig:dataL}. }

\textcolor{black}{
The $p$-values for the standardized test statistics are essentially zero for all stations and all $M=1,\ldots,5$. The testing results for the unstandardized statistics are summarized in Table \ref{tab:data}. Stationarity is rejected in favor of the alternative at the 1\% significance level at all measuring stations for $\hat Q_{M,u}^{(T)}$ with all specifications of $M$, with one notable exception: no choice of $M$ leads to a rejection of the null hypothesis at Boulia station. Additionally, rejection at Melbourne and Sydney stations is not possible at the smallest significance levels for several $M$. At all other measuring stations rejection of the null is very strong. Note that Boulia station showed the slowest eigenvalue decay in Figure \ref{fig:dataL} and the spectral behavior most separated from the other stations in Figure \ref{fig:eig_l_freq}. It is particularly interesting that around frequency $\pi$ there is little to no separation between first and second sample eigenvalues. The lack of estimation accuracy in the case of tied eigenvalues might help explain why Boulia station delivers results at odds with the findings at the other stations. In the future, it might be worthwhile looking into running the stationarity tests only in certain frequency bands, excluding those frequencies for which separation of sample eigenvalues is not sufficiently large. This is, however, beyond the scope of the current paper.
} 

\begin{table}[h]
\vspace{.6cm}
\begin{center}
\begin{tabular}{l@{\qquad}c@{\qquad}r@{\qquad}r@{\qquad}r@{\qquad}r@{\qquad}r}
\hline
Station\phantom{$\hat{\hat\hat{k^*}}k_{n_n}$} & $T$ & $M=1$ & $M=2$ & $M=3$ & $M=4$ & $M=5$  \\
\hline
Boulia &120 &  
0.71 & 0.17 & 0.20 & 0.36 & 0.44 \\ 
Robe & 130 & 
0.01 & 0.00 & 0.00 & 0.00 & 0.00 \\
Cape Otway & 150 & 
0.00 & 0.00 & 0.00 & 0.00 & 0.00 \\
Gayndah & 118 & 
0.00 & 0.00 & 0.00 & 0.00 & 0.00 \\
Gunnedah & 134 & 
0.00 & 0.00 & 0.00 & 0.00 & 0.00 \\
Hobart& 122 & 
0.00 & 0.00 & 0.00 & 0.00 & 0.00 \\
Melbourne & 158 & 
0.03 & 0.04 & 0.02 & 0.01 & 0.01 \\
Sydney & 154 & 
0.15 & 0.01 & 0.00 & 0.00 & 0.00
\\ \hline
\end{tabular}
\caption{\label{tab:data} Summary of results for eight Australian measuring stations. The column labeled $T$ reports the sample size, the other columns report $p$-values for the given choices of $M$ for $\hat Q_{M,u}^{(T)}$. 
}
\end{center}
\end{table}

\begin{figure}[h]
\centering
\vspace{-.5cm}
\includegraphics[height=0.6\textwidth]{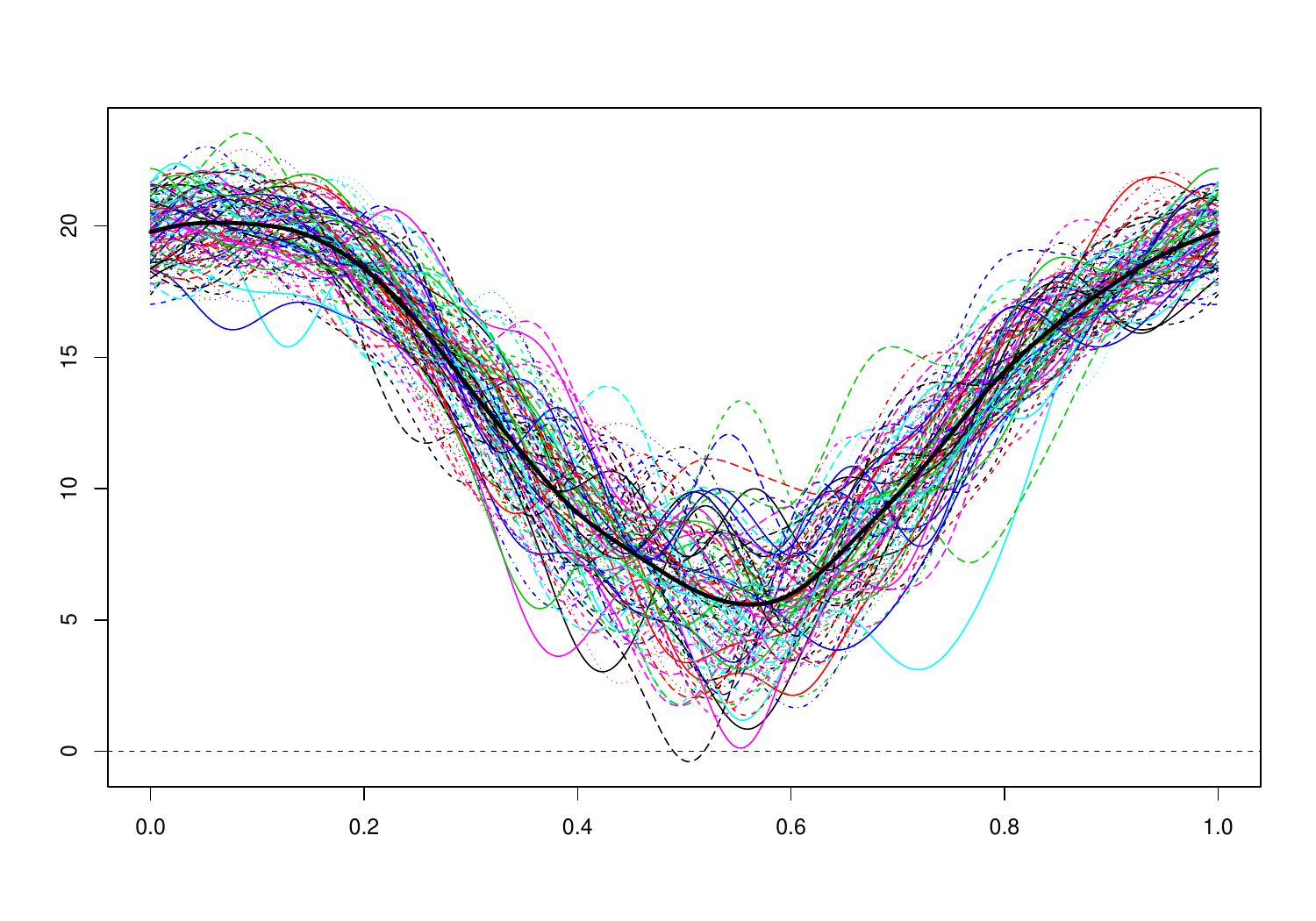} 
\caption{Annual temperature curves at Gayndah station.}
\label{fig:data0}
\end{figure}

\begin{figure}[h]
\centering
\vspace{-.3cm}
\includegraphics[width=0.4\textwidth]{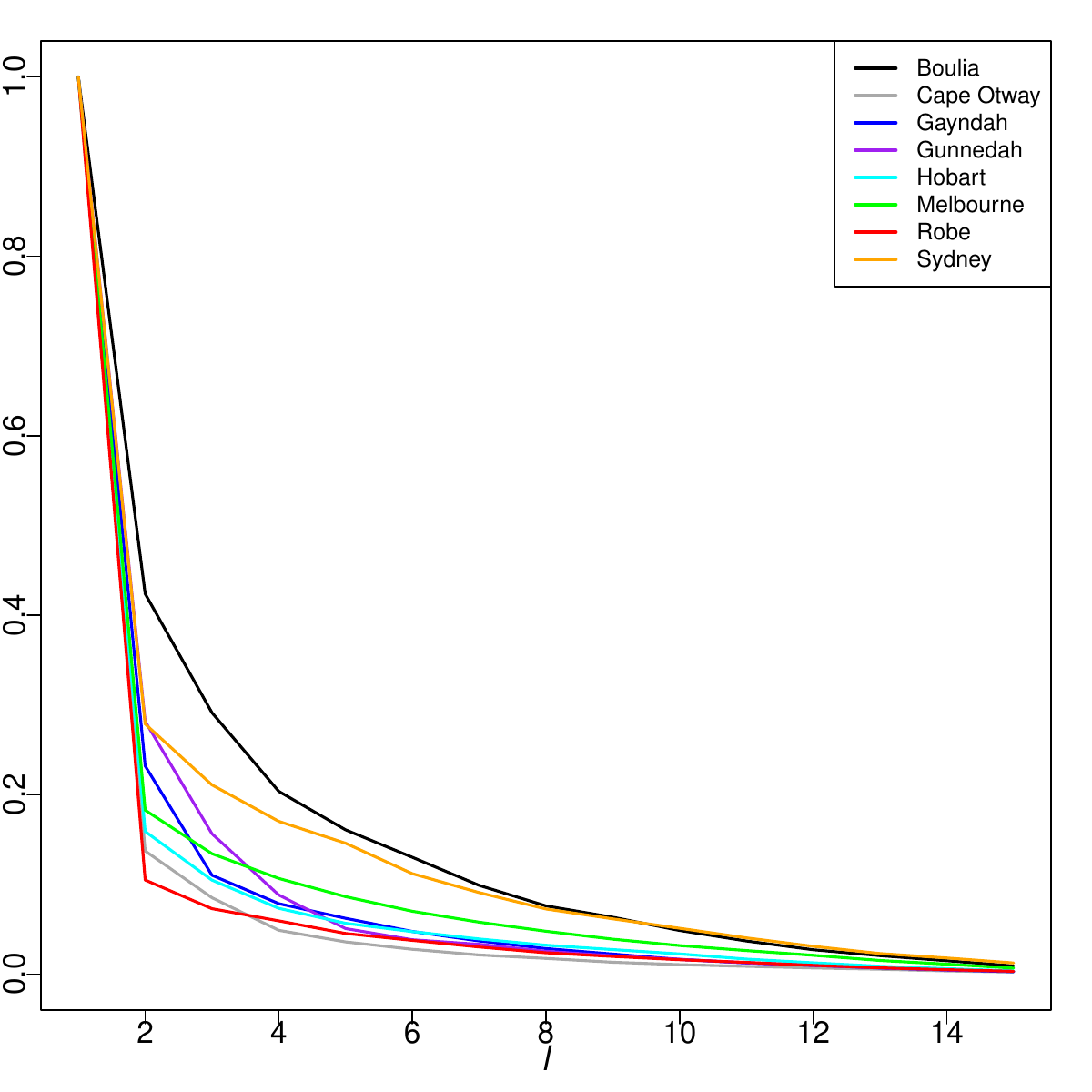}\quad
\includegraphics[width=0.4\textwidth]{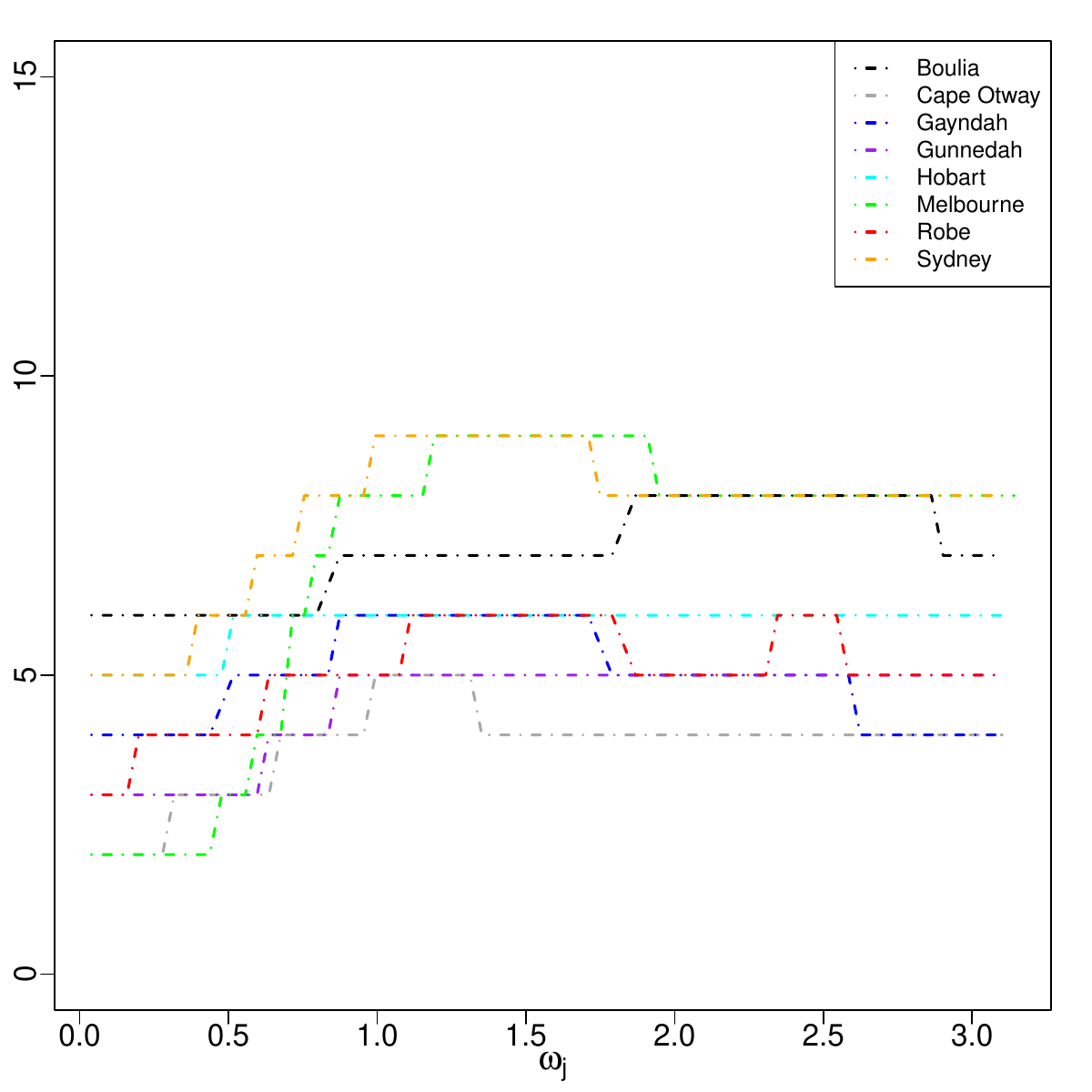} 
\caption{Average eigenvalue decay standardized with respect to the largest eigenvalue at each frequency, $\frac 1T\sum_{j=1}^T\lambda_l^{\omega_j}/\lambda_1^{\omega_j}$ (left) and truncation level $L_j$ across $\omega_j$ (right) across different measuring stations. }
\label{fig:dataL}
\end{figure}

\begin{figure}[h]
\centering
\includegraphics[width=0.4\textwidth]{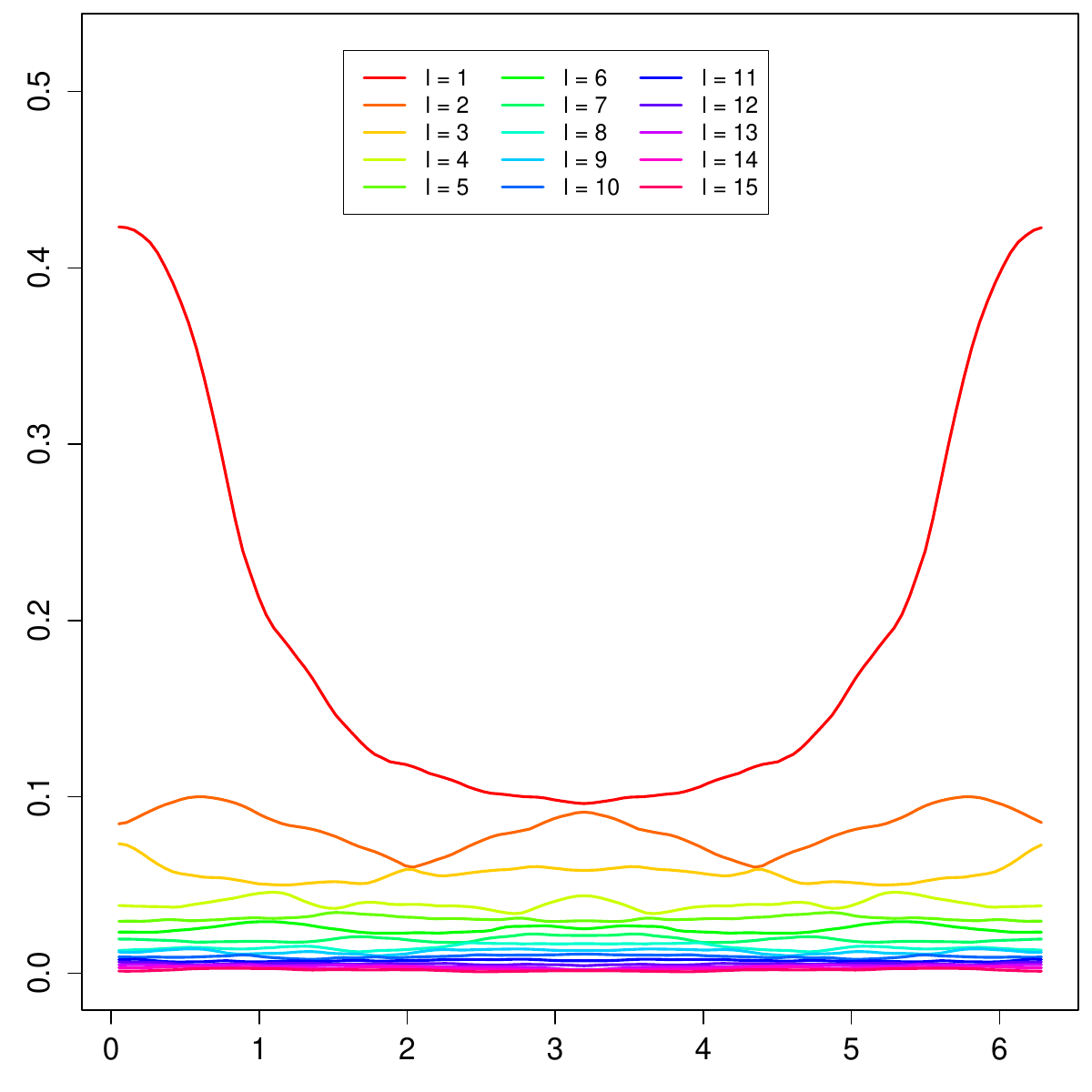}\qquad
\includegraphics[width=0.4\textwidth]{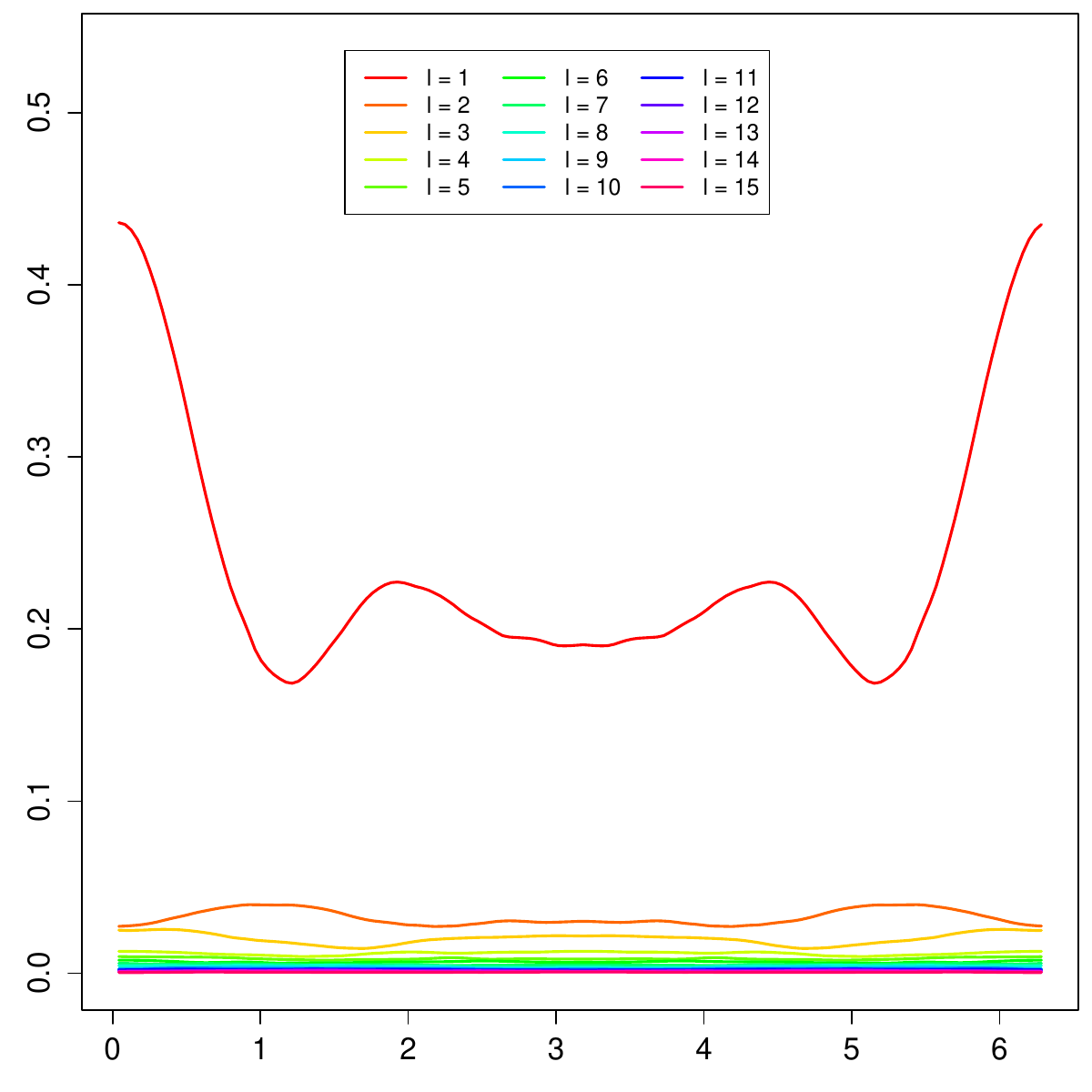} \\[.2cm]
\includegraphics[width=0.4\textwidth]{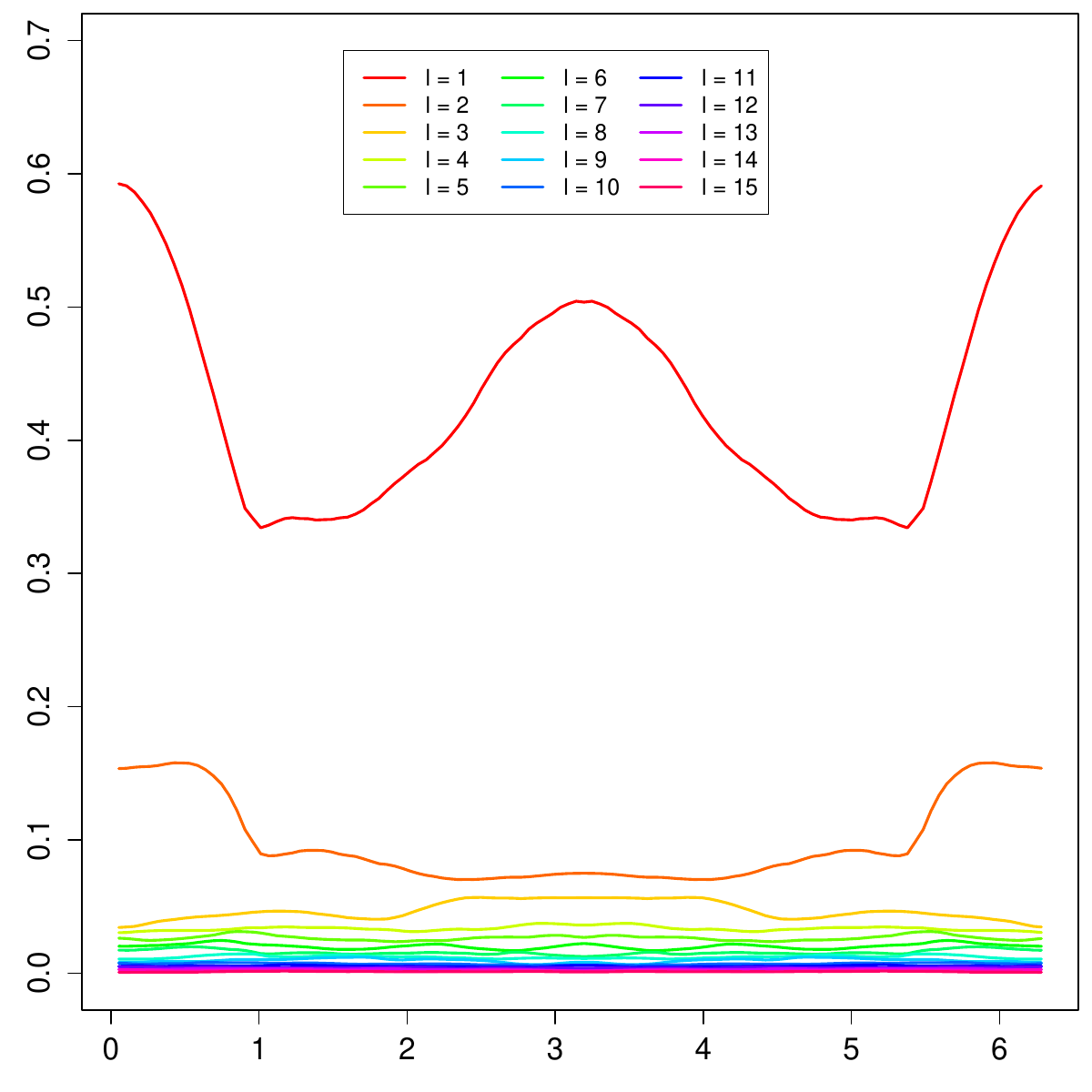}\qquad
\includegraphics[width=0.4\textwidth]{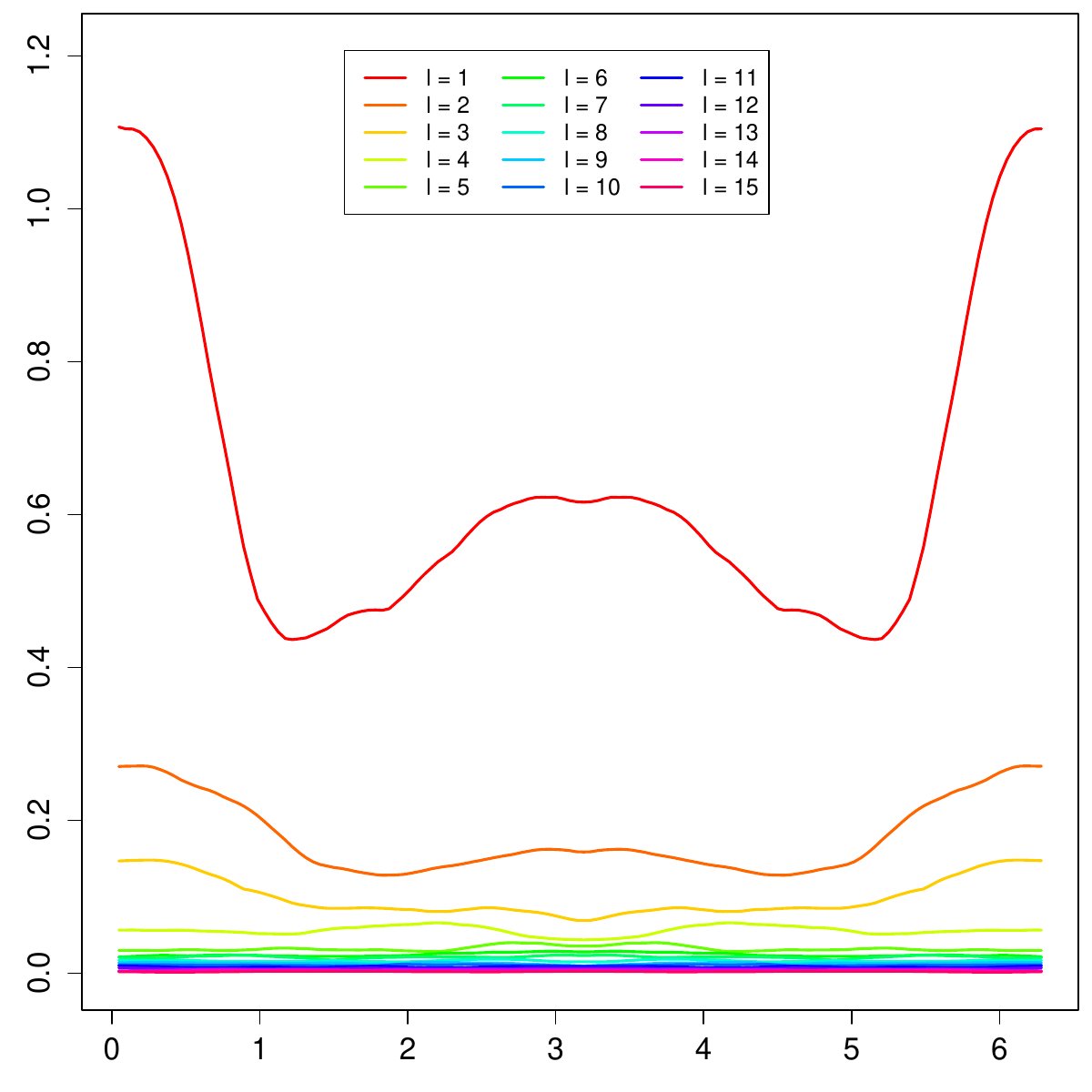} 
\caption{Plots of the 15 largest sample sample eigenvalues across Fourier frequencies at Boulia (top left), Cape Otway (top right), Gayndah (bottom left) and Gunnedah (bottom right) stations.}
\label{fig:eig_l_freq}
\end{figure}

\section{Conclusions and future work}
\label{sec:conclusions}

In this paper methodology for testing the stationarity of a functional time series is put forward. The tests are based on frequency domain analysis and exploit that fDFTs at different canonical frequencies are uncorrelated if and only if the underlying functional time series are stationary. The limit distribution of the quadratic form-type test statistics has been determined under the null hypothesis as well as under the alternative of local stationarity. Finite sample properties were highlighted in simulation experiments with various data generating processes and an application to annual temperature profiles. 

The empirical results show promise for further applications to real data, but future research has to be devoted to a further fine-tuning of the proposed method; for example, an automated selection of frequencies $h_m$ outside of the standard choice $h_m=m$ for all $m=1,\ldots,M$. This can be approached through a more refined analysis of the size of the various ${\hat{\beta}}^{(T)}_{h_m,x}$ in \eqref{eq:weighted} whose real and imaginary part make up the vector $\hat{\mathbf{b}}_{M,x}^{(T)}$ in the test statistics $\hat Q_{M,x}^{(T)}$.


\appendix

\section{A functional Cram\'er representation}
\label{subsec:proof:cramer}
\begin{proof}[Proof of Proposition \ref{start}]
Let $(X_t\colon t \in\znum)$ be a zero-mean, weakly stationary $H$-valued stochastic process. It has been shown \citep[][Thm 4.4]{vde18} that for processes with a trace-class spectral measure $F$,
there exists an isomorphic mapping between the subspaces $\overline{\text{sp}}(X_t \colon t \in \znum)$ of $L^2_\mathbb{C}(\Omega)$ and $\overline{\text{sp}}(e^{\im t \cdot} \colon t \in \znum)$ of $L^2([-\pi,\pi], \mathscr{B},\mu_F)$. As a consequence, $X$ admits the representation in \eqref{cramrep}. 
Conversely, 
\begin{align*} 
\cov(X_t,X_s) 
= \E\bigg[ \int_{-\pi}^{\pi} e^{\im t \lambda_1} d Z_{\lambda_1} \otimes  \int_{-\pi}^{\pi} e^{\im s \lambda_2} d Z_{\lambda_2}\bigg] 
 =  \int_{-\pi}^{\pi} e^{\im (t-s) \lambda} d \mu_F = \mathcal{C}_{t-s},
 \end{align*}
showing that a process that admits representation \eqref{cramrep} must be weakly stationary.
\end{proof}

\section{Properties of functional cumulants}
\label{subsec:proof:cumulants}

For random elements $X_1, \ldots, X_k$ in a Hilbert space $H$, the {\em moment tensor of order $k$} can be defined as
\begin{align*}
\E \big[ X_1 \otimes \cdots \otimes X_k\big] 
= \sum_{l_1, \ldots l_k  \in \nnum} \E \Big[ \prod_{t=1}^{k} \langle X_t, \psi_{l_t} \rangle \Big] 
(\psi_{l_1} \otimes \cdots \otimes \psi_{l_k}),  
\end{align*}   
where the elementary tensors $(\psi_{l_1} \otimes \cdots \otimes \psi_{l_k} \colon l_1, \ldots, l_k \in \nnum)$ form an orthonormal basis in the tensor product space $\bigotimes_{j=1}^{k} H$ if $(\psi_l\colon l \in \nnum)$ is an orthornormal basis of the separable Hilbert space $H$. Similarly, define the {\em $k$-th order cumulant tensor} by
\begin{align}
\label{cumtens}
\cum \big( X_1, \ldots, X_k\big) 
= \sum_{l_1, \ldots l_k\in \nnum} \cum \big( \inprod{X_1}{\psi_{l_1}},\ldots,\inprod{X_k}{\psi_{l_k}}\big)
 (\psi_{l_1} \otimes \cdots \otimes \psi_{l_k}),
\end{align}   
where the cumulants on the right-hand side are as usual given by
\[
\cum\big(\inprod{X_1}{\psi_{l_1}},\ldots,\inprod{X_k}{\psi_{l_k}}\big)
=\sum_{\nu=(\nu_1,\ldots,\nu_p)}(-1)^{p-1}\,(p-1)!\,\prod_{r=1}^p \E\Big[\prod_{t\in\nu_r}\inprod{X_t}{\psi_{l_t}}\Big],
\]
the summation extending over all unordered partitions $\nu$ of $\{1,\ldots,k\}$. The following is a generalization of the product theorem for cumulants \citep[Theorem 2.3.2]{b81}.

\begin{theorem} 
\label{prodcumthm} 
Consider the tensor $X_t=\otimes_{j=1}^{J_t}X_{tj}$ for random elements $X_{tj}$ in $H$ with $j=1,\ldots,J_t$ and $t=1,\ldots,k$. Let $\nu = \{\nu_1,\ldots, \nu_p\}$ be a partition of $\{1,\ldots, k\}$. The joint cumulant tensor 
is given by
\[
\cum(X_1,\ldots,X_k)
=\sum_{r_{11},\ldots,r_{kJ_t}}\sum_{\nu=(\nu_1,\ldots,\nu_p)}\prod_{n=1}^{p}
\cum\big(\inprod{X_{tj}}{\psi_{r_{tj}}}|(t,j)\in \nu_n\big)\,
\psi_{r_{11}}\otimes\cdots\otimes\psi_{r_{kJ_t}}, 
\]
where the summation extends over all indecomposable partitions $\nu=(\nu_1,\ldots,\nu_p)$ of the table
\[
\begin{matrix}
(1,1) &\cdots& (1,J_1)\\
\vdots&\ddots& \vdots\\
(k,1) &\cdots& (k,J_t).
\end{matrix}
\]
\end{theorem}

Formally, abbreviate this by
\[
\cum(X_1,\ldots,X_k)
=\sum_{\nu=(\nu_1,\ldots,\nu_p)}S_\nu\Big(\otimes_{n=1}^{p}
\cum\big(X_{tj}|(t,j)\in\nu_n\big)\Big),
\]
where $S_\nu$ is the permutation that maps the components of the tensor back into the original order, that is, $S_\nu(\otimes_{r=1}^p\otimes_{(t,j)\in\nu_r} X_{tj})=X_{11}\otimes\cdots\otimes X_{kJ_t}$. 

Next, expressions and bounds for cumulants of the fDFT are given in both locally stationary and stationary regimes.

\begin{lemma}[{\textbf{Cumulants of the fDFT under local stationarity}}]
\label{cumboundglsp}
Let $(X_{t,T}\colon t\leq T, T\in\mathbb{N})$ be a $k$-th order locally stationary process in $H$ satisfying \ref{cumglsp}($k$,\,1) for arbitrary fixed $k$. The cumulant tensor of the local fDFT satisfies
\begin{align}
\label{eq:cumboundglsp}
\cum\big(D^{(T)}_{\omega_{j_1}},\ldots, D^{(T)}_{\omega_{j_k}}\big) 
&= \frac{(2\pi)^{k/2-1}}{T^{k/2}}\sum_{t=0}^{T-1} \Floc{{t}/{T}}{k}{j}e^{-\im \sum_{l=1}^{k}t \omega_{j_l}} +R_{k,T}\\
&= \frac{(2\pi)^{k/2-1}}{T^{k/2-1}} \Fourloc{k}{j} + R_{k,T},  \nonumber
\end{align}
where $\|R_{k,T}\|_2= O({T^{-k/2}})$  and the operator 
\begin{align}
\label{eq:FCbounds2}
\Fourlocs{k}{j}{s}= \int_0^{1}\Floc{u}{k}{j}e^{-\im 2\pi s u} du
\end{align}
denotes the $s$-th Fourier coefficient of $\Floc{u}{k}{j}$ and belongs to $S_2$. 
\end{lemma}
The proof can be found in Section \ref{subsec:proof:OScumulant} of the Online Supplement. 
 Lemma \ref{cumboundglsp} provides a relation between the $k$-th order cumulant tensor of the local fDFT and the 
Fourier coefficients of the $k$-th order time-varying spectral density tensors, which induce Hilbert--Schmidt operators. The proof of \eqref{eq:FCbounds2} makes apparent that the dependence structure of the local fDFT behaves in a very specific manner that is based on the distance of the frequencies. The Fourier coefficients additionally provide an upper bound on the norm of the cumulant operator. 
\begin{cor}
\label{cumbounds}
If \ref{cumglsp}($k$,\,2) holds for arbitrary fixed $k$, then
\begin{enumerate}
\item[(i)] $\begin{aligned}[t] 
\snorm{\cum(D^{(T)}_{\omega_{j_1}},\ldots, D^{(T)}_{\omega_{j_k}})}_2 \le \frac{C}{T^{k/2-1}|j_1+\cdots+j_k|^2}
+ O\bigg(\frac{1}{T^{k/2}}\bigg); 
\end{aligned}$ 
\item[(ii)] $\begin{aligned}[t]
\sup_{\omega} \sum_{s \in \znum} \snorm{\tilde{\mathcal{F}}_{s;\omega} }_2 \le \infty. 
\end{aligned}$
\end{enumerate}
\end{cor}
Note that if $\sum_{l=1}^{k}\omega_{j_l}=0 \mod 2\pi$, then \eqref{eq:cumboundglsp} yields approximately a time average of the $k$-th order time-varying spectral density tensor. In case the process does not depend on time $u$, $\Fourlocs{2k}{j}{s} = O_H$ for $s \ne 0$. That is, the operator $\Fourlocs{2k}{j}{s}$ maps any $\psi \in L^2([0,1]^{k},\mathbb{C})$ to the origin for $s \ne 0$. Consequently, under $k$-th order stationarity the following corollary holds.

\begin{cor}[{\textbf{Cumulants of the fDFT under stationarity}}]
\label{FCboundssd}
Let $(X_t\colon t\in\mathbb{Z})$ be a $k$-th order stationary sequence taking values in $H_\rnum$ that satisfies \ref{Statcase}($k$,1) for arbitrary fixed $k$. Then the cumulant tensor of the fDFT satisfies
\begin{align}
\label{eq:cumDFTSstat}
\cum\big(D^{(T)}_{\omega_{j_1}},\ldots, D^{(T)}_{\omega_{j_k}}\big) 
= \frac{(2\pi)^{k/2-1}}{T^{k/2}} \Delta_T^{( \sum_{l=1}^{k}\omega_{j_l})} \Fopkatfour{k}{j} + R_{T,k} ,
\end{align}
where the function $\Delta_T^{(\omega)} = T$ for $\omega \equiv 0 \!\mod 2 \pi$, $\Delta_T^{(\omega_j)} = 0$ for $j \not\equiv 0 \!\mod T$ and the remainder satisfies $\snorm{R_{T,k}}_2=O(T^{-k/2})$. 
\end{cor}


\section{First- and second-order dependence structure}
\label{sec:proof:first+second}

\subsection{Expectation}
\label{sec:proof:first}
From Lemma \ref{cumboundglsp}, for $h \ne 0\mod T$,
\begin{align*}
\Bigsnorm{\frac{1}{T}\sum_{j=1}^{T}\E \big(D^{(T)}_{\omega_{j}} \otimes D^{(T)}_{\omega_{j+h}}\big)}_2 =
\Bigsnorm{\frac{1}{T}\sum_{j=1}^{T}\frac{1}{T}\sum_{t=1}^{T} \F_{t/T;\omega_{j}}e^{-\im t\omega_h}+R_{T,2}}_2 
=
\begin{cases}
O(T^{-1}) & \mbox{under } H_0.\\ 
O(h^{-2}+T^{-1}) & \mbox{under } H_A.
\end{cases}
\end{align*}
In particular, using that the operator-valued functions $(u,\omega) \mapsto \F_{u,\omega}$ are Lipschitz continuous in $(u,\omega)$, yields that, under $H_A$,
\begin{align*}
\frac{1}{T}\sum_{j=1}^{T}\frac{1}{T}\sum_{t=1}^{T} \F_{t/T;\omega_{j}}e^{-\im t\omega_h} \to   \frac{1}{2\pi} \int_0^{2\pi} \int_{0}^{1} \F_{u,\omega}e^{-\im 2\pi u h}du d\omega,
\end{align*}
where the convergence is in $S_2(H)$. 
Since $\E[\|D^{(T)}_{\omega}\|^2_2] <\infty$, 
the Cauchy--Schwarz inequality implies Fubini's theorem can be applied. Together with the above, it follows that the expectation of $\beta^{(T)}_{h,u}$ satisfies
\begin{align*}
\E[\beta^{(T)}_{h,u}] &
= \frac{1}{T} \sum_{j=1}^{T}  \sum_{l=1}^{L(\omega_j)}  \sum_{l^\prime=1}^{L(\omega_{j+h})} \inprod{ \E\big(D^{(T)}_{\omega_{j}} \otimes D^{(T)}_{\omega_{j+h}}\big)}{\tilde{\phi}^{\omega_j}_{l} \otimes \tilde{\phi}^{\omega_{j+h}}_{\lpr}}_{S}
\\& =
\frac{1}{T} \sum_{j=1}^{T}  \sum_{l=1}^{L(\omega_j)}  \sum_{l^\prime=1}^{L(\omega_{j+h})}  \inprod{\frac{1}{T}\sum_{t=0}^{T-1} {\F}_{{t}/{T},\omega_{j}} e^{-\im t \omega_{h}}+R_{2,T}}{\tilde{\phi}^{\omega_j}_{l} \otimes \tilde{\phi}^{\omega_{j+h}}_{\lpr}}_{S} =O\bigg(\frac{1}{h^{2}}\bigg)+O\bigg(\frac{1}{T}\bigg)
\\& \to \frac{1}{2\pi} \int_0^{2\pi} \int_{0}^{1}  \sum_{l=1}^{L(\omega)}  \sum_{l^\prime=1}^{L(\omega+\omega_{h})}  \inprod{\F_{u,\omega}e^{-\im 2\pi u h}}{\tilde{\phi}^{\omega}_{l} \otimes \tilde{\phi}^{\omega+\omega_{h}}_{\lpr}}_{S} du d\omega,
\end{align*}
where the stated order $O(\cdot)$ for the projections follows from the previously stated convergence in norm. 
Similarly, 
\[
\E[\beta^{(T)}_{h,s}]
\to  \frac{1}{2\pi} \int_0^{2\pi} \int_{0}^{1}  \sum_{l=1}^{L(\omega)}  \sum_{l^\prime=1}^{L(\omega+\omega_{h})} \frac{ \inprod{\F_{u,\omega} (\tilde{\phi}^{\omega+\omega_{h}}_\lpr)}{  \tilde{\phi}^{\omega}_l  }e^{-\im 2\pi u h}  }{\sqrt{{ \tilde{\lambda}_{l}^{\omega} \tilde{\lambda}_{\lpr}^{\omega+\omega_{h}}}}} du d\omega.
\]
under Condition $C_s$. 

\subsection{Covariance structure}\label{subsec:cov-struc}
Theorem \ref{prodcumthm} implies that the covariance structure of the cross-periodogram operators is given by 
\begin{align}
\label{eq:covariancep}
\Cov \big(D^{(T)}_{\omega_{j_1}}\otimes D^{(T)}_{\omega_{j_1+h_1}}, D^{(T)}_{\omega_{j_2}}\otimes D^{(T)}_{\omega_{j_2+h_2}} \big)
=\,&\cum(D^{(T)}_{\omega_{j_1}},D^{(T)}_{-\omega_{j_1+h_1}},
D^{(T)}_{-\omega_{j_2}}, D^{(T)}_{\omega_{j_2+h_2}})
\\&
\nonumber
+S_{1324}\big(
\cum(D^{(T)}_{\omega_{j_1}},D^{(T)}_{-\omega_{j_2}})\otimes
\cum(D^{(T)}_{-\omega_{j_1+h_1}},D^{(T)}_{\omega_{j_2+h_2}})\big)\\&
+S_{1423}\big(
\cum(D^{(T)}_{\omega_{j_1}},D^{(T)}_{\omega_{j_2+h_2}})\otimes
\cum(D^{(T)}_{-\omega_{j_1+h_1}},D^{(T)}_{-\omega_{j_2}})\big),
\nonumber
\end{align}
where $S_{ijkl}$ denotes the permutation operator on $\otimes_{i=1}^4 L^2_\cnum([0,1])$ that permutes the components of a tensor according to the permutation $(1,2,3,4)\mapsto(i,j,k,l)$, that is, $S_{ijkl}(x_1\otimes\cdots\otimes x_4)=x_i\otimes\cdots\otimes x_l$. Under \ref{cumglsp}(4,2), we obtain from Lemma \ref{cumboundglsp},
\begin{align*}
\tageq \label{eq:covHa}
\Cov &\Bigg( \frac{1}{\sqrt{T}}\sum_{j_1}^{T}D^{(T)}_{\omega_{j_1}}\otimes D^{(T)}_{\omega_{j_1+h_1}},\frac{1}{\sqrt{T}}\sum_{j_2}^{T}D^{(T)}_{\omega_{j_2}}\otimes D^{(T)}_{\omega_{j_2+h_2}} \Bigg)\\
=& \frac{1}{T}\sum_{j_1,j_2}^{T}\Bigg(  \frac{2\pi)}{T^2}\sum_t \F_{t/T:\omega_{j_1},-\omega_{j_1+h_1},-\omega_{j_2}}  e^{-\im t(\omega_{h_2}-\omega_{h_1})}+R_{T,4}\bigg)\\&
+S_{1324}\Bigg(\big( \frac{1}{T}\sum_t \F_{t/T:\omega_{j_1}}e^{-\im t(\omega_{j_1}-\omega_{j_2})}  + R_{T,2}\big) \otimes  \big( \frac{1}{T}\sum_t \F_{t/T:-\omega_{j_1+h_1}}e^{-\im t(-\omega_{j_1+h_1}+\omega_{j_2+h_2}) } + R_{T,2}\big) \Bigg)
\\&+S_{1423}\Bigg( \big( \frac{1}{T}\sum_t \F_{t/T:\omega_{j_1}}e^{-\im t (\omega_{j_1}+\omega_{j_2+h_2})} +R_{T,2}\big) \otimes \big(\frac{1}{T}\sum_t \F_{t/T:-\omega_{j_1+h_1}}e^{-\im t(-\omega_{j_1+h_1}-\omega_{j_2})} + R_{T,2}\big)\Bigg).  
\end{align*}
Using Minkowski's inequality and Corollary \ref{cumbounds}(ii) it follows that, for all $T, h_1, h_2$, 
\begin{align*}\biggsnorm{\Cov \big( \frac{1}{\sqrt{T}}\sum_{j_1}^{T}D^{(T)}_{\omega_{j_1}}\otimes D^{(T)}_{\omega_{j_1+h_1}},\frac{1}{\sqrt{T}}\sum_{j_2}^{T}D^{(T)}_{\omega_{j_2}}\otimes D^{(T)}_{\omega_{j_2+h_2}} \big)}_2=O(1), \tageq \label{eq:OrderVar}
\end{align*}
both under $H_A$ and $H_0$. The focus is here on the covariance structure under fourth-order stationarity. The more general expression is derived in Section \ref{sub:covHA} of the Online Supplement.

\begin{proof}[Proof of Theorem \ref{momentstatcase}]
Under \ref{Statcase}(4,2), Corollary \ref{FCboundssd} implies that \eqref{eq:covHa} becomes
\begin{align*}
\Cov \bigg( \frac{1}{\sqrt{T}}\sum_{j_1}^{T}&D^{(T)}_{\omega_{j_1}}\otimes D^{(T)}_{\omega_{j_1+h_1}},\frac{1}{\sqrt{T}}\sum_{j_2}^{T}D^{(T)}_{\omega_{j_2}}\otimes D^{(T)}_{\omega_{j_2+h_2}} \bigg)\\
=& \frac{1}{T}\sum_{j_1,j_2}^{T}\Bigg(  \frac{(2\pi)}{T^{2}} \F_{\omega_{j_1},-\omega_{j_1+h_1},-\omega_{j_2}}\Delta_T^{(\omega_{h_2}-\omega_{h_1})} +R_{T,4}\bigg)\\&
+S_{1324}\Bigg(\Big(\F_{\omega_{j_1}}\frac{1}{T} \Delta_T^{(\omega_{j_1}-\omega_{j_2})} + R_{T,2}\Big) \otimes  \Big(\F_{-\omega_{j_1+h_1}}\frac{1}{T} \Delta_T^{(-\omega_{j_1+h_1}+\omega_{j_2+h_2})} + R_{T,2}\Big) \Bigg)
\\&+S_{1423}\Bigg( \Big(\F_{\omega_{j_1}}\frac{1}{T} \Delta_T^{(\omega_{j_1}+\omega_{j_2+h_2})} +R_{T,2}\Big) \otimes \Big(\F_{-\omega_{j_1+h_1}}\frac{1}{T}\Delta_T^{(-\omega_{j_1+h_1}-\omega_{j_2})} + R_{T,2}\Big)\Bigg). 
\end{align*}
By the properties of $\Delta^{(\cdot)}_T$, the term on the second line is of lower order unless $h_1-h_2=0 \mod T$, while the third line requires $j_1-j_2=0 \mod T$ and $h_1-h_2=0 \mod T$. For the fourth line to not be of lower order we require  $j_1+j_2+h_2=0 \mod T$ and $-j_1-h_1-j_2 =0\mod T$ which give the constraints $j_1+j_2=T-h_2$ and $j_1+j_2=T-h_1$, implying we must have $j_1+j_2=T-h$. It follows therefore that the covariance is of order $O(T^{-1})$ in Hilbert--Schmidt norm if  $h_1-h_2 \ne 0 \mod T$. If $h_1-h_2 =0 \mod T$, then
\begin{align*}
\Cov &\bigg( \frac{1}{\sqrt{T}}\sum_{j_1}^{T}D^{(T)}_{\omega_{j_1}}\otimes D^{(T)}_{\omega_{j_1+h}}, \frac{1}{\sqrt{T}}\sum_{j_2}^{T}D^{(T)}_{\omega_{j_2}}\otimes D^{(T)}_{\omega_{j_2+h}} \bigg)= \frac{1}{T}\sum_{j_1,j_2}^{T} \frac{(2\pi)}{T} \F_{\omega_{j_1},-\omega_{j_1+h},-\omega_{j_2}} +R_{T,2} \\&
+ \frac{1}{T}\sum_{j_1}^{T}\Bigg( \big(\F_{\omega_{j_1}} + R_{T,2}\big) \widetilde{\otimes}\big(\F_{\omega_{j_1+h}}+ R_{T,2}\big)
+\big(\F_{\omega_{j_1}}+R_{T,2}\big) \widetilde{\otimes}_\top \big(\F_{\omega_{j_1+h_1}}+ R_{T,2}\big)\Bigg),\tageq \label{eq:covFun2}
\end{align*}
where Definition \ref{def:Krn} was used. Thus, as $T \to \infty$, this converges in norm to
\begin{align*}
 \frac{1}{4\pi} \int \int \F_{\omega,-\omega-\omega_h,-\omega^\prime} d\omega d\omega^\prime +\int \F_{\omega} \widetilde{\otimes} \F_{\omega+\omega_h} 
+\F_{\omega}  \widetilde{\otimes}_\top \F_{\omega+\omega_h} d\omega.
\end{align*}
Consider then the covariance structure of $\tprojes_{h,u}$, which is obtained by projecting the fDFT onto the eigenfunctions of $\F_\omega$. Write this covariance structure as 
\begin{align*}
&\Cov(\sqrt{T}{\bm{\beta}}_{h_1,u}^{(T)},\sqrt{T}{\bm{\beta}}_{h_2,u}^{(T)}) =\\&  
\frac{1}{T}\sum_{j_1,j_2}^{T}  \sum_{\small{\substack{l_1 \in [L(\omega_{j_1})], l_2 \in [L(\omega_{j_1+h_1})], \\ l_3 \in [L(\omega_{j_2})] , l_4 \in [L(\omega_{j_2+h_2})]}}}  \Big\langle \Cov\Big(D^{(T)}_{\omega_{j_1}}\otimes D^{(T)}_{\omega_{j_1+h}}, D^{(T)}_{\omega_{j_2}}\otimes D^{(T)}_{\omega_{j_2+h}}\Big) \big(\phi^{\omega_{j_2}}_{l_{3}} \otimes \phi^{\omega_{j_2+h_2}}_{l_4}\big),\phi^{\omega_{j_1}}_{l_{1}} \otimes \phi^{\omega_{j_1+h_1}}_{l_2} \Big\rangle.
\end{align*} 
Under the conditions of Theorem \ref{momentstatcase}, \eqref{eq:covFun2} yields that the summand of the above expression reduces to 
{\small
\begin{align*}
&= \frac{(2\pi)}{T^{2}}\langle \F_{\omega_{j_1},-\omega_{j_1+h_1},-\omega_{j_2}} (\phi^{\omega_{j_2}}_{l_{3}} \otimes \phi^{\omega_{j_2+h_2}}_{l_4}), \phi^{\omega_{j_1}}_{l_{1}} \otimes \phi^{\omega_{j_1+h_1}}_{l_2}\rangle \Delta_T^{(\omega_{h_2}-\omega_{h_1})} + O\Big(\frac{1}{T^{2}}\Big)\\&
+\bigg[\lambda^{\omega_{j_1}}_{l_1}  \langle\phi^{\omega_{j_2}}_{l_{3}}, \phi^{\omega_{j_1}}_{l_1}\rangle \frac{1}{T}\Delta_T^{(\omega_{j_1}-\omega_{j_2})} + O\Big(\frac{1}{T}\Big)\bigg]\bigg[\lambda^{-\omega_{j_1+h_1}}_{l_2} \langle\phi^{-\omega_{j_2+h_2}}_{l_{4}}, \phi^{-\omega_{j_1+h_1}}_{l_2}\rangle \frac{1}{T}\Delta_T^{(\omega_{j_1+h_1}-\omega_{j_2+h_2})} + O\Big(\frac{1}{T}\Big)\bigg]
\\&
+\bigg[\lambda^{\omega_{j_1}}_{l_1}\langle  \phi^{-\omega_{j_2+h_2}}_{l_{4}}, \phi^{\omega_{j_1}}_{l_1}\rangle\frac{1}{T}\Delta_T^{(\omega_{j_1}+\omega_{j_2+h_2})}+ O\Big(\frac{1}{T}\Big)\bigg]\bigg[  \lambda^{-\omega_{j_1+h_1}}_{l_2}\langle\phi^{\omega_{j_2}}_{l_{3}}, \phi^{-\omega_{j_1+h_1}}_{l_2}\rangle \frac{1}{T}\Delta_T^{(-\omega_{j_1+h_1}-\omega_{j_2})} +O\Big(\frac{1}{T}\Big)\bigg],
\end{align*}
}
where self-adjointness of the spectral density operator gave 
\[\inprod{\F_{\omega_{j_1}}(\phi^{\omega_{j_2}}_{l_{2}})}{\phi^{\omega_{j_1}}_{l_1}} = \inprod{\phi^{\omega_{j_2}}_{l_{2}}}{\F_{\omega_{j_1}} (\phi^{\omega_{j_1}}_{l_1})} = \lambda^{\omega_{j_1}}_{l_1}\inprod{\phi^{\omega_{j_2}}_{l_{2}}}{\phi^{\omega_{j_1}}_{l_1}}.
\]
Self-adjointness of $\F_{\omega}$, orthogonality of the eigenfunctions and $2\pi$-periodicity of the eigenelements imply that
\begin{align*}
\Cov(\sqrt{T}{\bm{\beta}}_{h,u}^{(T)},\sqrt{T}{\bm{\beta}}_{h,u}^{(T)})= &\frac{2\pi}{T^2} \sum_{j_1,j_2=1}^{T} \sum_{\small{\substack{l_1 \in [L(\omega_{j_1})], l_2 \in [L(\omega_{j_1+h})], \\ l_3 \in [L(\omega_{j_2})] , l_4 \in [L(\omega_{j_2+h})]}}} \langle\F_{\omega_{j_1},-\omega_{j_1+h},-\omega_{j_2}} (\phi^{\omega_{j_2}}_{l_{3}} \otimes \phi^{\omega_{j_2+h}}_{l_4}), \phi^{\omega_{j_1}}_{l_{1}} \otimes \phi^{\omega_{j_1+h}}_{l_2}\rangle   \\&
+\frac{2}{T} \sum_{j_1}^{T}  \sum_{\small{\substack{l_1 \in [L(\omega_{j_1})], l_2 \in [L(\omega_{j_1+h})]}}} \lambda^{\omega_{j_1}}_{l_1} \lambda^{\omega_{j_1+h}}_{l_2} +O\big(\frac{1}{T}\big),
\end{align*}
in case $h_1=h_2=h$ and $\Cov(\sqrt{T}{\bm{\beta}}_{h_1,u}^{(T)},\sqrt{T}{\bm{\beta}}_{h_2,u}^{(T)})=O(T^{-1})$ if $h_2 \ne h_1 \mod T$. It can then be derived similarly that $T\cov({\bm{\beta}}_{h_1,u}^{(T)},\overline{{\bm{\beta}}_{h_2,u}^{(T)}}) = O(T^{-1})$ for $h_2 \ne T-h_1 \mod T$.  Since 
\[
\Re\bm{\beta}_{h_1,u}^{(T)}
= \frac{1}{2}(\bm{\beta}_{h_1,u}^{(T)}+\overline{\bm{\beta}_{h_1,u}^{(T)}})
\quad
\text{ and }\quad
\Im  \bm{\beta}_{h_1,u}^{(T)}
= \frac{1}{2\im}(\bm{\beta}_{h_1,u}^{(T)} -\overline{\bm{\beta}_{h_1,u}^{(T)}}),
\]
it follows therefore that
\[ T\cov\big(\Re\bm{\beta}_{h_1,u}^{(T)}, \Im \bm{\beta}_{h_2,u}^{(T)})=O(T^{-1})
\] 
uniformly in $h_1,h_2$ and thus
$
T\,\cov\big(\Re \bm{\beta}_{h_1,u}^{(T)}, \Re \bm{\beta}_{h_2,u}^{(T)}\big) 
= T\,\cov\big( \Im \bm{\beta}_{h_1,u}^{(T)}, \Im \bm{\beta}_{h_2,u}^{(T)}\big) 
 = \frac{T}{2} \cov\big( \bm{\beta}_{h_1,u}^{(T)}, \bm{\beta}_{h_2,u}^{(T)}\big).
$
Finally, using Lipschitz-continuity of $\omega \mapsto \F_{\omega}$ and of its eigenelements to replace the Riemann approximations with their limits completes the proof. 
\end{proof}

\subsection{Proof of  Theorem \ref{UnifCon}}
\label{subsec:proof:alt}

\begin{proof}[Proof of Theorem \ref{UnifCon}]
(i) In order to prove the first assertion of the theorem, introduce the bias-variance decomposition
\begin{align}
\label{VBd}
\E \Big[\bigsnorm{\hat{\mathcal{F}}^{(T)}_{\omega}
- \E\big[\hat{\mathcal{F}}^{(T)}_{\omega}\big]
+ \E\big[\hat{\mathcal{F}}^{(T)}_{\omega}\big] - {G}_{\omega}}^2_2 \Big]
=   \E\Big[ \bigsnorm{\hat{\mathcal{F}}^{(T)}_{\omega} - \E\big[\hat{G}_{\omega}\big]}^2_2\Big]
+\E \Big[\bigsnorm{\E\big[\hat{\mathcal{F}}^{(T)}_{\omega}\big]-{G}_{\omega}}^2_2\Big]. 
\end{align}
The cross terms cancel because 
$\E[{\langle}  \hat{\mathcal{F}}^{(T)}_{\omega}- \E[\hat{\mathcal{F}}^{(T)}_{\omega}],  
\E[\hat{\mathcal{F}}^{(T)}_{\omega}]-{G}_{\omega} {\rangle}_{H \otimes H}]$
and $\E[\hat{\mathcal{F}}^{(T)}_{\omega}- \E[\hat{\mathcal{F}}^{(T)}_{\omega}]] = O_H$. Now, by Corollary \ref{FCboundssd},
 \begin{align*} 
\cum\big(D^{(T)}_{\omega}, D^{(T)}_{-\omega}\big) 
= \frac{1}{T} \sum_{t=0}^{T-1} \F_{{t}/{T},\omega} + R_{T,2} 
= G^{(T)}_{\omega} + R_{T,2},
\end{align*}
where $\| R_{T,2}\|_2 =O(T^{-1})$. Note that the integral approximation in time direction does not change the error term because of Lipschitz continuity of the mapping $(u,\omega) \mapsto \F_{u,\omega}$ in $u$. Convolution of the cumulant tensor with the smoothing kernel, replacing the integral approximation with the limit and a change of variables give
\begin{align*} 
\E \big[\hat{\F}^{(T)}_{\omega}\big] 
&= \frac{2 \pi}{b T}\sum_{j= 1}^{T} K_b(\omega-\omega_j) \cum\big(D^{(T)}_{\omega_j}, D^{(T)}_{-\omega_j}\big) 
\\& = \int  K(x) {G}_{\omega-xb} dx+R_{b,T}+R_{T,2},
\end{align*}
where $\|R_{b,T}\|_2 = O({b T}^{-1})$. Since $\sup_{\omega,u} \snorm{\F_{u,\omega}}_2 < \infty$ and $\sup_{\omega, u} \snorm{\frac{\partial^2}{\partial \omega_2} \F_{u,\omega}}_2 < \infty$, the mapping $\omega\mapsto G_{\omega}$ is twice differentiable and $\sup_{\omega} \snorm{\frac{\partial^2}{\partial \omega_2} G_{\omega}}_2 < \infty$. Therefore, a Taylor expansion around $\omega$ and symmetry of the kernel then lead to

\begin{align*}
\E \big[\hat{\F}^{(T)}_{\omega}\big]  =\int  K(x) {G}_{\omega-xb} dx= G_{\omega} + \sum_{i=1}^{2} \frac{1}{i! } (b)^i  \frac{\partial^i G_{\nu}}{\partial \nu^i} \Big|_{\nu=\omega}   \int x^i K (x) dx  ={G}_{\omega} + \epsilon_{b,T},
\end{align*}
where $\|\epsilon_{b,T}\|_2 = O(b^2+({b T}))^{-1}$. Thus, the second term on the right-hand side of (\ref{VBd}) satisfies
\begin{align}\label{eq:biasF}
\E\big[ \snorm{\E\hat{\mathcal{F}}^{(T)}_{\omega}-G_{\omega} }^2_2 \big]
=O\bigg(b^2 + \frac{1}{b T} \bigg)^2
\end{align}
uniformly in $\omega \in [-\pi,\pi]$. To bound the first term of the right-hand side in (\ref{VBd}), observe that, for $j_1+ j_2 \equiv 0 \mod T$, Lemma \ref{cumboundglsp} with $k=2$ yields
\begin{align*}
\cum(D^{(T)}_{\omega_{j_1}},D^{(T)}_{\omega_{j_2}}) = \frac{1}{T} \sum_{t=1}^{T} \F_{t/T} e^{-i (\omega_{j_1}+\omega_{j_2}) t}+ R_{2,T} \to \int_0^1 {\F}_{u; \omega_{j_1}} e^{-i 2\pi u ({j_1}+{j_2})} du = \tilde{\F}_{j_1+j_2; \omega_{j_1}}.
\end{align*}
Furthermore, from Corollary \ref{FCboundssd} and Minkowski's inequality 
\begin{align*}
\Bigsnorm{\cum(D^{(T)}_{\omega}, D^{(T)}_{-\omega}, D^{(T)}_{\omega^\prime}, D^{(T)}_{-\omega^\prime})}_2 
&\le \frac{1}{T} \biggsnorm{\frac{1}{T}\sum_{t=0}^{T-1}\mathcal{F}_{\frac{t}{T},\omega,-\omega,\omega^\prime}}_2 
+ O\bigg(\frac{1}{T^2}\bigg) \\
&= \frac{1}{T}\Bigsnorm{G^{(T)}_{\omega,-\omega,\omega^\prime}}_2+O\bigg(\frac{1}{T^2}\bigg)
= O\bigg(\frac{1}{T}\bigg). \tageq \label{eq:covF4term}
\end{align*} 
The last equality follows since $\sup_{u,\omega,\omega^\prime}\snorm{\mathcal{F}_{{t}/{T},\omega,-\omega,\omega^\prime}}_2 \le\sum_{h_1,h_2,h_3 \in \znum} \|\kappa_{3;h_1,h_2,h_3}\|_2 = O(1)$ by \ref{cumglsp}. Theorem \ref{prodcumthm} hence implies that
\begin{align*}
\cov(\hat{\F}_{\omega},\hat{\F}_{\omega}) &
 =\frac{1}{(b T)^2} \sum_{k_1,k_2=1}^{T} K\bigg(\frac{\omega-\omega_{k_1}}{b} \bigg)  K\bigg(\frac{\omega-\omega_{k_2}}{b} \bigg)\Big[\cum\big(D^{(T)}_{\omega_{k_1}},D^{(T)}_{-\omega_{k_1}},
D^{(T)}_{-\omega_{k_2}}, D^{(T)}_{\omega_{k_2}}\big)
\\&
 \qquad \phantom{\cov(\hat{\F}_{\omega_{j_1}},\hat{\F}_{\omega_{j_1}}) =}
 +S_{1324}\Big(
\cum\big(D^{(T)}_{\omega_{k_1}},D^{(T)}_{-\omega_{k_2}}\big)\otimes
\cum\big(D^{(T)}_{-\omega_{k_1}},D^{(T)}_{\omega_{k_2}}\big)\Big)\\&
\qquad \phantom{\cov(\hat{\F}_{\omega_{j_1}},\hat{\F}_{\omega_{j_1}}) =}
+S_{1423}\Big(
\cum\big(D^{(T)}_{\omega_{k_1}},D^{(T)}_{\omega_{k_2}}\big)\otimes
\cum\big(D^{(T)}_{-\omega_{k_1}},D^{(T)}_{-\omega_{k_2}}\big)\Big)\Big].
\end{align*}
Using Lemma \ref{cumboundglsp}, this equals
\begin{align*}
\frac{1}{(b T)^2} \sum_{k_1,k_2=1}^{T}
& K\bigg(\frac{\omega-\omega_{k_1}}{b} \bigg)  K\bigg(\frac{\omega-\omega_{k_2}}{b} \bigg)\Big[
S_{1324}\big(
\tilde{\F}_{k_1-k_2;\omega_{k_1}}\otimes \tilde{\F}_{-k_1+k_2;-\omega_{k_1}}\big)\\&+S_{1423}\big(\tilde{\F}_{k_1+k_2;\omega_{k_1}}\otimes \tilde{\F}_{-k_1-k_2;-\omega_{k_1}}\big)+R_{2,T}\Big]. \tageq \label{eq:covariancF4}
\end{align*} 
where we used \eqref{eq:covF4term} is of order $O(\frac{1}{T})$ in $S_2$ uniformly in $-\pi \le \omega, \omega^\prime \le \pi$. Using a change of variables, the properties of the smoothing kernel, H\"older's inequality and Corollary \ref{cumbounds}, it follows that
\begin{align*}
& \biggsnorm{\frac{1}{(b T)^2} \sum_{k_1,k_2=1}^{T} K\bigg(\frac{\omega-\omega_{k_1}}{b} \bigg)  K\bigg(\frac{\omega-\omega_{k_2}}{b} \bigg)\tilde{\F}_{k_1-k_2;\omega_{k_1}}\widetilde{\otimes}\tilde{\F}_{-k_1+k_2;-\omega_{k_1}}}_2 
\\& \le 
\biggsnorm{\frac{1}{(b T)^2} \sum_{k_1}^{T} K\bigg(\frac{\omega-\omega_{k_1}}{b} \bigg)^2   \sum_{s}\tilde{\F}_{s;\omega_{k_1}}\otimes \tilde{\F}_{-s;-\omega_{k_1}}}_2 
\\& 
\le
\sup_{\omega} \sum_{s \in \mathbb{Z}}\bigsnorm{\tilde{\F}_{s;\omega}}^2_2  \bigg\vert\frac{1}{(b T)^2} \sum_{k_1}^{T} K\bigg(\frac{\omega_{j_1}-\omega_{k_1}}{b} \bigg)^2 \bigg\vert =O\bigg(\frac{1}{bT}\bigg).
\end{align*} 
A similar argument holds for the remaining term of \eqref{eq:covariancF4}. Hence,
 \begin{align}
\Bigsnorm{\Cov  \big(\hat{\F}^{(T)}_{\omega}, \hat{\F}^{(T)}_{\omega}\big)}_2& =\Bigsnorm{  \bigg(\frac{2 \pi}{T}\bigg)^2 \sum_{j,j^\prime = 1}^{T} K_b(\omega-\omega_j) K_b(\omega-\omega_{j^\prime})\cov(I^{(T)}_{\omega_j}, I^{(T)}_{\omega_{j^\prime}}) \nonumber}_2 =O\bigg(\frac{1}{bT}\bigg)  \label{eq:variancep3}.
 \end{align}
Fubini's theorem together with the above implies that the first term of (\ref{VBd}) satisfies
\begin{align*}
\E \big[\snorm{\hat{\mathcal{F}}^{(T)}_{\omega} - \E\hat{\mathcal{F}}^{(T)}_{\omega}}^2_2\big] 
=\text{trace}(\Var (\hat{\mathcal{F}}^{(T)}_{\omega} ))
= O\bigg(\frac{1}{ b T} \bigg)
\end{align*}
uniformly in $\omega \in [-\pi,\pi]$. This establishes (i). \\
(ii) This part of the proof requires the following lemma verified in Section \ref{sec:S2stuff} of the Online Supplement.
\begin{lemma}\label{lem:ChebSP}
Let $Y_{\nu}, \nu \in [a,b]$ be a zero-mean $L^2([0,1]^k)$-valued stochastic process of which the derivative mapping $\nu \mapsto \frac{\partial}{\partial \nu}Y_{\nu}$ is well-defined in $L^2([0,1]^k)$ for any $\nu \in [a,b]$. If $\E \|Y_{\nu}\|^2_2 <\infty$ and $\E \|\frac{\partial Y_{\nu}}{\partial \nu}\|^2_2 <\infty$, then 
\[
2\E \sup_{a\le\nu \le b}\|Y_{\nu}\|_2^2  \le \E\|Y_{a}\|_2^2 +\E \| Y_b\|_2^2 + \int_{a}^{b} \sqrt{\E  \|\frac{\partial}{\partial \alpha} Y_{\alpha}\|^2_2} \sqrt{\E\| \overline{Y_{\alpha}} \|_2^2}+ \int_{a}^{b} \sqrt{\E\|Y_{\alpha}\|^2_2}d\alpha  \sqrt{\E\| \frac{\partial}{\partial \alpha}\overline{Y_{\alpha}}\|^2_2}d\alpha.\]
\end{lemma}

Lemma  \ref{lem:ChebSP} with $k=2$ applied to the spectral density kernel function $\hat{f}_{\omega}$ implies --- due to the norm equivalence with the operator $\hat{\F}_{\omega}$ --- that
\begin{align*}
\E \sup_{0\le\omega \le \pi}2 \snorm{\hat{\F}_{\omega}-\E {\hat{\F}}_{\omega}}^2_2 \le &\E \snorm{\hat{\F}_{0}-\E {\hat{\F}}_{0}}^2_2 + \E \snorm{\hat{\F}_{\pi}-\E {\hat{\F}}_{\pi}}^2_2 \\&+ 2 \int^\pi_0 \sqrt{\E \snorm{\hat{\F}_{\omega}-\E {\hat{\F}}_{\omega}}^2_2} \sqrt{\E \snorm{\frac{\partial}{\partial \omega}(\hat{\F}_{\omega}-\E {\hat{\F}}_{\omega} )}^2_2}  d\omega
\\& =\tr \var(\hat{\F}_{0})+\tr \var(\hat{\F}_{\pi}) + 2 \int^\pi_0 \sqrt{\tr \var(\hat{\F}_{\omega})} \sqrt{\tr \var(\frac{\partial}{\partial \omega}\hat{\F}_{\omega})}  d\omega
\\& = O\bigg(\frac{1}{bT}\bigg) + O\bigg(\frac{1}{\sqrt{b T} \sqrt{b^2 T}}\bigg) = O\bigg(\frac{1}{b^2T}\bigg), \tageq \label{eq:Exsup}
\end{align*}
where the latter follows from part (i). The rate for the covariance structure of the operator-valued function $\omega \mapsto \frac{\partial}{\partial \omega}\hat{\F}_{\omega}$ follow as before, noting that an application of the chain rule of the derivative will lead to an extra $O(\frac{1}{b^2})$ term in $S_2(H)$ in comparison to the covariance of $\hat{\F}_{\omega}$. Minkowski's inequality therefore implies
\begin{align*}
\mathbb{P}\bigg(\sup_{\omega \in [-\pi,\pi] }\!\snorm{\hat{\mathcal{F}}^{(T)}_{\omega}- {G}_{\omega}}_{2} > \epsilon\bigg) \!\!
\le  \mathbb{P}\bigg(\sup_{\omega \in [-\pi,\pi] }\!\snorm{\hat{\mathcal{F}}^{(T)}_{\omega}- \E {\hat{\F}}_{\omega}}_{2} >\frac{\epsilon}{2}\bigg) 
\!\!+\!\mathbb{P}\bigg(\sup_{\omega \in [-\pi,\pi] }\!\snorm{\E\hat{\mathcal{F}}^{(T)}_{\omega}-  {G}_{\omega}}_{2} > \frac{\epsilon}{2}\bigg).
\end{align*}
Using Markov's inequality together with \eqref{eq:Exsup}, for any $\epsilon >0$,
\[
\mathbb{P}\bigg(\sup_{\omega \in [-\pi,\pi] }\snorm{\hat{\mathcal{F}}^{(T)}_{\omega}- \E {\hat{\F}}_{\omega}}_{2} > \frac{\epsilon}{2}\bigg)  
\le O\bigg(\frac{1}{\epsilon^2 b^2 T}\bigg) \to 0
\]
as $b^2 T \to \infty$. Similary, Markov's inequality together with \eqref{eq:biasF} yields
\[
\mathbb{P}\bigg(\sup_{\omega \in [-\pi,\pi] }\snorm{\E\hat{\mathcal{F}}^{(T)}_{\omega}-  {G}_{\omega}}_{2} > \frac{\epsilon}{2}\bigg) 
\le  O\bigg( \frac{1}{\epsilon^2} \Big(b^2+\frac{1}{bT}\Big)^2 \bigg)\to 0
\]
as $b T \to \infty$ and $b\to 0$ as $T \to \infty$. The result therefore holds provided Assumption \ref{windowfunction} is satisfied. 
\end{proof}

\section{Weak convergence} 
\label{sec:proofs:weak_conv}
The proof of the distributional properties of $\tprojes_{h,x}$ as stated in Theorem \ref{th:test_null} and \ref{th:test_alt} are established in this section. The proof consists of several steps. First, the distributional properties are derived for $\tproj_{h,x}$, when spectral density operators and their corresponding eigenelements are known. For this, we investigate the distributional properties of the operator
\begin{align}
\WT= \frac{1}{T} \sum_{j=1}^{T} D^{(T)}_{\omega_{j}} \otimes D^{(T)}_{\omega_{j+h}} \quad h =1, \ldots T-1. \label{eq:WT}
\end{align}
Theorem \ref{CLT_alt} below shows that $\sqrt{T}(\WT -\E \WT)$ converges weakly to a functional Gaussian process both under the null and the alternative. The distributional properties of $\tproj_{h,x}$ immediately follow from this result and thus converge weakly to a Gaussian process under both hypotheses.  Focus is finally on $\tprojes_{h,x}$, where the effect of replacing the eigenelements with their empirical counterparts on the distributional properties is clarified. In particular, Theorems \ref{boundstatcase} and \ref{boundlocstatcase} are established as well as the orders of 
\begin{align*}
\E \sqrt{T}|\tprojes_{h,u}-\tproj_{h,u}| 
\quad\text{ and }\quad
\E \sqrt{T}|\tprojes_{h,s}-\tproj_{h,s}|.
\end{align*}
\subsection{Weak convergence on the function space}

To demonstrate weak convergence of \eqref{eq:WT}, the following result by \citet{ck86} is used, which considerably simplifies the verification of the usual tightness condition often invoked in weak convergence proofs of Banach space-valued random variables. 

\begin{lemma}
\label{convergencecriterion}
Let $(\mathscr{T},\mathcal{A},\mu)$ be a measure space, let $(B,|\cdot|)$ be a Banach space, and let $X=(X_n\colon n\in\nnum)$ be a sequence of random elements in $L^p_B(\mathscr{T},\mu)$ such that
\begin{enumerate}
\itemsep-.3ex
\item[(i)] the finite-dimensional distributions of $X$ converge weakly to those of a random element $X_0$ in $L^p_B(\mathscr{T},\mu)$;
\item[(ii)] $\displaystyle\limsup_{n\to\infty}\E[\|X_n\|_p^p]\leq\E[\|X_0\|_p^p] < \infty$.
\end{enumerate}
Then, $X$ converges weakly to $X_0$ in $L^p_B(\mathscr{T},\mu)$.
\end{lemma} 

To apply Lemma \ref{convergencecriterion} in the present context, consider the sequence $(\ET_h\colon T \in \nnum)$ of random elements in $L^2_\cnum([0,1]^2)$, for $h=1,\ldots,T-1$ defined through
\begin{align*}
\ET_h =\sqrt{T}\Big(\WT -\E\big[\WT\big] \Big)=
\sum_{l,l^\prime=1}^\infty  \langle \ET_h , \psi_{ll^\prime}\rangle \psi_{ll^\prime}, 
\end{align*}
where the second equality uses a representation with respect to an $L^2_{\cnum}([0,1]^2)$ orthonormal basis $\psi_{ll^\prime}=\psi_l \otimes \psi_l^\prime$. From this representation it is easily seen that the finite-dimensional distributions of the basis coefficients provide a complete characterization of the distributional properties of $\ET_h$. 
To formalize this, we put the functional $\ET_h$ in duality with $(\ET_h)^{*} \in {L^2_{\cnum}([0,1]^2)}^{*}$ through the pairing
$
{\ET}_h(\phi) = \langle \ET_h , \phi \rangle
$ 
for all $\phi \in L^2_{\cnum}([0,1]^2)^{*}$. This leads to the following result, which is stated under the more general \ref{cumglsp}, which encompasses the stationary case.
\begin{theorem}[\bf Weak convergence]
\label{CLT_alt}
Let $(X_t\colon t\in\mathbb{Z})$ be a stochastic process taking values in $H_\rnum$ satisfying \ref{cumglsp} with $\ell=2$. Then,
\begin{equation} 
\big(\Re\ET_{h_i} , \Im\ET_{h_i} \colon i=1,\ldots,k\big)\stackrel{d}{\to} (\mathcal{R}_{h_i}, \mathcal{I}_{h_i}\colon i=1,\ldots,k),
\end{equation}
where $\mathcal{R}_{h},  \mathcal{I}_{h^\prime}$, $h, h^\prime \in \{1,\ldots,T-1\}$, are jointly Gaussian elements in $L^2_{\cnum}([0,1]^2)$ with means $\E[\mathcal{R}_{h}(\psi_{ll^\prime})]=\E[\mathcal{I}_{h^\prime}(\psi_{ll^\prime})]=0$ and covariance structure 
\begin{enumerate} \label{covaltemprroc}
\item $\begin{aligned}[t]
  \cov( &\mathcal{R}_{h}(\psi_{l_1l_1^\prime}),\mathcal{R}_{h^\prime}(\psi_{l_2l_2^\prime})) =\\& \frac{1}{4}\big[
\Upsilon_{h,h^\prime}(\psi_{l_1l_1^\prime\,l_2l_2^\prime})+\acute{\Upsilon}_{h,h^\prime}(\psi_{l_1l_1^\prime\,l_2l_2^\prime})+\grave{\Upsilon}_{h,h^\prime}(\psi_{l_1l_1^\prime\,l_2l_2^\prime})+\bar{\Upsilon}_{h,h^\prime}(\psi_{l_1l_1^\prime\,l_2l_2^\prime})\big]\end{aligned}$
\item $\begin{aligned}[t]
\cov(& \mathcal{I}_{h}(\psi_{l_1l_1^\prime}),\mathcal{R}_{h^\prime}(\psi_{l_2l_2^\prime})) =\\& \frac{1}{4 \im}\big[\Upsilon_{h,h^\prime}(\psi_{l_1l_1^\prime\,l_2l_2^\prime})+\acute{\Upsilon}_{h,h^\prime}(\psi_{l_1l_1^\prime\,l_2l_2^\prime})-\grave{\Upsilon}_{h,h^\prime}(\psi_{l_1l_1^\prime\,l_2l_2^\prime})-\bar{\Upsilon}_{h,h^\prime}(\psi_{l_1l_1^\prime\,l_2l_2^\prime})\big]\end{aligned}$
\item $\begin{aligned}[t]
\cov(& \mathcal{R}_{h}(\psi_{l_1l_1^\prime}),\mathcal{I}_{h^\prime}(\psi_{l_2l_2^\prime})) =\\& \frac{1}{4 \im}\big[\Upsilon_{h,h^\prime}(\psi_{l_1l_1^\prime\,l_2l_2^\prime})-\acute{\Upsilon}_{h,h^\prime}(\psi_{l_1l_1^\prime\,l_2l_2^\prime})+\grave{\Upsilon}_{h,h^\prime}(\psi_{l_1l_1^\prime\,l_2l_2^\prime})-\bar{\Upsilon}_{h,h^\prime}(\psi_{l_1l_1^\prime\,l_2l_2^\prime})\big]\end{aligned}$
\item $\begin{aligned}[t]
\cov(& \mathcal{I}_{h}(\psi_{l_1l_1^\prime}), \mathcal{I}_{h^\prime}(\psi_{l_2l_2^\prime}))=\\& \frac{1}{4}\big[\Upsilon_{h,h^\prime}(\psi_{l_1l_1^\prime\,l_2l_2^\prime})-\acute{\Upsilon}_{h,h^\prime}(\psi_{l_1l_1^\prime\,l_2l_2^\prime})-\grave{\Upsilon}_{h,h^\prime}(\psi_{l_1l_1^\prime\,l_2l_2^\prime})+\bar{\Upsilon}_{h,h^\prime}(\psi_{l_1l_1^\prime\,l_2l_2^\prime})\big]  \end{aligned}$
\end{enumerate}
for all $h, h^\prime$ and $l_1,l_1^\prime,l_2,l_2^\prime$, and where $\Upsilon_{h,h^\prime}, \acute{\Upsilon}_{h,h^\prime}, \grave{\Upsilon}_{h,h^\prime}$ and $\bar{\Upsilon}_{h,h^\prime}$ are given in \eqref{eq:nonnon_lim}--\eqref{eq:concon_lim}.
\end{theorem}

\begin{proof}
It remains to verifiy the conditions of Lemma \ref{convergencecriterion}. For the first, the following theorem establishes that the finite-dimensional distributions converge weakly to a Gaussian process both under the null and the alternative. 
\begin{theorem}\label{cumconv}
Under the conditions of Theorem \ref{CLT_alt}, we have for all $l_i,l_i^\prime \in \nnum$, $h_i=1,\ldots,T-1$, $i=1,\ldots, k$ and $k \ge 3$,
\begin{align*} 
\cum \Big( \ET_{h_1} (\psi_{l_1 l_1^\prime}), \ldots, \ET_{h_k} (\psi_{l_k l_k^\prime}) \Big)  =o(1) 
\qquad(T \to \infty).
\end{align*}
\end{theorem}
The proof of \ref{cumconv} can be found in Section \ref{asdist} of the Online Supplement. Note that, for the second condition of Lemma \ref{convergencecriterion},  Parseval's identity and the monotone convergence theorem yield
\begin{equation}
\label{cremerkadelka-condition}
\E\big[\big\|\ET_{h}\big\|^2_2\big]
=\sum_{l,l^\prime=1}^\infty\E\big[\big|\ET_{h}(\psi_{ll^\prime})\big|^2\big]\to
\sum_{l,l^\prime=1}^\infty\E\big[\big|E_{h}(\psi_{ll^\prime})\big|^2\big]
=\E\big[\big\|E_{h}\big\|^2_2 \big]
\qquad(T\to\infty),
\end{equation}
with $E_{h}$ denoting the limiting process. Observe that, from \eqref{eq:OrderVar} and the Cauchy--Schwarz inequality, the terms $\Upsilon_{h,h^\prime}, \acute{\Upsilon}_{h,h^\prime}, \grave{\Upsilon}_{h,h^\prime}$ and $\bar{\Upsilon}_{h,h^\prime}$ are finite. 
Condition (ii) of Lemma \ref{convergencecriterion} is then satisfied, since
\begin{align*}
\E\big[\big\|\ET_{h}\big\|^2_2\big]
=\int_{[0,1]^2} \var \big(\ET_{h}(\tau,\tau^\prime)\big) d\tau d\tau^\prime 
= T\E\|\WT\|^2_2 \to \tr(\var(\mathcal{R}_h)) +\tr(\var(\mathcal{I}_h))<\infty,
\end{align*}
where Tonelli's theorem was applied to obtain the first equality. This completes the proof.
\end{proof}

\subsection{Replacing eigenelements with estimates}
\label{subsec:proof:projtoest}

\subsubsection{Invariance under rotation}

We now focus on replacing the projection basis with estimates of the eigenfunctions of the spectral density operators. It can be shown \citep{mm03} that for 
rates of the bandwidth $b$ for which the estimated spectral density operator is a consistent estimator of the true spectral density operator, the corresponding estimated eigenprojectors $\hat{\Pi}_l^{\omega}= \hat{\phi_l}^{\omega} \otimes \hat{\phi_l}^{\omega}$ are consistent for the eigenprojectors ${\Pi_{l}^{\omega}}$. However, the estimated eigenfunctions are not unique and only identified up to rotation on the unit circle. In order to show that replacing the eigenfunctions with estimates does not affect the limiting distribution, the issue of rotation has to be considered first. More specifically, when estimating, 
a version $\hat{z}_l \hat{\phi}^{\omega_{j}}_{l}$, where $\hat{z}_l \in \mathbb{C}$ with modulus $|\hat{z}_l|=1$, is obtained 
which cannot be guaranteed to be close to the true eigenfunction ${\phi}^{\omega_{j}}_{l}$. It is therefore essential that the test statistic is invariant under rotations. To show this, write 
\[
\Psi_h(j,l,l^\prime) = \langle D^{(T)}_{\omega_{j}}, \hat{\phi}^{\omega_j}_{l} \rangle \overline{\langle D^{(T)}_{\omega_{j+h}}, \hat{\phi}^{\omega_{j+h}}_{l^\prime}}\rangle
\]
and let $\boldsymbol{\Psi}(h) = \text{vec}(\Psi_h(j,l,l^\prime,))$ be the stacked vector of dimension $\prod_{j=1}^{T-h}L(\omega_j) L(\omega_{j+h})$. Note that then 
$
\hat{\bm{\beta}}_{h,u}^{(T)} = e^{\top} \boldsymbol{\Psi}(h).
$ 
Construct the diagonal matrix
\begin{align*}
Z^{j}_{L(\omega_j)} = 
\begin{pmatrix}
 \hat{z}^{j}_1&\cdots & &\\
\vdots &  \hat{z}^{j}_2 & &\\
&  & \ddots &  \\
& & &  \hat{z}^{j}_{L(\omega_j)} 
\end{pmatrix},
\end{align*}
the block diagonal matrix $Z^{1:T}_{L(\omega_j)} = \text{diag}(Z^{j}_{L(\omega_{j})}\colon j =1,\ldots, T)$ and the Kronecker product
$
\boldsymbol{Z}(h) = Z^{1:T-h}_{L(\omega_j)}\otimes Z^{h:T}_{L(\omega_j)}.
$ 
This is a diagonal object of dimension $(\prod_{j=1}^{T-h}L(\omega_j) L(\omega_{j+h}))^2$, 
whose diagonal elements are given by $\{\hat{z}^{j}_l \overline{\hat{z}}^{j+h}_{l^\prime}\}$.
Rotating the eigenfunctions on the unit circle, yields versions 
\[
\hat{\bm{\beta}}_{h,u}^{(T)}=e^{\top} \boldsymbol{Z}(h)\boldsymbol{\Psi}(h).
\]
For these versions, write
$
\sqrt{T}\mathcal{Z}_{M}\hat{\bm{b}}_{M,u}^{(T)}
=\sqrt{T}\big(\Re\hat{\bm{\beta}}_{h_1,u}^{(T)},\ldots,\Im\hat{\bm{\beta}}_{h_M,u}^{(T)},
\Im\hat{\bm{\beta}}_{h_1,u}^{(T)},\ldots,\Im\hat{\bm{\beta}}_{h_M,u}^{(T)}\big)^\top,
$ 
where the block diagonal matrix is given by 
$ 
\mathcal{Z}_{M} = \text{diag}( \boldsymbol{\Re Z}(h_1), \ldots,\boldsymbol{\Re Z}(h_M),\boldsymbol{\Im Z}(h_1), \ldots \boldsymbol{\Im Z}(h_M))^{\top}.
$ 
The same rotation, however, also implies that $\hat{\Sigma}_{M,u}$ becomes $\mathcal{Z}_{M}\hat{\Sigma}_{M,u} \mathcal{Z}_{M}^{\top}$ and hence
\[
T(\hat{\bm{b}}_{M,u}^{(T)})^\top (\mathcal{Z}_{M})^{\top}{[\mathcal{Z}_{M} \hat{\Sigma}_{M,u} {\mathcal{Z}_{M}}^{\top}]}^{-1}\mathcal{Z}_{M} \hat{\bm{b}}_{M,u}^{(T)} = \hat{Q}_{M,u}^{(T)},
\]
thereby showing that the value of the test statistic is not affected by rotation of the estimated eigenfunctions. In the rest of the proof, focus is therefore only on estimates $\hat{\phi}^{\omega_{j+h}}_{l}$ and $\hat{\phi}^{\omega_{j+h}}_{l^\prime}$ and their respective unknown rotations $\hat{z}^{j}_l$ and $ \overline{\hat{z}}^{j+h}_{l^\prime}$ are ignored. 

\subsubsection{Limiting distristributions of $\tprojes_{h,u}$ and $\tprojes_{h,s}$}
\label{sec:bounderror}

We now investigate the rate of convergence of the statistic when the eigenfunctions as a basis are replaced with their empirical counterparts, and prove Theorems \ref{boundstatcase} and \ref{boundlocstatcase}. For this, it is sufficient to derive the order of the difference
 \begin{align}\label{eq:J}   \sqrt{T}\E |{\hat{\bm{\beta}}}_{h,x}^{(T)}- {\bm{\beta}}_{h,x}^{(T)}|. \end{align}
In the following we shall focus on ${\hat{\bm{\beta}}}_{h,u}^{(T)}$ and postpone the derivation for ${\hat{\bm{\beta}}}_{h,s}^{(T)}$ to Section \ref{sec:hatbS}. In order to bound \eqref{eq:J}, relate $\hat{\phi}^{\omega_j}_{l} \otimes \hat{\phi}^{\omega_{j+h}}_{\lpr}-{\phi}^{\omega_j}_{l} \otimes {\phi}^{\omega_{j+h}}_{\lpr}$ with $\hat{\F}_{\omega_{j_1}} \widetilde{\otimes}  \hat{\F}_{\omega_{j_1+h}} -\F_{\omega_{j_1}} \widetilde{\otimes}  \F_{\omega_{j_1+h}}$ from noting that
\begin{align*}
(\F_{\omega_{j}} \widetilde{\otimes}  \F_{\omega_{j+h}}) ({\phi}^{\omega_j}_{l} \otimes {\phi}^{\omega_{j+h}}_{\lpr}) & = \sum_{m, m^\prime} \lambda^{\omega_{j}}_{m}  \lambda^{\omega_{j+h}}_{m^\prime} {\phi}^{\omega_j}_{m} \otimes  {\phi}^{\omega_j}_{m}  ({\phi}^{\omega_j}_{l} \otimes {\phi}^{\omega_{j+h}}_{\lpr})
{\phi}^{\omega_j+h}_{m} \otimes  {\phi}^{\omega_j+h}_{m^\prime}  
\\& =  \lambda^{\omega_{j}}_{l}  \lambda^{\omega_{j+h}}_{l^\prime}( {\phi}^{\omega_j}_l \otimes \phi^{\omega_{j+h}}_{\lpr})
\end{align*}
where we used Definition \ref{def:Krn}(i). Similarly,
$(\hat{\F}_{\omega_{j}} \widetilde{\otimes}  \hat{\F}_{\omega_{j+h}}) (\hat{\phi}^{\omega_j}_{l} \otimes \hat{\phi}^{\omega_{j+h}}_{\lpr}) =  \hat{\lambda}^{\omega_{j}}_{l}  \hat{\lambda}^{\omega_{j+h}}_{l^\prime} (\hat{\phi}^{\omega_j}_l \otimes \hat{\phi}^{\omega_{j+h}}_{\lpr}).$
A first-order Taylor expansion of the eigenvalue-eigenvector equation yields \citep[e.g.,][]{HH-N}
\begin{align}\label{eq:eigapprox}
&\hat{\phi}^{\omega_j}_{l} \otimes \hat{\phi}^{\omega_{j+h}}_{\lpr}-({\phi}^{\omega_j}_{l} \otimes {\phi}^{\omega_{j+h}}_{\lpr}) \\
&=\!\!
\sum_{\substack{m \ne l \\ m^\prime \ne \lpr}} \frac{1}{\lambda^{\omega_{j}}_{l}  \lambda^{\omega_{j+h}}_{l^\prime}-\lambda^{\omega_{j}}_{m}  \lambda^{\omega_{j+h}}_{m^\prime}} 
\biginprod{\big(\hat{\F}_{\omega_{j_1}} \widetilde{\otimes}  \hat{\F}_{\omega_{j+h}} -\F_{\omega_{j}} \widetilde{\otimes}  \F_{\omega_{j+h}}\big)({\phi}^{\omega_j}_l \otimes {\phi}^{\omega_{j+h}}_{\lpr})}{{\phi}^{\omega_j}_m \otimes {\phi}^{\omega_{j+h}}_{m^{\prime}}}{\phi}^{\omega_j}_m \otimes {\phi}^{\omega_{j+h}}_{m^{\prime}} \!+\!R,
\nonumber
\end{align}
where the remainder $R$ is of order $\snorm{R}_2 = O_p(\snorm{\hat{\F}_{\omega_{j_1}} \widetilde{\otimes}  \hat{\F}_{\omega_{j_1+h}} -\F_{\omega_{j_1}} \widetilde{\otimes}  \F_{\omega_{j_1+h}}}^2_2)$ and will be of smaller order than the first term on the right-hand side of \eqref{eq:eigapprox}. In the proof we require thus that
 \begin{align} \label{eq:condeigenv}
\lambda^{\omega_{j}}_{l}  \lambda^{\omega_{j+h}}_{l^\prime}-\lambda^{\omega_{j}}_{m}  \lambda^{\omega_{j+h}}_{m^\prime}  =\lambda^{\omega_{j}}_{l} ( \lambda^{\omega_{j+h}}_{l^\prime}-\lambda^{\omega_{j+h}}_{m^\prime} )+ ( \lambda^{\omega_{j}}_{l}-\lambda^{\omega_{j}}_{m}) \lambda^{\omega_{j+h}}_{m^\prime} >0,
 \end{align}
which implies no multiplicity of eigenvalues. It is also required that the spectral density operators are strictly positive definite, a condition needed to ensure that their eigenfunctions form a complete orthonormal basis of $H$. Note, however, that the assumption of no multiplicity is without loss of generality as one can group multiple eigenelement pairs into blocks and apply the same techniques over these blocks \citep[e.g.,][]{mm03}. Given \eqref{eq:condeigenv} holds true, linearity and continuity of the inner product imply that the error can be rewritten as
\begin{align*}
&\scriptstyle\frac{1}{\sqrt{T}} \sum_{j=1}^{T} \sum_{l,\lpr}  \inprod{D_{\omega_{j}} \otimes D_{\omega_{j+h}}}{\phi^{\omega_j}_{l} \otimes \phi^{\omega_{j+h}}_{\lpr}-(\hat{\phi}^{\omega_j}_{l} \otimes \hat{\phi}^{\omega_{j+h}}_{\lpr})}_{S},
\\&\scriptstyle
=O_p(\frac{1}{\sqrt{T}} \sum_{j=1}^{T}  \sum_{l,\lpr} \sum_{\substack{m \ne l \\ m^\prime \ne \lpr}} \inprod{D_{\omega_{j}} \otimes D_{\omega_{j+h}}}{\phi^{\omega_j}_m \otimes {\phi}^{\omega_{j+h}}_{m^{\prime}}}_{S} \biginprod{\Big(\hat{\F}_{\omega_{j}} \widetilde{\otimes}  \hat{\F}_{\omega_{j+h}} -\F_{\omega_{j}} \widetilde{\otimes}  \F_{\omega_{j+h}}\Big)({\phi}^{\omega_j}_l \otimes {\phi}^{\omega_{j+h}}_{\lpr})}{{\phi}^{\omega_j}_m \otimes {\phi}^{\omega_{j+h}}_{m^{\prime}}}_{S})
%
%
\\& \scriptstyle
=  O_p\Big( \frac{1}{\sqrt{T}} \sum_{j=1}^{T} \sum_{l,\lpr}  \biginprod{D_{\omega_{j}} \otimes D_{\omega_{j+h}}}{{\big(\hat{\F}_{-\omega_{j}} \widetilde{\otimes}  \hat{\F}_{-\omega_{j+h}} -\F_{-\omega_{j}} \widetilde{\otimes}  \F_{-\omega_{j+h}}\big)({\phi}^{-\omega_j}_l \otimes {\phi}^{-\omega_{j+h}}_{\lpr})}}_{S}\Big), \tageq
 \label{eq:approxdif}
\end{align*}
using that $\inprod{A}{B}_{S}=\sum_{m\in \mathbb{N}} \inprod{A}{\psi_{m\mpr}}_{S}\inprod{\psi_{m\mpr}}{B}_{S}$ for any orthonormal basis $\{\psi_{m\mpr}\}_{m,\mpr \in \mathbb{N}}$ of $S_2$. In other words, the order of the difference is completely determined by the order of the difference when replacing the Kronecker products of the estimated spectral density operators with their empirical counterparts. This finding can be utilized to determine the order of \eqref{eq:J} by decomposing it as follows, and considering each of the terms separately:
\begin{align}
J_1
=&\frac{1}{\sqrt{T}} \sum_{j=1}^{T}  \sum_{l,\lpr}  \inprod{D_{\omega_{j}} \otimes D_{\omega_{j+h}}- \E\big(D_{\omega_{j}} \otimes D_{\omega_{j+h}})}{\phi^{\omega_j}_{l} \otimes \phi^{\omega_{j+h}}_{\lpr}-\E(\hat{\phi}^{\omega_j}_{l} \otimes \hat{\phi}^{\omega_{j+h}}_{\lpr})}_{S}, \label{eq:J1}
\\
J_2 
=&\frac{1}{\sqrt{T}} \sum_{j=1}^{T}  \sum_{l,\lpr}  \inprod{D_{\omega_{j}} \otimes D_{\omega_{j+h}}- \E\big(D_{\omega_{j}} \otimes D_{\omega_{j+h}})}{\E(\hat{\phi}^{\omega_j}_{l} \otimes \hat{\phi}^{\omega_{j+h}}_{\lpr})-\hat{\phi}^{\omega_j}_{l} \otimes \hat{\phi}^{\omega_{j+h}}_{\lpr}}_{S},  \label{eq:J2} 
\\  
J_3=&\frac{1}{\sqrt{T}} \sum_{j=1}^{T}  \sum_{l,\lpr}  \inprod{\E\big(D_{\omega_{j}} \otimes D_{\omega_{j+h}})}{\phi^{\omega_j}_{l} \otimes \phi^{\omega_{j+h}}_{\lpr}-\E(\hat{\phi}^{\omega_j}_{l} \otimes \hat{\phi}^{\omega_{j+h}}_{\lpr})}_{S}, \label{eq:J3}\\
J_4 =& \frac{1}{\sqrt{T}} \sum_{j=1}^{T} \sum_{l,\lpr}  \E( \inprod{D_{\omega_{j}} \otimes D_{\omega_{j+h}})}{\E(\hat{\phi}^{\omega_j}_{l} \otimes \hat{\phi}^{\omega_{j+h}}_{\lpr})-\hat{\phi}^{\omega_j}_{l} \otimes \hat{\phi}^{\omega_{j+h}}_{\lpr}}_{S}.  \label{eq:J4}
\end{align}
The following lemma contains the order of these four terms. 

\begin{lemma} \label{lemJ1}
Under \ref{cumglsp}(12,2),
\begin{align}
\E |J_1| &= O(\frac{1}{b T} +b^2), \label{J1}
\\ \E |J_2| & =\begin{cases}
O(\frac{1}{b T\sqrt{T}})+O(\frac{1}{b T}) & \mbox{ under $H_0$,}
\\O(\frac{1}{b \sqrt{T}})+O(\frac{1}{bT})& \hspace{0.25cm}\mbox{under $H_A$,}
\end{cases} \label{J2}
\\ \E |J_3| & = \begin{cases}
O(\frac{1}{b T\sqrt{T}}+\frac{b^2}{\sqrt{T}}) & \mbox{ under $H_0$,}
\\O(\frac{1}{b \sqrt{T}}+b^2\sqrt{T}) & \hspace{0.25cm}\mbox{under $H_A$,}
\end{cases}\label{J3}
\\ \E |J_4| & =\begin{cases}
O(\frac{1}{\sqrt{b} T}) & \mbox{ under $H_0$.}
\\O(1)& \hspace{0.25cm}\mbox{under $H_A$.}
\end{cases} \label{J4}
\end{align}
\end{lemma}
The proof is relegated to Section \ref{subsec:proof:denominator} of the Online Supplement.

\section{Limiting distribution under $H_A$}\label{sec:HAdist}

\begin{theorem}\label{thm:HAdist}
Under the conditions of Theorem \ref{th:test_alt}, we have, for all $h_i, h_j\in\znum$ with $i,j=1,\ldots, k$,
\begin{align}
 T^{k/2}\cum_{n,r}({\bm{\beta}}_{h_i}^{(T)},\mathcal{B}^{(T))}_{h_j}) =o(1)   
 \qquad(T \to \infty),
\end{align}
where $\cum_{n,r}({\bm{\beta}}_h^{(T)},\mathcal{B}^{(T))}_{h^\prime})$ denotes the joint cumulant
\[
\cum(\underbrace{{\bm{\beta}}_{h}^{(T)},\ldots,{\bm{\beta}}_h^{(T)}}_{n\ \text{times}},\underbrace{\mathcal{B}^{(T))}_{h^\prime},\ldots,\mathcal{B}^{(T))}_{h^\prime}}_{r\ \text{times}})\]
with $0 \le n, r \le k$ such that $n+r =k$.
\end{theorem}

\begin{proof}

We will show that $\sqrt{T}{\bm{\beta}}_h^{(T)}$ and $\sqrt{T}\mathcal{B}^{(T)}_h$ are jointly normal. Using  \eqref{eq:approxdif} and hence that the order of $\mathcal{B}^{(T)}_h$ is determined by the order of
\begin{align*}
\mathcal{V}^{(T)}_h = \frac{1}{{T}} \sum_{j=1}^{T}\inprod{ \E(D_{\omega_{j}} \otimes D_{\omega_{j+h}})}{{\Big(\hat{\F}_{\omega_{j_1}} \widetilde{\otimes}  \hat{\F}_{\omega_{j+h}} -\E\hat{\F}_{\omega_{j}} \widetilde{\otimes}  \hat{\F}_{\omega_{j+h}}\Big)({\phi}^{\omega_j}_l \otimes {\phi}^{\omega_{j+h}}_{\lpr})}}_S,
\end{align*}
 we will show that, for $k>2$,
\begin{align*}
& T^{k/2}\cum_{n,r}({\bm{\beta}}_h^{(T)},\mathcal{V}^{(T))}_h) 
=\cum(\underbrace{{\bm{\beta}}_h^{(T)},\ldots,{\bm{\beta}}_h^{(T)}}_{n\ \text{times}},\underbrace{\mathcal{V}^{(T))}_h,\ldots,\mathcal{V}^{(T))}_h}_{r\ \text{times}})
=o(1),
\end{align*}
where $0 \le n, r \le n$ such that $n+r =k$. First note that the operator $\E[D_{\omega_j} \otimes D_{\omega_{j+h}}]$ is compact and therefore separable. Without loss of generality, in order to ease notation, write therefore $\fdft{}{j_k}= \langle D^{(T)}_{\omega_{j_k}}, \psi_{l}\rangle$ and $\hat{\F}^{(lm)}_{\omega_j}=\inprod{\hat{\F}_{\omega_j}(\psi_{m})}{\psi_{l}}$, where $\{\psi_l\}_{l \in \nnum}$ forms a basis of $H$.
Using then Theorem \ref{prodcumthm} 
\begin{align*}
& T^{k/2}\cum_{n,r}({\bm{\beta}}_h^{(T)},\mathcal{V}^{(T))}_h) \\& 
 =T^{k/2} \sum_{j_1,\ldots, j_{k}} \cum\big(D^{(l_1)}_{\omega_{j_1}} D^{(l_1')}_{-\omega_{j_1+h_1}}, \ldots, D^{(l_{n})}_{\omega_{j_n}} D^{(l_{n}')}_{-\omega_{j_{n}+h_{n}}}, \hat{\F}^{(l_{n+1} m_{n+1})}_{\omega_{j_{n+1}}}\overline{\hat{\F}^{(l_{n+1}' m_{n+1}')}_{\omega_{j_{n+1}+h_{n+1}}}},\ldots,\hat{\F}^{(l_{k} m_{k})}_{\omega_{j_k}}\overline{\hat{\F}^{(l_{k}' m_{k}')}_{\omega_{j_k+h_k}}} \big)
 \\& =T^{-k/2}  \sum_{j_1,\ldots, j_k} \Big(\frac{2\pi}{bT}\Big)^{2r} \prod_{d=n+1}^{k}\sum_{q_d=1}^T K(\frac{\omega_{j_d}-\omega_{q_{2(d-n)-1}}}{b}) K(\frac{\omega_{j_d+h_d}-\omega_{q_{2(d-n)}}}{b})\\& 
 \phantom{T^{-k/2}  \sum_{j_1,\ldots, j_k} \Big(\frac{2\pi}{bT}\Big)^{2m} \prod_{r=n+1}^{m}\sum_{\substack{i=1,\ldots,2m-1}}}\times \sum_{i.p.} \cum(D^{(v_s)}_{\lambda_{{j}_s}}\colon s \in P_1) \cdots \cum(D^{(v_s)}_{\omega_{{j}_s}}\colon s \in P_Q),
\end{align*}
where the summation extends over all indecomposable partitions $P= \{P_1,\ldots, P_Q\}$ of the array
\begin{align*}\begin{array}{llll}
(1,1)&(1,2) & & \\ 
\quad\vdots &\quad \vdots &\\ 
(n,1)&(n,2) & & \\ 
(n+1,1)& (n+1,2) & (n+1,3)& (n+1,4) \\
\quad\vdots &\quad \vdots &\quad\vdots & \\
(k,1)& (k,2) & (k,3)& (k,4),
\end{array}  \tageq \label{arrayHa}\end{align*}
using similar notation as in the proof of Theorem \ref{CLT_alt}. In particular, the value $s=ii^\prime$ corresponds to entry $(i,i^\prime)$ of \eqref{arrayHa}. For a partition $P= \{P_1,\ldots, P_Q\}$, the elements of a set $P_\nu$ will be denoted by $s_{\nu1},\ldots,s_{\nu |P_\nu|}$, with $|P_\nu|$ being the number of elements in $P_\nu$. In this case, we associate with entry $s$ the frequency index $j_s =j_{ii^\prime} =(-1)^{i^\prime-1}(j_i+{h}^{i^\prime-1}_i)$ for $i\le n$; for $i >n$ we associate the frequency index $j_s =q_{ii^\prime } =(-1)^{i^\prime-1}q_{2(i-n)-1+\lfloor i^\prime/3 \rfloor}$ such that $\lambda_{{j}_s}=\frac{2\pi {j}_s }{T}$ and the basis function index $v_s=v_{ii^\prime}=l_{i}^{2-i^\prime}{l_i'}^{i^\prime-1}$ for $i=1,\ldots,k$ and $i^\prime=\{1,2\}$, while for $i^\prime =\{3,4\}$ we set $v_s=v_{ii^\prime}=m_{i}^{4-i^\prime}{m_i'}^{i^\prime-3}$.

For the array to be indecomposable, the rows must hook \citep[pp.\ 20/21]{b81}. Since interest is only in a bound for the partition of highest order, only partitions have to be considered for which each set satisfies $|P_\nu|=2$, since all other partitions will be of lower order. Without loss of generality, consider that row $i$ hooks with $i+1$ for $i = 2,\ldots k-1$ and let the first and the last row hook. In particular, a partition of highest order would be one for which $P_i =\{(i,2) \cup (i+1,1)\}$ for $i = 2,\ldots k$ and $P_1 =\{(1,1) \cup (k,2)\}$ and where the $2r$ variables in the third and fourth columns of the last $r$ rows are decomposable, meaning that $P_{ n+i} =\{ (n+i,3)\cup(n+i,4)\}$ for $i = 1,\ldots r$, s that these latter $r$ sets form proper submanifolds of the frequency manifold. Using Lemma \ref{cumboundglsp} such a partition can be written as
\begin{align*}
 & T^{-k/2}  \sum_{j_1,\ldots, j_k} \Big(\frac{2\pi}{bT}\Big)^{2r} \prod_{d=n+1}^{k}\sum_{q_d=1}^T K(\frac{\omega_{j_d}-\omega_{q_{2(d-n)-1}}}{b}) K(\frac{\omega_{j_d+h_d}-\omega_{q_{2(d-n)}}}{b})\\& 
 \phantom{T^{-k/2}  \sum_{j_1,\ldots, j_k} \Big(\frac{2\pi}{bT}\Big)^{2m} }\times \prod_{\nu=1}^{k}  \bigg[\big(\tilde{\F}^{(v_s)}_{\sum_s {j}_s; \lambda_{{j}_s}}\colon s \in P_\nu\big)+ O\bigg(\frac{1}{T}\bigg) \bigg] \prod_{\nu=n+1}^{n+r}\bigg[\big(\tilde{\F}^{(v_s)}_{0; \lambda_{{j}_s}}\colon s \in P_\nu\big)+ O\bigg(\frac{1}{T}\bigg) \bigg].
 \end{align*}
In exactly $k$ sets of the partition there are exactly $k-1$ equations of the form $y_s = \sum_{s} j_s$. In the above partition, the first $k$ sets yield the following set of equations
\begin{align*}
y_i & =(-1)^i (j_i+h_i -j_{i+1})\quad i =1,\ldots, n-1,\\
\tilde{y}_i & = (-1)^i (q_{2i-1} - q_{2i+1}) \quad i = 1, \ldots r-2,\\
\tilde{y}_{r-1} & = (j_n - q_1),
\\
\tilde{y}_r & = (j_1 -q_{2r-1}).
\end{align*}
By Corollary \ref{cumbounds} these equations correspond to $k-1$ summations out of the total $k+2r$ summations that are bounded. It can be verified that the above set of equations and an iterative change of variables shows that the other $2r+1$ free variables are interrelated via the $2r$ kernel functions. These means that $2r$ sums can at most be of order $bT$, while one of them can be of order $T$. Consequently,
\begin{align*}
T^{k/2}\cum_{n,r}({\bm{\beta}}_h^{(T)},\mathcal{V}^{(T))}_h)=O(T^{-k/2}{(bT)}^{-2r} {(bT)}^{2r} T) = O(T^{-k/2+1}),
\end{align*} 
which converges to zero for $k>2$ as $T \to \infty$, for any choice of $n$ and $r$ such that $n+r=k$.
\end{proof}

\setcounter{section}{0}
\setcounter{equation}{0}
\def\theequation{S\arabic{section}.\arabic{equation}}
\def\thesection{S\arabic{section}}
\setcounter{footnote}{0}

\newpage
\begin{center}
{\LARGE Online Supplement to \\ ``Testing for stationarity of functional time series \\ in the frequency domain''\footnote{AA was partially supported by NSF grants DMS 1305858 and DMS 1407530.  AvD was partially supported by Maastricht University, the contract ``Projet d'Actions de Recherche Concert{\'e}es'' No. 12/17-045 of the ``Communaut{\'e} fran\c{c}aise de Belgique'' and by the Collaborative Research Center 
``Statistical modeling of nonlinear dynamic processes'' (SFB 823, Project A1, C1, A7) of the German
Research Foundation (DFG).}}
 \par\vspace{10pt}
{ Alexander Aue\footnote{Department of Statistics, University of California, Davis, CA 95616, USA, email: \tt{aaue@ucdavis.edu}}
\hspace{30pt}Anne van Delft\footnote{Ruhr-Universit{\"a}t Bochum, Fakult{\"a}t f{\"u}r Mathematik, 44780 Bochum, Germany, email: \tt{Anne.vanDelft@rub.de}
\par }} \vspace{10pt}\\
{\date{\today}} \vspace{35pt}\\
\maketitle
\end{center}
\begin{abstract}
\setlength{\baselineskip}{1.8em}
\textcolor{black}{This supplement contains additional technical material necessary to complete the proofs of theorems of the main paper \citet{avd16_main}. Section \ref{sec:S2stuff} provides notation and results used both in the Appendix and Supplement. Section \ref{subsec:proof:OScumulant} contains the proofs of several auxiliary lemmas stated in Appendix \ref{subsec:proof:cumulants} of the main paper.  Section \ref{asdist} deals with convergence of the finite-dimensional distributions. Section \ref{subsec:proof:denominator} contains auxiliary results to finish the proofs on the limiting distributions of $\hat{\bm{\beta}}_{h,u}^{(T)}$ and $\hat{\bm{\beta}}_{h,x}^{(T)}$. 
Section \ref{sub:covHA} establishes the asymptotic covariance structure of the test under local stationarity. Section \ref{sec:F4proof} derives properties of the tri-spectral density operator estimator, while Section \ref{func_not_mult} contains an example highlighting differences between functional and multivariate mehtods.}
\medskip \\
\noindent {\bf Keywords:} Frequency domain methods, Functional data analysis, Locally stationary processes, Spectral analysis

\noindent {\bf MSC 2010:} Primary: 62G99, 62H99, Secondary: 62M10, 62M15, 91B84
\end{abstract}

\setlength{\baselineskip}{1.8em}



\section{Additional notation and auxiliary results: random $S_2$-valued operators}\label{sec:S2stuff}

\begin{definition}[Bounded maps of operators]\label{def:Krn}
For $A, B, C \in S_{2}(H)$, define the Kronecker product, transpose Kronecker product and Hilbert tensor product, respectively, by
\begin{enumerate}\itemsep-0.75ex
\item[$\mathrm{i.}$] $(A \widetilde{\otimes} B)C = ACB^{\dagger}$;
\item[$\mathrm{ii.}$] $(A \widetilde{\otimes}_{\top} B)C = (A \widetilde{\otimes} \overline{B})\overline{C}^{\dagger}$;
\item[$\mathrm{iii.}$] $(A \otimes B)C  = \biginprod{C}{B}A.$
\end{enumerate} 
\end{definition}

The following lemma introduces a convenient representation of certain moments of Hilbert--Schmidt inner products. 
\begin{lemma}\label{lem:aux}
For a probability space $(\Omega,\mathcal{A},P)$, let $X_i, Y_i$, $i \in \mathbb{N}$, be measurable mappings from $(\Omega,\mathcal{A})$ into $\big(S_2(\otimes_{n=1}^k H_n), \mathscr{B}\big)$, where $\mathscr{B}$ denotes the Borel $\sigma$-algebra in $S_2(\otimes_{n=1}^k H_n)$, i.e., 
$\E\snorm{Y_i}_2<\infty$ and $\E\snorm{X_i}_2<\infty$. Moreover, let $A_i \in \big(S_2(\otimes_{n=1}^k H_n)$. Then, for any $i_1, i_2 \in \mathbb{N}$,
\begin{enumerate}\itemsep-0.75ex
\item[$\mathrm{(i)}$] $\E(\inprod{Y_{i_1}}{A_{i_2}}_{S}) = \inprod{\E Y_{i_1}}{A_{i_2}}_{S}$;
\item[$\mathrm{(ii)}$] $ \E(\inprod{Y_{i_1}}{Y_{i_2}}_{S}) = \tr(\E (Y_{i_1} \otimes Y_{i_2}) )$.
\end{enumerate} 
If moreover, $\E\snorm{Y_i}^2_2<\infty$ and $\E\snorm{X_i}^2_2<\infty$, then
\begin{enumerate}\itemsep-0.75ex
\item[$\mathrm{(iii)}$] $\Cov(\inprod{Y_{i_1}}{A_{i_1}}_{S},\inprod{Y_{i_2}}{A_{i_2}}_{S}) =\inprod{ \Cov( Y_{i_1},Y_{i_2})}{A_{i_1} \otimes A_{i_2}}_{S}$;
\item[$\mathrm{(iv)}$] $ \cov(\inprod{X_{i_1}}{Y_{i_1}}_{S},\inprod{X_{i_2}}{Y_{i_2}}_{S} )=\tr\Big( \Cov(X_{i_1} \otimes Y_{i_1},X_{i_2} \otimes Y_{i_2})\Big)$.
\end{enumerate} 
\end{lemma}

\begin{proof}
(i) Follows directly from Fubini's theorem and the Cauchy--Schwarz inequality. For (ii), note that a basis expansion for each of the two operators $A, B \in S_2(\otimes_{n=1}^k H_n)$ yields
 \begin{align*}
\inprod{A}{B}_{S_2} &= \tr(AB^\dagger)
\\& =\sum_{l,k} \inprod{A \psi_k}{\psi_l}\inprod{ \psi_l}{B \psi_k}
\\& =\sum_{l,k} \inprod{A \otimes B}{(\psi_l\otimes \psi_k) \otimes (\psi_l \otimes \psi_k)}_{S}
\\& =\sum_{l,k}\inprod{\big(A \otimes B\big) (\psi_l \otimes \psi_k)}{\psi_l\otimes \psi_k}
\\& = \tr(A \otimes B).
\end{align*}
The interchange of race and expectation follows again from Fubini's theorem and the Cauchy--Schwarz inequality. For (iii), 
Fubini's theorem implies that, for sequences $(X_j\colon j \in \mathbb{N})$ and $(Y_j\colon j \in \mathbb{N})$ of random operators satisfying $\E\snorm{X_j}^2_2 <\infty$ and $\E\snorm{Y_j}^2_2 <\infty$, 
\begin{align*}
\var(\sum_j\inprod{X_{j}}{Y_{j}}_{S}) &=\sum_{j_1,j_2} \cov(\inprod{X_{j_1}}{Y_{j_1}}_{S},\inprod{X_{j_2}}{Y_{j_2}}_{S} )
\\& =\sum_{j_1,j_2}\tr\Big(  \E [X_{j_1} \otimes Y_{j_1}\, \widetilde{\otimes} \,Y_{j_2} \otimes X_{j_2}] - \E [ X_{j_1} \otimes Y_{j_1}] \, \widetilde{\otimes} \, \E [ Y_{j_2} \otimes X_{j_2}]\Big)
\\& =\sum_{j_1,j_2}\tr\Big( \Cov(X_{j_1} \otimes Y_{j_1},X_{j_2} \otimes Y_{j_2})\Big).
\end{align*}
Additionally note that, for a sequence of deterministic operators $(A_j\colon j\in\mathbb{N})$ with $\snorm{A_j}_2^2<\infty$,
\begin{align*}
\var\bigg(\sum_j\inprod{A_{j}}{Y_{j}}_{S}\bigg) &=\sum_{j_1,j_2} \cov(\inprod{A_{j_1}}{Y_{j_1}}_{S},\inprod{A_{j_2}}{Y_{j_2}}_{S} )
\\& = \sum_{j_1,j_2} \E(\inprod{A_{j_1} \otimes A_{j_2}}{Y_{j_1} \otimes Y_{j_2}}_{S}- \inprod{A_{j_1} \otimes A_{j_2}}{ \E Y_{j_1} \otimes \E Y_{j_2}}_{S}
\\&= \sum_{j_1,j_2}\inprod{ A_{j_1} \otimes A_{j_2}}{ \Cov( Y_{j_1},Y_{j_2})}_{S},
\end{align*}
and similarly $
\var(\sum_j\inprod{Y_{j}}{A_{j}}_{S}) =\sum_{j_1,j_2}\inprod{ \Cov( Y_{j_1},Y_{j_2})}{A_{j_1} \otimes A_{j_2}}_{S}$.
\end{proof}

\begin{proof}[Proof of Lemma \ref{lem:ChebSP}]
Integration by parts
with respect to $\nu$ yields
\begin{align*}
2 |Y_{\nu}(\boldsymbol{\tau})|^2 & \overset{L^2}{=}|Y_{a}(\boldsymbol{\tau})|^2 +| Y^2_b (\boldsymbol{\tau})|^2
+  \int_{a}^{\nu} \frac{\partial}{\partial \alpha} Y_{\alpha} (\boldsymbol{\tau}) \overline{Y_{\alpha} (\boldsymbol{\tau})} d\alpha-  \int_{\nu}^{b}  \frac{\partial}{\partial \alpha}Y_{\alpha} (\boldsymbol{\tau})  \overline{Y_{\alpha} (\boldsymbol{\tau})} d\alpha
 \\& +\int_{a}^{\nu} Y_{\alpha} (\boldsymbol{\tau})  \frac{\partial}{\partial \alpha}\overline{Y_{\alpha} (\boldsymbol{\tau})} d\alpha-  \int_{\nu}^{b}  Y_{\alpha} (\boldsymbol{\tau})\frac{\partial}{\partial \alpha}  \overline{Y_{\alpha} (\boldsymbol{\tau})} d\alpha 
 \\& \le |Y_{a}(\boldsymbol{\tau})|^2 +| Y^2_b (\boldsymbol{\tau})|^2
+  \int_{a}^{b} |\frac{\partial}{\partial \alpha} Y_{\alpha} (\boldsymbol{\tau}) \overline{Y_{\alpha} (\boldsymbol{\tau})} |d\alpha +\int_{a}^{b} |Y_{\alpha} (\boldsymbol{\tau})  \frac{\partial}{\partial \alpha}\overline{Y_{\alpha} (\boldsymbol{\tau})} | d\alpha
\end{align*}
for any $a\le\nu \le b$. This implies in particular
\begin{align*}
2\sup_{a\le\nu \le b}\int_{[0,1]^k} |Y_{\nu}(\boldsymbol{\tau})|^2  d\boldsymbol{\tau} & \le \int_{[0,1]^k}|Y_{a}(\boldsymbol{\tau})|^2 d\boldsymbol{\tau} +\int_{[0,1]^k}| Y^2_b (\boldsymbol{\tau})|^2 d\boldsymbol{\tau}
\\& +\int_{[0,1]^k}  \int_{a}^{b} |\frac{\partial}{\partial \alpha} Y_{\alpha} (\boldsymbol{\tau}) \overline{Y_{\alpha} (\boldsymbol{\tau})} |d\alpha d\boldsymbol{\tau}  +\int_{[0,1]^k} \int_{a}^{b} |Y_{\alpha} (\boldsymbol{\tau})  \frac{\partial}{\partial \alpha}\overline{Y_{\alpha} (\boldsymbol{\tau})} | d\alpha d\boldsymbol{\tau}. 
\end{align*}
Taking expectations on both sides yields
\begin{align*}
2 \E \sup_{a\le\nu \le b}\int_{[0,1]^k} |Y_{\nu}(\boldsymbol{\tau})|^2  d\boldsymbol{\tau} & \le \E \int_{[0,1]^k}|Y_{a}(\boldsymbol{\tau})|^2 d\boldsymbol{\tau} +\E \int_{[0,1]^k}| Y^2_b (\boldsymbol{\tau})|^2 d\boldsymbol{\tau}
\\& +\E \int_{[0,1]^k}  \int_{a}^{b} |\frac{\partial}{\partial \alpha} Y_{\alpha} (\boldsymbol{\tau}) \overline{Y_{\alpha} (\boldsymbol{\tau})} |d\alpha d\boldsymbol{\tau}  +\E \int_{[0,1]^k} \int_{a}^{b} |Y_{\alpha} (\boldsymbol{\tau})  \frac{\partial}{\partial \alpha}\overline{Y_{\alpha} (\boldsymbol{\tau})} | d\alpha d\boldsymbol{\tau}. 
\end{align*}
Tonelli's theorem allows to interchange the integrals in the last two terms from which we find
\begin{align*}
&2\E \sup_{a\le\nu \le b}\|Y_{\nu}\|_2^2\\ 
 & \le \E\|Y_{a}\|_2^2 +\E \| Y_b\|_2^2+ \int_{a}^{b} \E \int_{[0,1]^k}  |\frac{\partial}{\partial \alpha} Y_{\alpha} (\boldsymbol{\tau}) \overline{Y_{\alpha} (\boldsymbol{\tau})} | d\boldsymbol{\tau} d\alpha + \int_{a}^{b} \E \int_{[0,1]^k}|Y_{\alpha} (\boldsymbol{\tau})  \frac{\partial}{\partial \alpha}\overline{Y_{\alpha} (\boldsymbol{\tau})} | d\boldsymbol{\tau} d\alpha
\\& 
\le \E\|Y_{a}\|_2^2 +\E \| Y_b\|_2^2 + \int_{a}^{b} \sqrt{\E  \|\frac{\partial}{\partial \alpha} Y_{\alpha}\|^2_2} \sqrt{\E\| \overline{Y_{\alpha}} \|_2^2}+ \int_{a}^{b} \sqrt{\E\|Y_{\alpha}\|^2_2}d\alpha  \sqrt{\E\| \frac{\partial}{\partial \alpha}\overline{Y_{\alpha}}\|^2_2}d\alpha
\end{align*}
where the Cauchy--Schwarz inequality was applied twice to obtain the last inequality.
\end{proof}

\section{Properties of functional cumulants under local stationarity}
\label{subsec:proof:OScumulant}

\begin{lemma} \label{cumdiffglsp}
Let \ref{cumglsp}(k,1) be satisfied and let $\mathcal{C}_{u;t_1,\ldots,t_{k-1}}$ be as in \eqref{eq:tvsdoker} . 
Then,  
\begin{align*} \label{eq:cumdiffglsp}
\Big\|\cum\big(\XT{t_1},\ldots,\XT{t_{k-1}},\XT{t_{k}}\big) -{\mathcal{C}}_{{t_1}/{T};t_1-t_k,\ldots,t_{k-1}-t_k}\Big\|_2  \le \bigg(\frac{k}{T}+\sum_{j=1}^{k-1} \Big|\frac{t_j-t_k}{T}\Big|\bigg)\|\kappa_{k;t_1-t_k,\ldots,t_{k-1}-t_k}\|_2. 
\end{align*}
\end{lemma}

\begin{proof}
By linearity of the cumulant operation, consecutively taking differences leads, by equation \eqref{eq:repstatap} of the main paper and Minkowski's inequality, to
\begin{align*}
\Big\|\cum\big(\XT{t_1},\ldots,\XT{t_{k}}\big) - \cum\big(\Xuu{{t_1}/{T}}{1},\ldots,\Xuu{{t_{k}}/{T}}{k}\big)\Big\|_2 
& \le K \frac{k}{T}\|\kappa_{k;t_1-t_k,\ldots,t_{k-1}-t_k}\|_2,
\end{align*}
using part (i) of \ref{cumglsp}. By \eqref{eq:repstatap},
\begin{align}
 \Xuu{{t_j}/{T}}{j}- \Xuu{{t_k}/{T}}{k} = \frac{(t_j-t_k)}{T}Y_{t_j}^{({t_j}/{T},{t_k}/{T})}.
\end{align}
Similarly,
\begin{align*}
\Big\|\cum&\big( \Xuu{{t_1}/{T}}{1},\ldots, \Xuu{{t_k}/{T}}{k})-  c_{{t_1}/{T};t_1-t_k,\ldots,t_{k-1}-t_k}\Big\|_2 
\le \sum_{j=1}^{k-1}\frac{|t_j-t_k|}{T}\|\kappa_{k;t_1-t_k,\ldots,t_{k-1}-t_k}\|_2,
\end{align*}
which follows from part (iii) of \ref{cumglsp}. Minkowski's inequality then implies the lemma.
\end{proof}

\begin{lemma}
\label{tvsdo}
Consider a sequence of functional processes $(\XT{t}\colon t\leq T, T \in \nnum)$ as in Definition \ref{lsp} which satisfies \ref{cumglsp}(2,2). Then, $(\XT{t}\colon t\leq T, T \in \nnum)$ uniquely characterizes the time-varying local spectral density operator 
\begin{align} 
\mathcal{F}_{u,\omega} = \frac{1}{2 \pi}\sum_{h \in \mathbb{Z}} \mathcal{C}_{u,h}e^{-\im \omega h},
\end{align}
which belongs to $S_2(H)$. Denoting by $(u,\omega) \mapsto \frac{\partial^{i+j}  }{\partial u^i \partial \omega^{j}}\F_{u,\omega}$ the derivative map of the operator-valued function $\F_{u,\omega}$ of order $i$ in $u$-direction and of order $j$ in $\omega$-direction, we have
\begin{enumerate}\itemsep-.3ex
\item[(i)]  $\sup_{u,\omega}\bigsnorm{\frac{\partial^i  }{\partial u^i}\F_{u, \omega}}_2 < \infty$ for $i = 1,2$,
\item[(ii)] $\sup_{u,\omega}\bigsnorm{\frac{\partial^i  }{\partial^i{\omega}} \F_{u, \omega}}_2 < \infty$ for $i = 1,2$,
\item[(iii)] $\sup_{u,\omega}\bigsnorm{\frac{\partial^2  }{\partial {\omega}\partial u} \F_{u, \omega}}_2 < \infty$.
\end{enumerate}
\end{lemma}
\begin{proof}
Using Lemma \ref{cumdiffglsp}, 
it can be shown that 
$(X_{t}^{T}\colon t\leq T,T\in\nnum)$ uniquely determines the time-varying spectral density operator, that is, 
\begin{align}
\int^{\pi}_{-\pi}\snorm{ \mathcal{F}^{(T)}_{u,\omega}-\mathcal{F}_{u,\omega} }^2_2 \,d \omega =o(1) 
\qquad (T \to \infty). 
\end{align}
Existence of the derivatives follows from the dominated convergence theorem, justified by \ref{cumglsp} (iv) and \eqref{eq:kapmix}, and the product rule for differentiation in Banach spaces \citep{n69}.
\end{proof}

\begin{proof}[Proof of Lemma \ref{cumboundglsp}]
The first line of \eqref{eq:cumboundglsp} 
follows on replacing the cumulants $\cum(\XT{t_1},\ldots,\XT{t_{k-1}},\XT{t_{k}}) $ with $\mathcal{C}_{{t_k}/{T};t_1-t_k,\ldots,t_{k-1}-t_k}$ and Lemma \ref{cumdiffglsp}. 
The second line follows because the discretization of the integral is an operation of order $O({T^{-2}})$. 

Using part (iv) of \ref{cumglsp}, it is seen that
the kernel of $u \mapsto \frac{\partial}{\partial u}\Floc{u}{k}{}$ satisfies
\begin{align*} 
\Big\|\sup_u \frac{\partial}{\partial u} f_{u;\omega_{1},\ldots,\omega_{{k-1}}}\Big\|_2 
\le \frac{1}{{(2 \pi)}^{k-1}} \sum_{t_1,\ldots,t_k}\|\kappa_{k;t_1-t_k,\ldots,t_{k-1}-t_k}\|_2 < \infty.
\end{align*}
The dominated convergence theorem therefore yields
\begin{align} 
\label{eq:FCboundsglsp1}
\sup_{u,\omega_1,\ldots,\omega_{k-1}}\Big{\|}\frac{\partial }{\partial u} f_{u, \omega_1,\ldots, \omega_{k-1}} \Big{\|}_2 < \infty.
\end{align}
Finally, integration by parts for a periodic function in $L^2([0,1]^k)$ with existing $n$-th directional derivative in $u$, yields
\begin{align*}
& \|\tilde{f}_{s; \omega_{j_1},\ldots, \omega_{j_{k-1}}}\|^2_2 \\
& =\int_{[0,1]^k} \bigg{|}\bigg[ \frac{\frac{\partial^{n-1} }{\partial u^{n-1}} \flocatFour{u}{k}{j}(\boldsymbol{\tau})}{{(-\im 2\pi s)}^{n-1}}e^{-\im s 2\pi u}\bigg]_0^1 - \int_0^1 \frac{e^{-\im s 2\pi u}}{{(-\im 2 \pi s)}^{n}}\frac{\partial^{n} }{\partial u^{n}} \flocatFour{u}{k}{j}(\boldsymbol{\tau})du\bigg{|}^2 d\boldsymbol{\tau} \displaybreak[0] \\
& =
\int_{[0,1]^{k+2}}  \frac{1}{ (2 \pi s)^{2n}}e^{\im 2\pi s (u-v)}\frac{\partial^2  }{\partial u^2} \flocatFour{u}{k}{j}(\boldsymbol{\tau}) \frac{\partial^2  }{\partial v^2} \flocatFour{v}{k}{j}(\boldsymbol{\tau}) d\boldsymbol{\tau} du dv\\
&\le  \frac{1}{ (2 \pi s)^{2n}}\int_{[0,1]^{2}} \bigg\|\frac{\partial^2  }{\partial u^2}  \flocatFour{u}{k}{j}\bigg\|_2 \bigg\|\frac{\partial^2  }{\partial v^2} \flocatFour{v}{k}{j}\bigg\|_2 du dv \displaybreak[0]\\
& \le \frac{1}{  (2 \pi s)^{2n}} \bigg(\sup_{u} \bigg\|\frac{\partial^2  }{\partial u^2}  \flocatFour{u}{k}{j}\bigg\|_2\bigg)^2 < \infty,
\end{align*}
where the Cauchy--Schwarz inequality was applied in the second-to-last equality. The interchange of integrals is justified by Fubini's theorem. Thus,
\begin{align}  \label{eq:FCboundsgslp2}
\sup_{\omega_1,..,\omega_{k-1}} \|\tilde{f}_{s; \omega_{j_1},\ldots, \omega_{j_{k-1}}} \|_2 
\le \frac{1}{  (2 \pi )^{2n}}  \sup_{u,\omega_1,\ldots,\omega_n}
\bigg{\|}\frac{\partial^n  }{\partial u^n}  \flocatFour{u}{k}{j} \bigg{\|}_2 |s|^{-n}
\end{align}
and the proof is complete.
\end{proof}

\begin{proof}[Proof of Corollary \ref{cumbounds}]
Part (i) follows directly from equation \eqref{eq:FCboundsgslp2}, the isometry with the Hilbert--Schmidt class, and part (iv) of \ref{cumglsp}.
To elaborate on part (ii) of the corollary, observe that 
\begin{align*}
\snorm{\tilde{\mathcal{F}}_{0:\omega}}_2  \le \sup_{\omega,u}\snorm{ \mathcal{F}_{u,\omega}}_2 <\sum_h \| \kappa_{2,h}\|_2< \infty.
\end{align*} 
The $p$-harmonic series for $p=2$ then yields
\begin{align} 
\label{eq:FCsumbound}
\sup_{\omega} \sum_{s \in \znum} \snorm{\tilde{\mathcal{F}}_{s;\omega} }_2 
\le  \sum_h \| \kappa_{2,h}\|_2 \bigg(1+  \frac{1}{  (2 \pi )^{4}}  \frac{\pi^2}{3}\bigg) 
< \infty,
\end{align}
where the constant ${(2 \pi )^{-4}}$ is implied by \eqref{eq:FCboundsgslp2}. 
\end{proof}
\section{Proof of Theorem \ref{cumconv}: Convergence of finite-dimensional distributions}
\label{asdist}

We shall now prove Theorem \ref{cumconv}, which is repeated here for convenience.
\begin{theorem}\label{asdistthmSUPP}
Under the conditions of Theorem \ref{CLT_alt}, we have for all $l_i,l_i^\prime \in \nnum$, $h_i=1,\ldots,T-1$, $i=1,\ldots, k$ and $k \ge 3$,
\begin{align*} 
\cum \Big( \ET_{h_1} (\psi_{l_1 l_1^\prime}), \ldots, \ET_{h_k} (\psi_{l_k l_k^\prime}) \Big)  =o(1) 
\qquad(T \to \infty).
\end{align*}
\end{theorem}
\begin{proof}
We first provide the outset and then derive the result under local stationarity as this encompasses the stationary case. \smallskip

\noindent {\it Preliminaries.} 
As explained in Section \ref{sec:proofs:weak_conv} of the Appendix, it will be shown that the finite-dimensional distributions of $\ET$ converge to a Gaussian distribution by proving that the higher-order cumulants of the terms $ \ET(\psi_{l {l'}}) = \inprod{\ET}{\psi_{l {l'}}}$ vanish asymptotically. To formulate this, consider an array of the form
\begin{equation}
\label{eq:cumtable}
\begin{matrix}
  (1,1) & (1,2)\\
  \vdots  & \vdots \\
  (k,1) & (k,2)
 \end{matrix}
\end{equation}
and let the value $s=ii^\prime$ correspond to entry $(i,i^\prime)$. For a partition $P= \{P_1,\ldots, P_Q\}$, the elements of a set $P_q$ will be denoted by $s_{q1},\ldots,s_{qm_q}$ where $|P_q|=m_q$  is the corresponding number of elements in $P_q$. Associate with entry $s$ the frequency index $j_s =j_{ii^\prime} =(-1)^{i^\prime-1}(j_i+{h}^{i^\prime-1}_i)$, Fourier frequency $\lambda_{{j}_s}=\frac{2\pi {j}_s }{T}$ and the basis function index $v_s=v_{ii^\prime}=l_{i}^{2-i^\prime}{l_i'}^{i^\prime-1}$ for $i=1,\ldots,k$ and $i^\prime=1,2$. \smallskip

\noindent{\it Proof.} To ease notation, write $\fdft{}{j_k}= \langle D^{(T)}_{\omega_{j_k}}, \psi_{l}\rangle$ and $( \F^{(v_s)}_{{t}/{T}; \lambda_{{j}_s}}\colon s \in P_q) =\inprod{f_{{t}/{T};\lambda_{{j}_{q1}},\ldots,\lambda_{{j}_{q{m_q-1}}}}}{\otimes_{i^\prime=1}^{m_q} \psi_{v_{s_{qi^\prime}}}}$
where, by Corollary \ref{cumbounds}, the latter quantities are well-defined both under $H_A$ and $H_0$  since the convergence in norm implies convergence of the coefficients. Furthermore, since $X_t \in L^2(\Omega)$, we have $\E\|D_\omega\|^2_2 < \infty$ and therefore the fDFT's are in $L^2_\cnum([0,1])$. Therefore, we can consider an application of the product theorem for cumulants yields on the coefficients,
\begin{align*}
\cum&\Big(\sum_{j_1=1}^{T}D^{(l_1)}_{\omega_{j_1}} D^{(l_1')}_{-\omega_{j_1+h_1}},\ldots, \sum_{j_k=1}^{T}D^{(l_{k})}_{\omega_{j_k}} D^{((l_{k}')}_{-\omega_{j_k+h_k}} \Big)\\
&=  \sum_{j_1,\ldots, j_k} \sum_{i.p.} \cum(D^{(v_s)}_{\lambda_{{j}_s}}\colon s \in P_1) \cdots \cum(D^{(v_s)}_{\lambda_{{j}_s}}\colon s \in P_Q), 
\end{align*}
where the summation extends over all indecomposable partitions $P= \{P_1,\ldots, P_Q\}$ of \eqref{eq:cumtable}. Because $X_t$ has zero-mean, the number of elements within each set must satisfy $m_q \ge 2$ and thus $Q \le k$. By Lemma \ref{cumbounds} we obtain
we obtain
\begin{align*}
\frac{1}{T^{k/2}}  \sum_{j_1,\ldots, j_k=1}^{T}\sum_{i.p.} 
&\prod_{q=1}^Q\cum\big(D^{(v_s)}_{\lambda_{{k}_s}}\colon s \in P_q\big)
\\& =
\frac{1}{T^{k/2}} \sum_{j_1,\ldots, j_k=1}^{T}\sum_{i.p.} \prod_{q=1}^Q 
\bigg[\frac{(2\pi)^{m_q/2-1}}{T^{m_q/2-1}} \big(\tilde{\F}^{(v_s)}_{\sum_s {j}_s; \lambda_{{j}_s}}\colon s \in P_q\big) 
+ O\bigg(\frac{1}{T^{m_q/2}}\bigg) \bigg].
\end{align*} 
Note that, by Corollary \ref{cumbounds} 
and the Cauchy--Schwarz inequality,
\begin{align*}
\sum_{j=1}^{T}\big|\tilde{\F}^{(v_s)}_{\sum_s {j}_s; \lambda_{{j}_s}}\big| 
\le\sup_{\omega}\sum_{j \in \znum}\|\tilde{\F}_{j; \omega}\|_2 \prod_{i=1}^{m_q}\|\psi_{v_{qi}} \|_2<\infty,
\qquad s \in P_q,
\end{align*}
for all $q =1,\ldots, Q$. If $Q < k$ or if $Q = k$ and there are $h_{i_1}$ and $h_{i_2}$ such that $h_{i_1} \ne h_{i_2}$ for $i_1, i_2 \in \{1,\ldots, k\}$ within the same set, then there is dependence on $Q$ of the $k$ sums $j_1,\ldots,j_n$. On the other hand, if the size of the partition is equal to $k$ and $h_{i_1}= h_{i_2}$ for all  $i_1, i_2 =1,\ldots, k$, then there are  $Q-1$ constraints on $j_1,\ldots,j_n$. Thus
it follows that the order is 
\begin{align*}
O(T^{-k/2}  T^{k-Q+1} T^{-2k/2+Q}) =  O(T^{-k/2+1}).
\end{align*}
The cumulants of order $k\geq 3$ will therefore tend to 0 as $T \to \infty$. 
\end{proof}

\section{Dealing with condition $C_u$}
\label{sec:increase}

Condition $C_u$ regulates how to handle the number of fPCs included in the unstandardized test statistics. It specifically allows the number of fPCs to be increased logarithmically with sample size. This small section provides a heuristic argument for why this does not change the asymptotics. Note that one can focus without loss of generality on the fixed $L=\min_jL_j$ in place of the frequency-dependent truncations $L_j$, as the difference is asymptotically negligible so that 
\begin{align*} 
\frac 1T\sum_{j=1}^TD_{\omega_j}^{(T)}\otimes D_{\omega_{j+h}}^{(T)}
&\approx\frac 1T\sum_{j=1}^T\sum_{l=1}^{L}\sum_{l^\prime=1}^{L}
\langle D_{\omega_j}^{(T)},\phi_l^{\omega_j}\rangle
\overline{\langle D_{\omega_{j+h}}^{(T)}, \phi_{l^\prime}^{\omega_{j+h}}\rangle}
\phi_l^{\omega_j}\otimes \phi_{l^\prime}^{\omega_{j+h}}.
\end{align*}
Subsequently sending $L$ to $\infty$ in logarithmic fashion does not alter the limit distribution.

\section{Auxiliary proofs for Theorem \ref{boundstatcase} and Theorem \ref{boundlocstatcase} } 
\label{subsec:proof:denominator}

\subsection{Proof of Lemma \ref{lemJ1}}

In the proof of Lemma \ref{lemJ1}, we shall make use Lemma of \ref{lem:aux} and of the following result. 
\begin{lemma}\label{lem:boundsFtensors}
If 
 \ref{cumglsp}(4,\,2) is satisfied, then 
\begin{align*}
\sup_{\omega_1,\omega_2}\snorm{\E \hat{\F}_{\omega_{1}} \widetilde{\otimes}  \hat{\F}_{\omega_{2}} -\F_{\omega_{1}} \widetilde{\otimes}  \F_{\omega_2}}_2 = O\bigg(\frac{1}{b T} +b^2\bigg).
\end{align*} 
\end{lemma}
\begin{proof}
The proof mimics the first part of the proof of Theorem \ref{UnifCon} and is therefore omitted.
\end{proof}

\begin{proof}[Proof of ~Lemma \ref{lemJ1}] 
We provide the proof for each of the four terms separately below. To ease notation, we shall derive the result for fixed $l, \lpr$. Note that this is without loss of generality under conditions $C_u$ and $C_s$. Under $C_s$, the number of directions $L_j$ and $L_{j+h}$ are finite, whilst under $C_u$, the number of directions are allowed to go to infinity in a controlled manner, in which case the directions become independent of frequency and can be taken out of the sum over frequency; see Section \ref{sec:increase}.
\begin{proof}[Proof of ~\eqref{J1}]
Using \eqref{eq:approxdif} we obtain for $ \eqref{eq:J1}$ 
\[
J_1=\frac{1}{\sqrt{T}} \sum_{j=1}^{T}    \inprod{D_{\omega_{j}} \otimes D_{\omega_{j+h}}- \E\big(D_{\omega_{j}} \otimes D_{\omega_{j+h}})}{\phi^{\omega_j}_{l} \otimes \phi^{\omega_{j+h}}_{\lpr}-\E(\hat{\phi}^{\omega_j}_{l} \otimes \hat{\phi}^{\omega_{j+h}}_{\lpr})}_{S} = O_p(\breve{J}_1),
\]
where 
\begin{align*}
 \breve{J}_1 \!=\!  \frac{1}{\sqrt{T}} \sum_{j=1}^{T}  \biginprod{D_{\omega_{j}} \otimes D_{\omega_{j+h}} \!-\! \E\big(D_{\omega_{j}} \otimes D_{\omega_{j+h}})}{\big(\F_{-\omega_{j}} \widetilde{\otimes}  \F_{-\omega_{j+h}} \!-\!\E (\hat{\F}_{-\omega_{j}} \widetilde{\otimes}  \hat{\F}_{-\omega_{j+h}}) \big)\! ({\phi}^{-\omega_j}_l \otimes {\phi}^{-\omega_{j+h}}_{\lpr})}_{S}.
\end{align*}
Note that $E|J_1| \le \sqrt{\E |J_1|^2} =\sqrt{\var(J_1) + (\E [J_1])^2}$ and therefore consider bounds on $\E  \breve{J}_1$ and $\Var( \breve{J}_1)$. Using Lemma \ref{lem:aux}(i) it is immediate that $\E  \breve{J}_1=0$. 
Secondly, using \ref{lem:aux}(iii) \begin{align*}
&\var(\breve{J}_1) = \frac{1}{T} \sum_{j_1,j_2=1}^{T} \Big \langle \cov\big( D_{\omega_{j_1}} \otimes D_{\omega_{j_1+h}}- \E\big(D_{\omega_{j_1}} \otimes D_{\omega_{j_1+h}}), D_{\omega_{j_2}} \otimes D_{\omega_{j_2+h}}- \E\big(D_{\omega_{j_2}} \otimes D_{\omega_{j_2+h}})\big) ,
\\&
\overline{\big(\E\hat{\F}_{\omega_{j_1}} \widetilde{\otimes} \hat{\F}_{\omega_{j_1+h}} -\F_{\omega_{j_1}} \widetilde{\otimes}  \F_{\omega_{j_1+h}}\big)({\phi}^{\omega_{j_1}}_l \otimes {\phi}^{\omega_{j_1+h}}_{\lpr})
 {\otimes} \,\big(\E\hat{\F}_{\omega_{j_2}} \widetilde{\otimes} \hat{\F}_{\omega_{j_2+h}} -\F_{\omega_{j_2}} 
 \widetilde{\otimes}  \F_{\omega_{j_2+h}}\big)({\phi}^{\omega_{j_2}}_n \otimes {\phi}^{\omega_{j_2+h}}_{o})
} \Big\rangle_{S},
\end{align*}
As shown in Section \ref{sec:proof:first+second} of the main paper,
\begin{align*}
& \frac{1}{{T}}\sum_{j_1,j_2}\snorm{\cov(D_{\omega_{j_1}} \otimes D_{\omega_{j_1+h}},D_{\omega_{j_2}} \otimes D_{\omega_{j_2+h}})}_2 =O(1).
\end{align*}
Therefore, the Cauchy--Schwarz inequality, separability of the tensor in norm, and Lemma \ref{lem:boundsFtensors}(i) imply
\begin{align*}
&\var(\breve{J}_1) \\
& \le  \frac{1}{{T}}\sum_{j_1,j_2}\snorm{\cov(D_{\omega_{j_1}} \otimes D_{\omega_{j_1+h}},D_{\omega_{j_2}} \otimes D_{\omega_{j_2+h}})}_2
\Bigsnorm{ \overline{\big(\E\hat{\F}_{\omega_{j_1}} \widetilde{\otimes} \hat{\F}_{\omega_{j_1+h}} -\F_{\omega_{j_1}} \widetilde{\otimes}  \F_{\omega_{j_1+h}}\big)({\phi}^{\omega_{j_1}}_l \otimes {\phi}^{\omega_{j_1+h}}_{\lpr})}}_2
\\& \phantom{ \frac{1}{{T}}\sum_{j_1,j_2}} \times
\Bigsnorm{\overline{\big(\E\hat{\F}_{\omega_{j_2}} \widetilde{\otimes} \hat{\F}_{\omega_{j_2+h}} -\F_{\omega_{j_2}} \widetilde{\otimes}  \F_{\omega_{j_2+h}}\big)({\phi}^{\omega_{j_2}}_n \otimes {\phi}^{\omega_{j_2+h}}_{0})
}}_2
\\& \le 
\frac{c}{T} \sum_{j_1,j_2=1}^{T}  \Bigsnorm{ \cov\big( D_{\omega_{j_1}} \otimes D_{\omega_{j_1+h}}, D_{\omega_{j_2}} \otimes D_{\omega_{j_2+h}}) }_2\sup_{\omega}\Bigsnorm{\E\hat{\F}_{\omega} \widetilde{\otimes} \hat{\F}_{\omega+{h}} -\F_{\omega} \widetilde{\otimes}  \F_{\omega+{h}}}^2_2
\\&
=O(1)O\bigg(\frac{1}{bT}+b^2\bigg)^2.
\end{align*}
Consequently, both under $H_0$ and $H_A$, $\E |J_1| = O(\frac{1}{b T} +b^2).$
\end{proof}
\begin{proof}[Proof of ~\eqref{J2}]
 Use the above to write $ \eqref{eq:J2}$ as
 \begin{align*}
J_2
&\!=\!\frac{1}{\sqrt{T}} \sum_{j=1}^{T}  \inprod{D_{\omega_{j}} \otimes D_{\omega_{j+h}}\!-\! \E\big(D_{\omega_{j}} \otimes D_{\omega_{j+h}})}{\E(\hat{\phi}^{\omega_j}_{l} \otimes \hat{\phi}^{\omega_{j+h}}_{\lpr})\!-\!\hat{\phi}^{\omega_j}_{l} \otimes \hat{\phi}^{\omega_{j+h}}_{\lpr}}_{S} =O_p(\breve{J}_2),
\end{align*} 
where
 \begin{align*}
 \breve{J}_2\!=\! \frac{1}{\sqrt{T}} \sum_{j=1}^{T}\biginprod{D_{\omega_{j}} \otimes D_{\omega_{j+h}} \!-\! \E\big(D_{\omega_{j}} \otimes D_{\omega_{j+h}})}{\big(\E (\hat{\F}_{-\omega_{j}} \widetilde{\otimes}  \hat{\F}_{-\omega_{j+h}}) \!-\!\hat{\F}_{-\omega_{j}} \widetilde{\otimes}  \hat{\F}_{-\omega_{j+h}} \big)({\phi}^{-\omega_j}_l \otimes {\phi}^{-\omega_{j+h}}_{\lpr})}_{S} 
\end{align*}
To consider first the order of $\E \breve{J}_2$, use Lemma \ref{lem:aux}(ii) which requires to consider the order in $S$ of the operator 
\begin{align*}
\E& \Big( D_{\omega_{j}} \otimes D_{\omega_{j+h}}- \E\big(D_{\omega_{j}} \otimes D_{\omega_{j+h}}) \otimes \big(\E\hat{\F}_{-\omega_{j}} \widetilde{\otimes}  \hat{\F}_{-\omega_{j+h}} -\hat{\F}_{-\omega_{j}} \widetilde{\otimes}  \hat{\F}_{-\omega_{j+h}}\big) \Big)
\\& = \Cov\Big( D_{\omega_{j}} \otimes D_{\omega_{j+h}}, \big(\hat{\F}_{-\omega_{j}} \widetilde{\otimes}  \hat{\F}_{-\omega_{j+h}}\big) \Big)\\& 
 =\frac{1}{{(bT)}^2}\sum_{k_1,k_2} K\bigg(\frac{k_1}{b}\bigg)K\bigg(\frac{k_2}{b}\bigg) \Cov\Big( D_{\omega_{j}} \otimes D_{\omega_{j+h}}, D_{\omega_{k_1-j}} \otimes  D_{\omega_{k_1-j}}  \widetilde{\otimes} D_{\omega_{k_2-j-h}} \otimes  D_{\omega_{k_2-j-h}}  \Big).
\end{align*}
Using then Theorem \ref{prodcumthm}, we are looking for all indecomposable partitions of the array
\[\begin{matrix}
\underbrace{D_{\omega_{j}}}_{1}&\underbrace{D_{-\omega_{j+h}}}_2 & & \\ 
\underbrace{D_{\omega_{k_1-j}}}_{3}&\underbrace{D_{-\omega_{k_2-j-h}}}_4 & \underbrace{D_{\omega_{j-k_1}}}_5& \underbrace{ D_{\omega_{k_2-j-h}}}_6.
\end{matrix}\]
Careful consideration shows that many terms are of lower order. Those that remain are second-order cumulant tensors with a partition that hooks the rows but keeps as many elements with conjugate pairs in the same set. We focus on the partitions of highest order. These are of the same order as the partition $S_{(13)(25)(46)}$. Lemma \ref{FCboundssd} implies under $H_0$ that we obtain
\begin{align*}
S_{(13)(25)(46)}\Bigg(\frac{1}{{(bT)}^2}\sum_{k_1,k_2}& K\bigg(\frac{k_1}{b}\bigg)K\bigg(\frac{k_2}{b}\bigg) \frac{\Delta_T^{(\omega_{k_1})} \Delta_T^{(\omega_{k_1+h})}}{T^2} \\
&\times (\F_{\omega_j}+ R_{T,2}) \otimes (\F_{-\omega_{j+h}}+R_{T,2}) \otimes (\F_{\omega_{j+h-k_2}} +R_{T,2})\Bigg)
\end{align*}
which is of order $O({b^{-1}T^{-2}})$ in Hilbert--Schmidt norm since $h \ne 0$. Under $H_A$, this partition is given by
\begin{align*}
S_{(13)(25)(46)}\bigg(\frac{1}{{(bT)}^2}\sum_{k_1,k_2}\!\! K\bigg(\frac{k_1}{b}\bigg)K\bigg(\frac{k_2}{b}\bigg) (\tilde{\F}_{k_1:\omega_j} +R_{T,2} )\otimes (\tilde{\F}_{-k_1-h:-\omega_{j+h}}+R_{T,2}) \otimes (G_{\omega_{j+h-k_2}} +R_{T,2})\!\bigg)
\end{align*}
and using Corollary \ref{cumbounds} and a similar derivation as in the proof of Theorem \ref{UnifCon} this is of order $O((b T)^{-1})$ in Hibert--Schmidt norm under $H_A$. By Lemma \ref{lem:aux}(ii), an application of the Cauchy--Schwarz inequality therefore implies that $| \E \breve{J}_2|$ is bounded by
\begin{align*}
&\bigg| \tr\bigg(\frac{1}{{(bT)}^2}\sum_{k_1,k_2} K\bigg(\frac{k_1}{b}\bigg)K\bigg(\frac{k_2}{b}\bigg) \\
&\qquad\times\Cov\Big( D_{\omega_{j}} \otimes D_{\omega_{j+h}}, (D_{\omega_{k_1-j}} \otimes  D_{\omega_{k_1-j}})  \widetilde{\otimes} (D_{\omega_{k_2-j-h}} \otimes  D_{\omega_{k_2-j-h}}) ( {\phi}^{-\omega_j}_l \otimes {\phi}^{-\omega_{j+h}}_{l^{\prime}}) \Big) \bigg)\bigg|
\\& 
=O\bigg(\frac{1}{b \sqrt{T}T} \sup_{\omega}\snorm{\F_{\omega}}_2 \snorm{ \F_{\omega} ( {\phi}^{\omega}_l)}_2  \snorm{ \F_{\omega} ( {\phi}^{\omega}_\lpr)}_2 \bigg) 
= O\bigg(\frac{1}{b T\sqrt{T}}\bigg).
\end{align*}
Using Corollary \ref{cumbounds}, a similar reasoning shows that $\E J_2 =O(b^{-1}T^{-1/2})$, under $H_A$. We now investigate the variance of $J_2$. Using Lemma \ref{lem:aux}, this equals
\begin{align*}
\tr\bigg(\Var\Big(\frac{1}{\sqrt{T}} \sum_{j_1=1}^{T} \Big[&D_{\omega_{j}} \otimes D_{\omega_{j+h}} \\ &-\E\big(D_{\omega_{j}} \otimes D_{\omega_{j+h}})\otimes (\E (\hat{\F}_{-\omega_{j_1}} \widetilde{\otimes}  \hat{\F}_{-\omega_{j+h}}) -\hat{\F}_{-\omega_{j}} \widetilde{\otimes}  \hat{\F}_{-\omega_{j+h}})( {\phi}^{-\omega_j}_l \otimes {\phi}^{-\omega_{j+h}}_{l^{\prime}})\Big]\Big)  \bigg).
\end{align*}
Observe that the correction for the means implies a covariance structure of the form 
\begin{align*}
\Cov&( (X-\E X) \otimes (W-\E W), (Y-\E Y) \otimes (Z-\E Z)) 
\\ = & \cum(X \otimes W -\E [X] \otimes W-X\otimes \E [W]), Y \otimes Z -\E [Y] \otimes Z-Y \otimes \E [Z]) 
\\ =&\cum(X\otimes W,Y\otimes Z) -S_{3124}\cum(Y) \otimes \cum( X\otimes W, Z) -\cum(X \otimes W,Y)\otimes \cum(Z)
\\& -\cum(X) \otimes \cum(W, Y\otimes Z)  -S_{2134}\cum(W)\otimes\cum(X,Y\otimes Z)
\\&+S_{1243}\cum(X)\otimes\cum(W,Z)\otimes\cum(Y)+\cum(X)\otimes\cum(W,Y)\otimes\cum(Z) 
\\& +S_{2143}\cum(W)\otimes\cum(X,Z)\otimes\cum(Y) + S_{2134}\cum(W)\otimes\cum(X,Y)\otimes\cum(Z).
\end{align*}
The last four terms will be cancelled by subsets belonging to the second to fifth terms while other subsets of these terms themselves will cancel several partitions of the first term. In particular, we are interested in decomposable partitions of the array
\[\begin{matrix}
\underbrace{\underbrace{D_{\omega_{j_1}}}_{1} \underbrace{D_{-\omega_{j_1+h}}}_2}_X \quad  & \underbrace{\underbrace{D_{-\omega_{j_1-k_1}}}_{3}\underbrace{D_{-\omega_{j_1+h-k_2}}}_4  \underbrace{D_{\omega_{j_1-k_1}}}_5  \underbrace{ D_{\omega_{j_1+h-k_2}}}_6}_W
\\[.2cm] \underbrace{\underbrace{D_{-\omega_{j_2}}}_{7} \underbrace{D_{\omega_{j_2+h}}}_8}_Y \quad & \underbrace{\underbrace{D_{\omega_{j_2-k_3}}}_{9} \underbrace{D_{\omega_{j_2+h-k_4}}}_{10}  \underbrace{D_{-\omega_{j_2-k_3}}}_{11}  \underbrace{ D_{-\omega_{j_2+h-k_4}}}_{12}}_Z,
\end{matrix}\]
but where we only have to consider the partitions that are not cancelling out, i.e., we can disregard those partitions where at least one of the sets $X, W, Y$ or $Z$ form a proper set within the partition. In other words, elements in the sets $X, W,  Y$ and $Z$ must hook with an element from one of the other sets. Taking into account the above constraints, we look for the structure with highest order, i.e., that allows as a partition of which as many sets form a proper submanifold. These partitions are of the form 
$S_{(13)(28)(46)(5,11)(7,9)(10,12)}$. Under $H_0$, Lemma \ref{FCboundssd} yields the constraints $k_1=0, j_2-j_1=0 \mod T, j_1-k_1-j_2+k_3=0 \mod T, k_3=0\mod T$ due to $\Delta_T^{(\omega_k)}$, which implies we are left with only $j_2, k_2, k_3$ as free variables. Via a similar derivation as for the expectation one obtains $\Var(\breve{J}_2) =  O(\frac{1}{T} T \frac{1}{(b T)^2}) = O(\frac{1}{ {(b T)}^2})$ under $H_0$. Similarly, using Lemma \ref{cumboundglsp} and Corollary \ref{cumbounds}(ii) we also obtain under $H_A$ that $\Var(\breve{J}_2) =  O(\frac{1}{ {(b T)}^2})$. All together we obtain that $E|J_2| =O(\frac{1}{b T^{3/2}})+O(\frac{1}{b T})$ under $H_0$ and $E|J_2| =O(\frac{1}{b \sqrt{T}})+O(\frac{1}{bT})$ under $H_A$.
\end{proof}

\begin{proof}[Proof of ~\eqref{J3}]
 Write $ \eqref{eq:J3}$ as
 \begin{align*}
J_3
&=\frac{1}{\sqrt{T}} \sum_{j=1}^{T}  \inprod{\E\big(D_{\omega_{j}} \otimes D_{\omega_{j+h}})}{\phi^{\omega_j}_{l} \otimes \phi^{\omega_{j+h}}_{\lpr}-\E(\hat{\phi}^{\omega_j}_{l} \otimes \hat{\phi}^{\omega_{j+h}}_{\lpr})}_{S} = O_p(\breve{J}_3),
\end{align*}
where
\begin{align*}
\breve{J}_3= \frac{1}{\sqrt{T}} \sum_{j=1}^{T} \Biginprod{ \E\big(D_{\omega_{j}} \otimes D_{\omega_{j+h}})}{\big(\F_{-\omega_{j}} \widetilde{\otimes}  \F_{-\omega_{j+h}} -\E (\hat{\F}_{-\omega_{j}} \widetilde{\otimes}  \hat{\F}_{-\omega_{j+h}}) \big)({\phi}^{-\omega_j}_l \otimes {\phi}^{-\omega_{j+h}}_{\lpr})}_{S}.
\end{align*}
By the Cauchy--Schwarz inequality and H\"older's inequality we obtain 
\begin{align*}
\E|\breve{J}_3| &\le \frac{1}{\sqrt{T}} \sum_{j=1}^{T} \snorm{\E\big(D_{\omega_{j}} \otimes D_{\omega_{j+h}})}_2 \sup_{\omega}\snorm{\big(\F_{\omega} \widetilde{\otimes}  \F_{\omega+\omega_{h}} -\E (\hat{\F}_{\omega} \widetilde{\otimes}  \hat{\F}_{\omega+\omega_{h}}) \big)({\phi}^{\omega}_l \otimes {\phi}^{\omega+\omega_{h}}_{\lpr})}_2 
\\& \le 
 \frac{1}{\sqrt{T}} \sum_{j=1}^{T} \snorm{\E\big(D_{\omega_{j}} \otimes D_{\omega_{j+h}})}_2 \sup_{\omega}\snorm{\big(\F_{\omega} \widetilde{\otimes}  \F_{\omega+\omega_{h}} -\E (\hat{\F}_{\omega} \widetilde{\otimes}  \hat{\F}_{\omega+\omega_{h}}) \big)}_{\infty} \|{\phi}^{\omega}_l \|_2  \|{\phi}^{\omega+\omega_{h}}_\lpr \|_2
 \\&\le   \frac{1}{\sqrt{T}} \sum_{j=1}^{T} \snorm{\E\big(D_{\omega_{j}} \otimes D_{\omega_{j+h}})}_2 \sup_{\omega}\snorm{\big(\F_{\omega} \widetilde{\otimes}  \F_{\omega+\omega_{h}} -\E (\hat{\F}_{\omega} \widetilde{\otimes}  \hat{\F}_{\omega+\omega_{h}}) \big)}_{2}.
\end{align*}
Recall that under $H_0$, we have
$\sup_j \snorm{\E (D_{j} \otimes D_{j+h})}_2 = O(\frac{1}{T})$
for $h\ne 0$ whilst under the alternative hypothesis
$\sup_j \snorm{\E (D_{j} \otimes D_{j+h})}_2 = O(\frac{1}{h^2}).$
The bounds in \eqref{J3} therefore immediately follow from Lemma \ref{lem:boundsFtensors}(i).
\end{proof}

\begin{proof}[Proof of ~\eqref{J4}]
 Write $ \eqref{eq:J4}$ as
 \begin{align*}
J_4
&=\frac{1}{\sqrt{T}} \sum_{j=1}^{T}  \inprod{\E\big(D_{\omega_{j}} \otimes D_{\omega_{j+h}})}{\E(\hat{\phi}^{\omega_j}_{l} \otimes \hat{\phi}^{\omega_{j+h}}_{\lpr})-\hat{\phi}^{\omega_j}_{l} \otimes \hat{\phi}^{\omega_{j+h}}_{\lpr})}_{S} = O_p(\breve{J}_4),
\end{align*}
where
\begin{align*}  
\breve{J}_4= \frac{1}{\sqrt{T}} \sum_{j=1}^{T} \Biginprod{ \E\big(D_{\omega_{j}} \otimes D_{\omega_{j+h}})}{\big(\E (\hat{\F}_{-\omega_{j}} \widetilde{\otimes}  \hat{\F}_{-\omega_{j+h}}) -\hat{\F}_{-\omega_{j}} \widetilde{\otimes}  \hat{\F}_{-\omega_{j+h}} \big)({\phi}^{-\omega_j}_l \otimes {\phi}^{-\omega_{j+h}}_{\lpr})}_{S}.  \end{align*}
Under $H_0$, \eqref{J4} the result follows now from an application of the Cauchy--Schwarz inequality and Lemma \ref{lem:boundsFtensors}(ii). Under the alternative, we consider a bound on $\E \breve{J}_4$ and $\var(\breve{J}_4)$. It is immediate from Lemma \ref{lem:aux}(i) that $\E \breve{J}_4=0$. Theorem \ref{thm:HAdist} with $n=0, r=2$ then implies that $\var(\breve{J}_4)=O(1)$.
\end{proof}
\noindent The four previous steps complete the proof.
\end{proof}

\subsection{Bound on $\sqrt{T}\E\|\hat{\bm{\beta}}_{h,s}^{(T)}-{\bm{\beta}}_{h,s}^{(T)}|$}\label{sec:hatbS}

\begin{proof}[\textcolor{black}{Proof of Theorem \ref{boundstatcase} and Theorem \ref{boundlocstatcase} for $\hat{\bm{\beta}}_{h,s}^{(T)}$}]
Note that in this case we are interested in the difference
\[
\sqrt{T}\E\|\hat{\bm{\beta}}_{h,s}^{(T)}-{\bm{\beta}}_{h,s}^{(T)}|=\frac{1}{\sqrt{T}}\E
\sum_{j=1}^{T} \inprod{D^{(T)}_{\omega_j}\otimes D^{(T)}_{\omega_{j+h}}}{(\hat{\lambda}_l^{\omega_j} \hat{\lambda}_\lpr^{\omega_{j+h}})^{-1/2} \hat{\phi}^{\omega_j}_l \otimes \hat{\phi}^{\omega_{j+h}}_{\lpr}-(\lambda_l^{\omega_j} \lambda_\lpr^{\omega_{j+h}})^{-1/2} \phi^{\omega_j}_l \otimes \phi^{\omega_{j+h}}_\lpr}_{S}.
\]
Observe that 
\begin{align*}
&(\hat{\lambda}_l^{\omega_j} \hat{\lambda}_\lpr^{\omega_{j+h}})^{-1/2} \hat{\phi}^{\omega_j}_l \otimes \hat{\phi}^{\omega_{j+h}}_{\lpr}-(\lambda_l^{\omega_j} \lambda_\lpr^{\omega_{j+h}})^{-1/2} \phi^{\omega_j}_l \otimes \phi^{\omega_{j+h}}_\lpr
\\& 
=(\hat{\lambda}_l^{\omega_j} \hat{\lambda}_\lpr^{\omega_{j+h}})^{-1/2}[\hat{\phi}^{\omega_j}_l \otimes \hat{\phi}^{\omega_{j+h}}_{\lpr} -\phi^{\omega_j}_l \otimes \phi^{\omega_{j+h}}_\lpr]+[(\hat{\lambda}_l^{\omega_j} \hat{\lambda}_\lpr^{\omega_{j+h}})^{-1/2}- ({\lambda}_l^{\omega_j} {\lambda}_\lpr^{\omega_{j+h}})^{-1/2}]\phi^{\omega_j}_l \otimes \phi^{\omega_{j+h}}_\lpr.
\end{align*}
Using a Taylor expansion of $(\hat{\lambda}_l^{\omega_j} \hat{\lambda}_\lpr^{\omega_{j+h}})^{-1/2}$ around $({\lambda}_l^{\omega_j} {\lambda}_\lpr^{\omega_{j+h}})^{-1/2}$, yields
\begin{align*}
(\hat{\lambda}_l^{\omega_j} \hat{\lambda}_\lpr^{\omega_{j+h}})^{-1/2}=&({\lambda}_l^{\omega_j} {\lambda}_\lpr^{\omega_{j+h}})^{-1/2} -\frac{1}{2}({\lambda}_l^{\omega_j} {\lambda}_\lpr^{\omega_{j+h}})^{-3/2} \big((\hat{\lambda}_l^{\omega_j} \hat{\lambda}_\lpr^{\omega_{j+h}})-({\lambda}_l^{\omega_j} {\lambda}_\lpr^{\omega_{j+h}})\big)  \\&+\frac{3}{4}({\lambda}_l^{\omega_j} {\lambda}_\lpr^{\omega_{j+h}})^{-5/2} \big((\hat{\lambda}_l^{\omega_j} \hat{\lambda}_\lpr^{\omega_{j+h}})-({\lambda}_l^{\omega_j} {\lambda}_\lpr^{\omega_{j+h}})\big)^2 .
\end{align*}
Additionally, it follows from solving the perturbed eigenelement approximation \citep[see e.g., ][]{Kato66},
 \begin{align*}
(\hat{\lambda}_l^{\omega_j} \hat{\lambda}_\lpr^{\omega_{j+h}})- ({\lambda}_l^{\omega_j} {\lambda}_\lpr^{\omega_{j+h}}) &= \inprod{(\hat{\F}_{\omega_{j}} \widetilde{\otimes}  \hat{\F}_{\omega_{j+h}} -\F_{\omega_{j}} \widetilde{\otimes}  \F_{\omega_{j+h}} )\phi^{\omega_j}_l \otimes \phi^{\omega_{j+h}}_\lpr}{\phi^{\omega_j}_l \otimes \phi^{\omega_{j+h}}_\lpr}_{S} + R_{\omega_j,h},
\end{align*}
where $R_{\omega_j,h}$ is a remainder term that satisfies $O_p(R_{\omega_j,h})=O_p(\snorm{\hat{\F}_{\omega_{j}} \widetilde{\otimes}  \hat{\F}_{\omega_{j+h}} -\F_{\omega_{j}} \widetilde{\otimes}  \F_{\omega_{j+h}} }^2_2$).
We therefore decompose 
\begin{align*}
\frac{1}{\sqrt{T}}&\bigg|\sum_{j=1}^{T} \inprod{D^{(T)}_{\omega_j}\otimes D^{(T)}_{\omega_{j+h}}}{(\hat{\lambda}_l^{\omega_j} \hat{\lambda}_\lpr^{\omega_{j+h}})^{-1/2} \hat{\phi}^{\omega_j}_l \otimes \hat{\phi}^{\omega_{j+h}}_{\lpr}-(\lambda_l^{\omega_j} \lambda_\lpr^{\omega_{j+h}})^{-1/2} \phi^{\omega_j}_l \otimes \phi^{\omega_{j+h}}_\lpr}_{S}\bigg| \\
&\le |J_{s,1}| +|J_{s,2}| +|J_{s,3}| +|J_R|,
\end{align*}
where 
\begin{align*} J_{s,1}&=\frac{1}{\sqrt{T}}\sum_{j=1}^{T}\inprod{D^{(T)}_{\omega_j}\otimes D^{(T)}_{\omega_{j+h}}}{({\lambda}_l^{\omega_j} {\lambda}_\lpr^{\omega_{j+h}})^{-1/2} \big( \hat{\phi}^{\omega_j}_l \otimes \hat{\phi}^{\omega_{j+h}}_{\lpr}- \phi^{\omega_j}_l \otimes \phi^{\omega_{j+h}}_\lpr \big)}_{S},
\\
J_{s,2}&=\frac{1}{2\sqrt{T}}\sum_{j=1}^{T}  \biginprod{\frac{D^{(T)}_{\omega_j}\otimes D^{(T)}_{\omega_{j+h}}}{({\lambda}_l^{\omega_j} {\lambda}_\lpr^{\omega_{j+h}})^{3/2}}}{\big( \hat{\phi}^{\omega_j}_l \otimes \hat{\phi}^{\omega_{j+h}}_{\lpr}- \phi^{\omega_j}_l \otimes \phi^{\omega_{j+h}}_\lpr \big)}_{S}
\\&\phantom{\frac{1}{\sqrt{T}}\sum_{j=1}^{T}}\quad \quad \quad\times \overline{\biginprod{(\hat{\F}_{\omega_{j}} \widetilde{\otimes}  \hat{\F}_{\omega_{j+h}} -\F_{\omega_{j}} \widetilde{\otimes}  \F_{\omega_{j+h}} )\phi^{\omega_j}_l \otimes \phi^{\omega_{j+h}}_\lpr}{\phi^{\omega_j}_l \otimes \phi^{\omega_{j+h}}_\lpr}}_{S},
\\J_{s,3}&=\frac{1}{2 \sqrt{T}}\sum_{j=1}^{T}  \inprod{\frac{D^{(T)}_{\omega_j}\otimes D^{(T)}_{\omega_{j+h}}}{({\lambda}_l^{\omega_j} {\lambda}_\lpr^{\omega_{j+h}})^{3/2}}}{\big(  \phi^{\omega_j}_l \otimes \phi^{\omega_{j+h}}_\lpr \big)}_{S}
\\&\phantom{\frac{1}{\sqrt{T}}\sum_{j=1}^{T}}\quad \quad \quad\times \overline{\biginprod{(\hat{\F}_{\omega_{j}} \widetilde{\otimes}  \hat{\F}_{\omega_{j+h}} -\F_{\omega_{j}} \widetilde{\otimes}  \F_{\omega_{j+h}} )\phi^{\omega_j}_l \otimes \phi^{\omega_{j+h}}_\lpr}{\phi^{\omega_j}_l \otimes \phi^{\omega_{j+h}}_\lpr}}_{S},
\\ J_{R}&=\frac{1}{\sqrt{T}}C \sum_{j=1}^{T}({\lambda}_l^{\omega_j} {\lambda}_\lpr^{\omega_{j+h}})^{-3/2}R_{\omega_j,h,l,\lpr}  \inprod{D^{(T)}_{\omega_j}\otimes D^{(T)}_{\omega_{j+h}}}{\big( \hat{\phi}^{\omega_j}_l \otimes \hat{\phi}^{\omega_{j+h}}_{\lpr}- \phi^{\omega_j}_l \otimes \phi^{\omega_{j+h}}_\lpr \big)}_{S},
\end{align*}
for some constant $C>0$. Note that, using \eqref{eq:approxdif}, we have $|J_{s,1}| = O_p(J_1+J_2+J_3+J_4) $, from which the respective order follows from Lemma \ref{lemJ1}. The same holds for $J_{s,3}$. A similar decomposition as in Lemma \ref{lemJ1} will show that $J_{s,2}$ and $J_{s,R}$ are of strictly lower order. The proof follows along the lines of the proof of Lemma \ref{lemJ1} and is therefore omitted. 
\end{proof}

\section{Covariance structure under alternative hypothesis of local stationarity}\label{sub:covHA}

\subsection{Completion of covariance structure of Theorem \ref{th:test_alt}}
\begin{proof}[Completion of covariance structure of Theorem \ref{th:test_alt}]
We now derive the covariance structure of $\tprojes_{h,x}$ under $H_A$ and focus on $\tprojes_{h,u}$. We have
\begin{align*}
&\Cov(\sqrt{T}{\bm{\beta}}_{h_1,u}^{(T)},\sqrt{T}{\bm{\beta}}_{h_2,u}^{(T)}) =\\&  
\frac{1}{T}\sum_{j_1,j_2}^{T}  \sum_{\small{\substack{l_1 \in [L(\omega_{j_1})], l_2 \in [L(\omega_{j_1+h_1})], \\ l_3 \in [L(\omega_{j_2})] , l_4 \in [L(\omega_{j_2+h_2})]}}}  \Big\langle \Cov\Big(D^{(T)}_{\omega_{j_1}}\otimes D^{(T)}_{\omega_{j_1+h}}, D^{(T)}_{\omega_{j_2}}\otimes D^{(T)}_{\omega_{j_2+h}}\Big) \Big(\phi^{\omega_{j_2}}_{l_{3}} \otimes \phi^{\omega_{j_2+h_2}}_{l_4}\Big), \phi^{\omega_{j_1}}_{l_{1}}\otimes \phi^{\omega_{j_1+h_1}}_{l_2} \Big\rangle,
\end{align*} 
where now $({\lambda}^{\omega}_l, {\phi}^{\omega}_l\colon l  \ge 1)$ are the eigenelements of the time-integrated spectral density operator $G_{\omega}$. Lemma \ref{cumboundglsp} implies the $h$-lag covariance operator of the fDFT's has covariance structure as in \eqref{eq:covHa}. Hence, we find that
\begin{align*}\tageq \label{eq:covBetaHa}
&\Cov(\sqrt{T}{\bm{\beta}}_{h_1,u}^{(T)},\sqrt{T}{\bm{\beta}}_{h_2,u}^{(T)}) =\\&  
\frac{1}{T}\sum_{j_1,j_2}^{T} \sum_{\small{\substack{l_1 \in [L(\omega_{j_1})], l_2 \in [L(\omega_{j_1+h_1})], \\ l_3 \in [L(\omega_{j_2})] , l_4 \in [L(\omega_{j_2+h_2})]}}}
 \Bigg\{  \frac{2\pi}{T}\langle \tilde{\F}_{h_2-h_1:\omega_{j_1},-\omega_{j_1+h_1},-\omega_{j_2}} (\phi^{\omega_{j_2}}_{l_{3}} \otimes \phi^{\omega_{j_2+h_2}}_{l_4}), \phi^{\omega_{j_1}}_{l_{1}} \otimes \phi^{\omega_{j_1+h_1}}_{l_2}\rangle +O\big(\frac{1}{T^2}\big)\bigg)
\\& \phantom{\sum_{j_1,j_2}^{T}}+
\bigg(\big\langle\tilde{\F}_{j_1-j_2:\omega_{j_1}}(\phi^{\omega_{j_2}}_{l_3}), \phi^{\omega_{j_1}}_{l_1}  \big\rangle+O\big(\frac{1}{T}\big)\bigg)   \bigg(\big\langle \tilde{\F}_{-j_1-h_1+j_2+h_2:-\omega_{j_1+h_1}}(\phi^{-\omega_{j_2+h_2}}_{l_4} ), \phi^{-\omega_{j_1+h_1}}_{l_2} \big\rangle  + O\big(\frac{1}{T}\big)\bigg) 
\\& \phantom{\sum_{j_1,j_2}^{T}}+\big(\big \langle\tilde{\F}_{j_1+j_2+h_2:\omega_{j_1}}(\phi^{-\omega_{j_2+h_2}}_{l_4}, \phi^{\omega_{j_1}}_{l_1}\big \rangle+O\big(\frac{1}{T}\big)\bigg)  \bigg(\big \langle \tilde{\F}_{-j_1-h_1-j_2:-\omega_{j_1+h_1}}(\phi^{\omega_{j_2}}_{l_3}), \phi^{-\omega_{j_1+h_1}}_{l_2}\big \rangle +O\big(\frac{1}{T}\big)\Bigg\}. 
\end{align*} 
and
\begin{align*}\tageq \label{eq:covBetaHaS}
&\Cov(\sqrt{T}{\bm{\beta}}_{h_1,s}^{(T)},\sqrt{T}{\bm{\beta}}_{h_2,s}^{(T)}) =\frac{1}{T}\sum_{j_1,j_2}^{T} \sum_{\small{\substack{l_1 \in [L(\omega_{j_1})], l_2 \in [L(\omega_{j_1+h_1})], \\ l_3 \in [L(\omega_{j_2})] , l_4 \in [L(\omega_{j_2+h_2})]}}} ({\lambda_{l_1}^{\omega_{j_1}} \lambda_{l_2}^{\omega_{j_1+h}}\lambda_{l_3}^{\omega_{j_2}} \lambda_{l_4}^{\omega_{j_2+h}}})^{-1/2}\\& 
 \phantom{\frac{1}{T}\sum_{j_1,j_2}^{T}}
 \times\Bigg\{  \frac{2\pi}{T}\langle \tilde{\F}_{h_2-h_1:\omega_{j_1},-\omega_{j_1+h_1},-\omega_{j_2}} (\phi^{\omega_{j_2}}_{l_{3}} \otimes \phi^{\omega_{j_2+h_2}}_{l_4}), \phi^{\omega_{j_1}}_{l_{1}} \otimes \phi^{\omega_{j_1+h_1}}_{l_2}\rangle +O\big(\frac{1}{T^2}\big)\bigg)
\\& \phantom{\sum_{j_1,j_2}^{T}}+
\bigg(\big\langle\tilde{\F}_{j_1-j_2:\omega_{j_1}}(\phi^{\omega_{j_2}}_{l_3}), \phi^{\omega_{j_1}}_{l_1}  \big\rangle+O\big(\frac{1}{T}\big)\bigg)   \bigg(\big\langle \tilde{\F}_{-j_1-h_1+j_2+h_2:-\omega_{j_1+h_1}}(\phi^{-\omega_{j_2+h_2}}_{l_4} ), \phi^{-\omega_{j_1+h_1}}_{l_2} \big\rangle  + O\big(\frac{1}{T}\big)\bigg) 
\\& \phantom{\sum_{j_1,j_2}^{T}}+\big(\big \langle\tilde{\F}_{j_1+j_2+h_2:\omega_{j_1}}(\phi^{-\omega_{j_2+h_2}}_{l_4}, \phi^{\omega_{j_1}}_{l_1}\big \rangle+O\big(\frac{1}{T}\big)\bigg)  \bigg(\big \langle \tilde{\F}_{-j_1-h_1-j_2:-\omega_{j_1+h_1}}(\phi^{\omega_{j_2}}_{l_3}), \phi^{-\omega_{j_1+h_1}}_{l_2}\big \rangle +O\big(\frac{1}{T}\big)\Bigg\}. 
\end{align*} 
Note then once more that 
\[
\Re \tproj_{h,x} = \frac 12\Big({\tproj_{h,x}+\overline{\tproj_{h,x}}}\Big)
\quad\mbox{and}\quad 
\Im  \tproj_{h,x} = \frac 1{2\im}\Big({\tproj_{h,x} -\overline{\tproj_{h,x}}}\Big).
\]
Under the alternative, these are in fact correlated and four separate cases will have to be considered: 
\begin{enumerate}[label={(S6.1.\arabic*)},ref=S6.1.\arabic*]
\item  \label{eq:covRR}
$\begin{aligned}[t] 
\cov\big( \Re \tproj_{h_1,x} , \Re  \tproj_{h_2,x}\big) 
& =\frac{1}{4}\Big[ \cov\big( \tproj_{h_1,x}, \tproj_{h_2,x}\big)
+\cov\big( \tproj_{h_1,x},\overline{\tproj_{h_2,x}}\big) \\
&\qquad 
+\cov\big( \overline{\tproj_{h_1,x}},{\tproj_{h_2,x}}\big)
+\cov\big(\overline{\tproj_{h_1,x}},\overline{\tproj_{h_2,x}}\big)\Big], \label{eq:covRR}
\end{aligned}$
\item \label{eq:covRI} 
$\begin{aligned}[t]  
\cov\big(\Re \tproj_{h_1,x} , \Im \tproj_{h_2,x}\big) 
& =\frac{1}{4 \overline{\im}}\Big[ \cov\big(\tproj_{h_1,x}, \tproj_{h_2,x})
-\cov\big( \tproj_{h_1,x},\overline{\tproj_{h_2,x}}) \\
&\qquad
+\cov\big( \overline{\tproj_{h_1,x}},{\tproj_{h_2,x}}\big) 
-\cov\big(\overline{\tproj_{h_1,x}},\overline{\tproj_{h_2,x}}\big) \Big],
\end{aligned}$
\item \label{eq:covIR}
$\begin{aligned}[t]  
\cov\big( \Im \tproj_{h_1,x}, \Re \tproj_{h_2,x}\big) 
& =\frac{1}{4 \overline{\im}}\Big[ \cov\big(\tproj_{h_1,x}, \tproj_{h_2,x}\big)
+\cov\big( \tproj_{h_1,x},\overline{\tproj_{h_2,x}}\big) \\
&\qquad
-\cov\big( \overline{\tproj_{h_1,x}},{\tproj_{h_2,x}}\big) 
-\cov\big(\overline{\tproj_{h_1,x}},\overline{\tproj_{h_2,x}}\big)\Big],
\end{aligned}$
\item \label{eq:covII}
$\begin{aligned}[t]   
\cov\big( \Im\tproj_{h_1,x} , \Im \tproj_{h_2,x}\big) 
& =\frac{1}{4 }\Big[ \cov\big(\tproj_{h_1}, \tproj_{h_2,x}\big)
-\cov\big( \tproj_{h_1,x},\overline{\tproj_{h_2}}\big) \\
&\qquad 
-\cov\big( \overline{\tproj_{h_1,x}},{\tproj_{h_2,x}}\big) 
+\cov\big(\overline{\tproj_{h_1,x}},\overline{\tproj_{h_2,x}}\big)\Big]. 
\end{aligned}$
\end{enumerate}
These expressions can be easily obtained from \eqref{eq:covBetaHa} and \eqref{eq:covBetaHaS} by taking the appropriate conjugates. The four terms on the right-hand sides of the above four equations are derived for arbitrary basis functions in Section \ref{sec:covCLT}. It then remains to replace the basis functions with the (standardized) eigenfunctions of the integrated spectral density operators and the sum over the dimensions.
\end{proof}
\subsection{ Covariance structure of Theorem \ref{CLT_alt}} \label{sec:covCLT}
 Lemma \ref{cumboundglsp} implies the $h$-lag covariance operator of the fDFT's has covariance structure as in \eqref{eq:covHa}.
Observe that the covariance structure of the real and imaginary parts are linear combinations of the four different combinations of the covariances with their conjugates similar to \eqref{eq:covRR}-\eqref{eq:covII}. Therefore, a tedious derivation shows that we obtain for the covariance structure of the projections in Theorem \ref{CLT_alt}:
\begin{align} 
\label{eq:nonnon_lim}
\Upsilon_{h_1,h_2}(\psi_{l_1l_1^\prime\,l_2l_2^\prime})
&= \lim_{T \to \infty} \frac{1}{T}\sum_{j_1,j_2=1}^{T} 
\Big(\big\langle\tilde{\F}_{j_1-j_2;\omega_{j_1}}(\psi_{l_2}),\psi_{l_1}\big\rangle 
\big\langle\tilde{\F}_{-j_1-h_1+j_2+h_2;-\omega_{j_1+h_1}}(\psi_{l_2^\prime}),\psi_{l_1^\prime}\big\rangle \nonumber  \\
& \phantom{ \lim_{T \to \infty} \frac{1}{T}\sum_{j_1,j_2=1}^{T}}
+\big\langle\tilde{\F}_{j_1+j_2+h_2;\omega_{j_1}}(\psi_{l_2^\prime}),\psi_{l_1}\big\rangle
\big\langle\tilde{\F}_{-j_1-h_1-j_2,-\omega_{j_1+h_1}}(\psi_{l_2}),\psi_{l_1^\prime}\big\rangle \nonumber
\\ 
& \phantom{ \lim_{T \to \infty} \frac{1}{T}\sum_{j_1,j_2=1}^{T}}+ \frac{2\pi}{T} 
\big\langle\tilde{\F}_{(-h_1+h_2 ; \omega_{j_1},-\omega_{j_1+h_1},-\omega_{j_2})}(\psi_{l_2\, l_2^\prime}),\psi_{l_1,l_1^\prime}\big\rangle\Big),
\end{align}
\begin{align} 
\label{eq:noncon_lim}
\acute{\Upsilon}_{h_1,h_2}(\psi_{l_1l_1^\prime\,l_2l_2^\prime})
&= \lim_{T \to \infty} \frac{1}{T}\sum_{j_1,j_2=1}^{T}\Big(
\big\langle\tilde{\F}_{j_1+j_2;\omega_{j_1}}(\psi_{l_2}),\psi_{l_1}\big\rangle 
\big\langle\tilde{\F}_{-j_1-h_1-j_2-h_2;-\omega_{j_1+h_1}}(\psi_{l_2^\prime}),\psi_{l_1^\prime}\big\rangle \nonumber \\
& \phantom{ \lim_{T \to \infty} \frac{1}{T}\sum_{j_1,j_2=1}^{T}}
+\big\langle\tilde{\F}_{j_1-j_2-h_2;\omega_{j_1}}(\psi_{l_2^\prime}),\psi_{l_1}\big\rangle
\big\langle\tilde{\F}_{-j_1-h_1+j_2;-\omega_{j_1+h_1}}(\psi_{l_2}),\psi_{l_1^\prime}\big\rangle \nonumber
\\& \phantom{ \lim_{T \to \infty} \frac{1}{T}\sum_{j_1,j_2=1}^{T}}+ \frac{2\pi}{T} 
\big\langle\tilde{\F}_{(-h_1-h_2 ; \omega_{j_1},-\omega_{j_1+h_1},\omega_{j_2})} (\psi_{l_2\, l_2^\prime}),\psi_{l_1\,l_1^\prime}\big\rangle\Big),
\end{align}
\begin{align}
\label{eq:concon_lim}
\bar{\Upsilon}_{h_1,h_2}(\psi_{l_1l_1^\prime\,l_2l_2^\prime})
&=\lim_{T \to \infty}  \frac{1}{T}\sum_{j_1,j_2=1}^{T}\Big(
\big\langle\tilde{\F}_{-j_1+j_2;-\omega_{j_1}}(\psi_{l_2}),\psi_{l_1}\big\rangle
\big\langle\tilde{\F}_{j_1+h_1-j_2-h_2;\omega_{j_1+h_1}}(\psi_{l_2^\prime}),\psi_{l_1^\prime}\big\rangle  \nonumber \\
& \phantom{ \lim_{T \to \infty} \frac{1}{T}\sum_{j_1,j_2=1}^{T}}
+\big\langle\tilde{\F}_{-j_1-j_2-h_2;-\omega_{j_1}}(\psi_{l_2^\prime}),\psi_{l_1}\big\rangle
\big\langle\tilde{\F}_{j_1+h_1+j_2;\omega_{j_1+h_1}}(\psi_{l_2}),\psi_{l_1^\prime}\big\rangle \nonumber
\\&  \phantom{ \lim_{T \to \infty} \frac{1}{T}\sum_{j_1,j_2=1}^{T}}+\frac{2\pi}{T} 
\big\langle\tilde{\F}_{(h_1-h_2 ; -\omega_{j_1},\omega_{j_1+h_1},\omega_{j_2})} (\psi_{l_2\, l_2^\prime}),\psi_{l_1\,l_1^\prime}\big\rangle\Big)
\end{align}
and
\begin{align} 
\label{eq:connon_lim}
\grave{\Upsilon}_{h_1,h_2}(\psi_{l_1l_1^\prime\,l_2l_2^\prime}) 
&=\lim_{T \to \infty}  \frac{1}{T}\sum_{j_1,j_2=1}^{T} \Big(
\big\langle\tilde{\F}_{-j_1-j_2;-\omega_{j_1}}(\psi_{l_2}),\psi_{l_1}\big\rangle 
\big\langle\tilde{\F}_{j_1+h_1+j_2+h_2;\omega_{j_1+h_1}}(\psi_{l_2^\prime}),\psi_{l_1^\prime}\big\rangle \nonumber \\
&  \phantom{ \lim_{T \to \infty} \frac{1}{T}\sum_{j_1,j_2=1}^{T}}
+\big\langle\tilde{\F}_{-j_1+j_2+h_2;-\omega_{j_1}}(\psi_{l_2^\prime}),\psi_{l_1}\big\rangle
\big\langle\tilde{\F}_{j_1+h_1-j_2;\omega_{j_1+h_1}}(\psi_{l_2}),\psi_{l_1^\prime}\big\rangle \nonumber
\\& 
\phantom{ \lim_{T \to \infty} \frac{1}{T}\sum_{j_1,j_2=1}^{T}}+\frac{2\pi}{T} 
\big\langle\tilde{\F}_{(h_1+h_2 ; -\omega_{j_1},\omega_{j_1+h_1},-\omega_{j_2})}(\psi_{l_2\, l_2^\prime}),\psi_{l_1\,l_1^\prime}\big\rangle\Big). 
\end{align}

\section{Estimation of the integrated tri-spectral density operator }\label{sec:F4proof}

\subsection{Consistency under $H_0$}\label{sec:F4H0}

\begin{theorem}\label{thm:intF4stat}
Suppose \ref{Statcase}(4,2) and \ref{cumglsp}(8,2) hold. Then the estimator in \eqref{eq:fourthorderest} of the tri-spectral density operator satisfies
\begin{align*}
& \snorm{\E \hat{\F}_{\omega_{j_1},\omega_{j_2},\omega_{j_3}} -{\F}_{\omega_{j_1},\omega_{j_2},\omega_{j_3}} }_2 
=O\bigg(\frac{1}{b_4 T}+b_4^2\bigg), 
 \\& 
 \snorm{\Cov( \hat{\F}_{\omega_{j_1},\omega_{j_2},\omega_{j_3}},\hat{\F}_{\omega_{j_1},\omega_{j_2},\omega_{j_3}}) }^2_2
 =O\bigg(\frac{1}{b_4^3 T}\bigg).
\end{align*}
Consequently, 
\begin{align}
\E \Bigsnorm{ \frac{(2\pi)^2}{T^2}\sum_{j_1,j_2=1}^{T}\hat{\F}_{\omega_{j_1},-\omega_{j_1+h},-\omega_{j_2}} - \int \int  {\F}_{\omega,-\omega +\omega_{h},-\omega^{\prime}}d\omega d\omega^{\prime}}^2_2 
=O\bigg(\frac{1}{b_4^3T} +b_4^4\bigg)  \tageq \label{eq:intF4}
\end{align}
which is therefore mean square consistent for bandwidths satisfying $b_4 \to 0$ such that $b_4^3 T\to \infty$ as $T\to \infty$. 
\end{theorem}

\begin{proof}[Proof of Theorem \ref{thm:intF4stat}]
Consider first the expectation of $\hat{\F}_{\omega_{j_1},\omega_{j_2},\omega_{j_3}}$ which is given by
\begin{align*}
& \E\hat{\F}_{\omega_{j_1},\omega_{j_2},\omega_{j_3}}=
\frac{(2\pi)^{3}}{{(b_4 T)}^{3}}\sum_{k_1,k_2, k_3}  K_4\Big(\frac{\omega_{j_1}-{\omega_{k_1}}}{b_4},\ldots,\frac{\omega_{j_4}+\tiny{\sum_{i=1}^{3}\omega_{k_i}}}{b_4}\Big) \E  \Phi(\boldsymbol{\omega_k})  I^{(T)}_{\omega_{k_1},\omega_{k_2},\omega_{k_3},-\tiny{\sum_{i=1}^{3}\omega_{k_i}}}
\\& =
\frac{(2\pi)^{3}}{{(b_4 T)}^{3}} \frac{T}{2\pi}\sum_{k_1,k_2, k_3}  K_4\Big(\frac{\omega_{j_1}-{\omega_{k_1}}}{b_4},\ldots,\frac{\omega_{j_4}+\tiny{\sum_{i=1}^{3}\omega_{k_i}}}{b_4}\Big)\E \big(\Phi(\boldsymbol{\omega_k}) D_{\omega_{k_1}} \otimes D_{\omega_{k_2}}\otimes D_{\omega_{k_3}}\otimes D_{-\tiny{\sum_{i=1}^{3}\omega_{k_i}}} \big),
\end{align*}
where we used in the second equality that the tri-periodogram tensor can be expressed in terms the cumulant tensors of the upscaled fDFTs. Using then Theorem \ref{prodcumthm}, we have
\begin{align*}
\E\big[\Phi(\boldsymbol{\omega_k}) I^{(T)}_{\omega_{j_1},\omega_{j_2},\omega_{j_3},\omega_{j_4}}\big]
 &=\frac{T}{{(2 \pi)}} \Big(\Phi(\boldsymbol{\omega_k})\cum\big(D^{(T)}_{\omega_{j_1}},\ldots, D^{(T)}_{\omega_{j_4}}\big)  \\& 
 \qquad+\Phi(\boldsymbol{\omega_k}) \cum\big(D^{(T)}_{\omega_{j_1}}, D^{(T)}_{\omega_{j_2}}\big) \otimes \cum\big(D^{(T)}_{\omega_{j_3}}, D^{(T)}_{\omega_{j_4}}\big)\\
 &\qquad  + \Phi(\boldsymbol{\omega_k})S_{1324}(\cum\big(D^{(T)}_{\omega_{j_1}}, D^{(T)}_{\omega_{j_3}}\big) \otimes \cum\big(D^{(T)}_{\omega_{j_2}}, D^{(T)}_{\omega_{j_4}}\big)) \\
 &\qquad +\Phi(\boldsymbol{\omega_k})S_{1423}(\cum\big(D^{(T)}_{\omega_{j_1}}, D^{(T)}_{\omega_{j_4}}\big)\big))
\otimes \cum\big(D^{(T)}_{\omega_{j_2}}, D^{(T)}_{\omega_{j_3}}\big)  \Big). \tageq
\label{eq:fourth-decomp}
\end{align*}
Note that, due to the inclusion of the function $\Phi$, only those terms are to be considered for which the frequencies satisfy $j_4= -j_1-j_2-j_3$ in such a way that $j_1 \ne j_2$ $j_1 \ne j_3$ and $j_2 \ne j_3$ and $j_1 \ne j_4$ $T$-periodically. For those values not contained on a proper submanifold, the products of second-order cumulant tensors are at most of order $O(T^{-2})$ in an $L^2$ sense under the null hypothesis. Using Lemma \ref{FCboundssd}, 
it follows therefore that 
\begin{align*}
\biggsnorm{\E\big[\Phi(\boldsymbol{\omega_k}) I^{(T)}_{\omega_{j_1},\omega_{j_2},\omega_{j_3},\omega_{j_4}}\big] - \frac{T}{{(2 \pi)}} \frac{2\pi}{T}  \F_{\omega_{j_1}, \omega_{j_2}, \omega_{j_3}}}_2 
 =  O\bigg(\frac{T}{T^2}\bigg)  
 =O\bigg(\frac{1}{T}\bigg)  
\end{align*}
and hence is asymptotically unbiased. Additionally, the smoothing kernel is defined as a product of one-dimensional smoothing kernels with compact support. 
Denote by $V_a^b(K)$ the total variation on $[a,b]$ of the function $K$, then a standard argument gives these kernels satisfy
\begin{align*}
\bigg|\int^b_a K(x) dx- \frac{1}{b_4 T}\sum_{j=1}K\bigg(\frac{x_j}{b_4} \bigg)\bigg| \le \frac{1}{T}V_{a}^{b}\bigg(\frac{1}{b_4}K\bigg(\frac{\cdot}{b_4}\bigg)\bigg) = O\bigg(\frac{1}{b_4 T}\bigg),
\end{align*}
where we used that $V_{a}^{b}\big(\frac{1}{b_4}K(\frac{\cdot}{b_4})\big)=O(\frac{1}{b_4})$. This together with a change of variables yields
\begin{align*}
& \biggsnorm{\E\big[ \hat{\F}^{(T)}_{\omega_{1},\omega_{2},\omega_{3}}\big] -\int \int \int  \frac{1}{b_4^3}K_4\Big(\frac{\omega_1-{\alpha}_1}{b_4},\ldots,\frac{\tiny{\sum_{i}^3 \alpha_i-\omega_i}}{b_4}\Big) \F_{\alpha_{1}, \alpha_{2}, \alpha_{3}} d\alpha_1 d\alpha_2 d\alpha_3}_2 
\\& =
\biggsnorm{\E\big[ \hat{\F}^{(T)}_{\omega_{1},\omega_{2},\omega_{3}}\big] -\int \int \int  K_4\Big(x_1,\ldots,-b_4\tiny{\sum_{i}^3 x_i }\Big) \F_{\boldsymbol{\omega}-\boldsymbol{x} b_4} d x_1 d x_2 d x_3}_2 
 =O\bigg(\frac{1}{b_4T}\bigg),
\end{align*}
using the more compact notation $\boldsymbol{\omega}-\boldsymbol{x} b_4:=(\omega_1-x_1 b_4, \omega_2 -x_2 b_4, \omega_3-x_3 b4)  \in \mathbb{R}^3$. For $\boldsymbol{\alpha} \in \mathbb{R}^3$, note that \ref{Statcase} with $\ell=2$ implies that the operator-valued derivative mappings $\boldsymbol{\alpha} \mapsto D^i \F_{\boldsymbol{\alpha}}$ are well-defined elements of $S(H \otimes H)$. Hence, a Taylor expansion of the operator-valued function $\F_{\boldsymbol{\alpha}}$ at the point $\boldsymbol{\omega_0}:= (\omega_1, \omega_2,\omega_3)$ yields 
\begin{align*}
\biggsnorm{\F_{\boldsymbol{\omega-x b_4}} -\Big( \F_{\boldsymbol{\omega_0}} +[b_4 \boldsymbol{ x}]^\top D_{\boldsymbol{\omega}}\F_{\boldsymbol{\omega}}\Big \vert_{\boldsymbol{\omega}=\boldsymbol{\omega_0}}  +[b_4 \boldsymbol{ x}]^\top D^2_{\boldsymbol{\omega}}\F_{\boldsymbol{\omega}} \Big\vert_{\boldsymbol{\omega}=\boldsymbol{\omega_0}} [b_4 \boldsymbol{x}]\Big)}_2 = o(b_4^2).
\end{align*}
Utilizing that the smoothing kernel is symmetric in each argument, we obtain for \eqref{eq:fourthorderest} 
\begin{align*}
\bigsnorm{\E\big[ \hat{\F}_{\omega_{j_1},\omega_{j_2},\omega_{j_3},\omega_{j_4}}\big] - \F_{\omega_{j_1}, \omega_{j_2}, \omega_{j_3}}}_2 
 =O\bigg(\frac{1}{b_4T} +b_4^2\bigg).  
\end{align*}
Consequently,
\begin{align*}
\biggsnorm{\E \int \int \hat{\F}_{\omega,-\omega +\omega_{h},-\omega^{\prime},\omega^{\prime}+\omega^{\prime}_h} d\omega d\omega^{\prime}- \int \int  {\F}_{\omega,-\omega +\omega_{h},-\omega^{\prime},\omega^{\prime}+\omega^{\prime}_h}d\omega d\omega^{\prime}}_2 
=O\bigg(\frac{1}{b_4T} +b_4^2\bigg).  
\end{align*}
Consider next the covariance structure of $\F_{\omega_{j_1},\omega_{j_2},\omega_{j_3}}$. By definition,
\small{
\begin{align}\label{eq:covF4hat}
&\text{Cov}\Big(\hat{\F}_{\omega_{j_1},\omega_{j_2},\omega_{j_3}}, \hat{\F}_{\omega_{l_1},\omega_{l_2},\omega_{l_3}}\Big)\\&
=\text{Cov}\Big(\frac{(2\pi)^{3}}{{(b_4 T)}^{3}}\sum_{k_1,k_2, k_3}  K_4\Big(\frac{\omega_{j_1}-{\omega_{k_1}}}{b_4},\ldots,\frac{\omega_{j_4}+\tiny{\sum_{i=1}^{3}\omega_{k_i}}}{b_4}\Big)  \Phi(\boldsymbol{\omega_k})  I^{(T)}_{\omega_{k_1},\omega_{k_2},\omega_{k_3},-\tiny{\sum_{i=1}^{3}\omega_{k_i}}} ,\nonumber\\&\phantom{\frac{(2\pi)^{6}}{{(b_4 T)}^{6}}\frac{T^2}{(2\pi)^2}\text{Cov}\Big(}
\frac{(2\pi)^{3}}{{(b_4 T)}^{3}}\sum_{s_1,s_2, s_3}  K_4\Big(\frac{\omega_{l_1}-{\omega_{s_1}}}{b_4},\ldots,\frac{\omega_{l_4}+\tiny{\sum_{i=1}^{3}\omega_{s_i}}}{b_4}\Big)  \Phi(\boldsymbol{\omega_s})  I^{(T)}_{\omega_{s_1},\omega_{s_2},\omega_{s_3},-\tiny{\sum_{i=1}^{3}\omega_{s_i}}} \Big)
\nonumber \\& =
\frac{(2\pi)^{6}}{{(b_4 T)}^{6}}\frac{T^2}{(2\pi)^2}\text{Cov}\Big(\sum_{k_1,k_2, k_3}  K_4\Big(\frac{\omega_{j_1}-{\omega_{k_1}}}{b_4},\ldots,\frac{\omega_{j_4}+\tiny{\sum_{i=1}^{3}\omega_{k_i}}}{b_4}\Big)  \Phi(\boldsymbol{\omega_k})  D_{\omega_{k_1}} \otimes D_{\omega_{k_2}}\otimes D_{\omega_{k_3}}\otimes D_{-\tiny{\sum_{i=1}^{3}\omega_{k_i}}}  ,\nonumber\\&
\phantom{\frac{(2\pi)^{6}}{{(b_4 T)}^{6}}\frac{T^2}{(2\pi)^2}\text{Cov}\Big(} 
\sum_{s_1,s_2, s_3}  K_4\Big(\frac{\omega_{l_1}-{\omega_{s_1}}}{b_4},\ldots,\frac{\omega_{l_4}+\tiny{\sum_{i=1}^{3}\omega_{s_i}}}{b_4}\Big)  \Phi(\boldsymbol{\omega_s})  D_{\omega_{s_1}} \otimes D_{\omega_{s_2}}\otimes D_{\omega_{s_3}}\otimes D_{-\tiny{\sum_{i=1}^{3}\omega_{s_i}}} \Big) \nonumber \\&
=\frac{(2\pi)^{4}}{ b_4^6 T^{4}} \sum_{k_1,k_2, k_3}  K_4\Big(\frac{\omega_{j_1}-{\omega_{k_1}}}{b_4},\ldots,\frac{\omega_{j_4}+\tiny{\sum_{i=1}^{3}\omega_{k_i}}}{b_4}\Big)  \sum_{s_1,s_2, s_3}  K_4\Big(\frac{\omega_{l_1}-{\omega_{s_1}}}{b_4},\ldots,\frac{\omega_{l_4}+\tiny{\sum_{i=1}^{3}\omega_{s_i}}}{b_4}\Big)  \nonumber  \\& 
\phantom{\frac{(2\pi)^{6}}{{(b_4 T)}}\frac{T^2}{(2\pi)}\Big(} 
\times \text{Cum}\Big( \Phi(\boldsymbol{\omega_k}) D_{\omega_{k_1}} \otimes D_{\omega_{k_2}}\otimes D_{\omega_{k_3}}\otimes D_{-\tiny{\sum_{i=1}^{3}\omega_{k_i}}},\Phi(\boldsymbol{\omega_s}) D_{-\omega_{s_1}} \otimes D_{-\omega_{s_2}}\otimes D_{-\omega_{s_3}}\otimes D_{\tiny{\sum_{i=1}^{3}\omega_{s_i}}} \Big) \nonumber
\end{align}}
\normalsize
By Theorem \ref{prodcumthm}, the cumulant term implies we are looking for all indecomposable partitions of the array
\[\begin{matrix}
\underbrace{D_{\omega_{k_1}}}_{1}&\underbrace{D_{\omega_{k_2}}}_2 & \underbrace{D_{\omega_{k_3}}}_3& \underbrace{ D_{-\tiny{\sum_{i=1}^{3}\omega_{k_i}}}}_4\\ 
\underbrace{D_{-\omega_{s_1}}}_{5}&\underbrace{-D_{\omega_{s_2}}}_6 & \underbrace{D_{-\omega_{s_3}}}_7& \underbrace{ D_{\tiny{\sum_{i=1}^{3}\omega_{s_i}}}}_8 
\end{matrix} \tageq \label{eq:decF4}\]
We shall ignore the  $\Phi(\boldsymbol{\omega_k})$ and $\Phi(\boldsymbol{\omega_s})$ as this will not change the order of the variance. Indecomposability implies the rows must hook so at least one tensor from the first row must be in the same component with an element from the second row. Observe that for partitions of which one element consists of a at least 4 fDFT tensors,  Lemma \ref{FCboundssd} implies at least two constraints on the summation will enter. Moreover, such an element is at most of order $O(T^{-1})$ in norm. Hence, \eqref{eq:covF4hat} is at most of order $O(\frac{b_4^4 T^4}{b_4^6 T^4}\frac{1}{T}) = O(\frac{1}{b_4^2 T})$ in $S(H \otimes H)$. We therefore only have to consider those partitions consisting of tensor products of two fDFT's. Notice that we at least will have to impose three restrictions in order to make such terms not disappear. For example,
  \begin{align*}
S_{(15)(26)(37)(48)}\
\end{align*}
implies the restrictions $k_1-s_3 \equiv 0 \mod T$, $k_2-s_2 \equiv 0 \mod T$ and $k_3-s_3 \equiv 0 \mod T$ . Consequently, the covariance becomes
\begin{align*}
&\biggsnorm{\text{Cov}\Big(\hat{\F}_{\omega_{j_1},\omega_{j_2},\omega_{j_3}}, \hat{\F}_{\omega_{l_1},\omega_{l_2},\omega_{l_3}}\Big)}_2\le \nonumber\\&
\sup_{\omega}\snorm{\F_{\omega}}^4_2\frac{(2\pi)^{4}}{ b_4^6 T^{4}}\bigg| \sum_{k_1,k_2, k_3}  K_4\Big(\frac{\omega_{j_1}-{\omega_{k_1}}}{b_4},\ldots,\frac{\omega_{j_4}+\tiny{\sum_{i=1}^{3}\omega_{k_i}}}{b_4}\Big)   K_4\Big(\frac{\omega_{l_1}+{\omega_{k_1}}}{b_4},\ldots,\frac{\omega_{l_4}-\tiny{\sum_{i=1}^{3}\omega_{k_i}}}{b_4}\Big) \bigg|\\& 
=O\bigg(\frac{1}{ b_4^6 T^{4}} \times b_4^3 T^3\bigg) = O\bigg(\frac{1}{b_4^3 T}\bigg),
\end{align*}
showing the estimator is consistent as $b_4 \to 0$ in such a way that $b_4^3 T \to \infty$. The last statement \eqref{eq:intF4} now follows from a bias-variance decomposition as in the proof of Theorem \ref{UnifCon}(i) and from noting that the Riemann-approximation does not change the order.
\end{proof}

\subsection{Distributional properties under $H_A$}\label{sec:F4HA}

\begin{theorem}\label{thm:intF4loc}
Suppose \ref{cumglsp}(8,\,2) holds.  Then,
\begin{align*}
&(i)~\biggsnorm{\E \frac{(2\pi)^2}{T^2}\sum_{j_1,j_2=1}^{T}\hat{\F}_{\omega_{j_1},-\omega_{j_1+h},-\omega_{j_2}} - \int \int  {G}_{\omega,-\omega +\omega_{h},-\omega^{\prime}}d\omega d\omega^{\prime}-\mathcal{Z}_h}_2 =O\bigg(\frac{1}{b_4 T}+b_4\bigg), 
 \\
 &(ii)~\bigsnorm{\Cov( \hat{\F}_{\omega_{j_1},\omega_{j_2},\omega_{j_3}},\hat{\F}_{\omega_{j_1},\omega_{j_2},\omega_{j_3}}) }^2_2
 =O\bigg(\frac{1}{b_4^3 T}\bigg),
\end{align*}
where $G_{\omega,-\omega +\omega_{h},-\omega^{\prime}}$ denotes the time-integrated tri-spectral operator and where $\mathcal{Z}_h \in S(H \otimes H)$ is a bias term of order $O(\snorm{\mathcal{Z}_h}_2)=1$.
\end{theorem}
\begin{proof}[Proof of Theorem \ref{thm:intF4loc}]

Using Theorem \ref{prodcumthm} and Lemma \ref{cumboundglsp} we find for the expectation of the tri-spectral operator estimator
\begin{align*}
&\E \hat{\F}_{\omega_{j_1},\omega_{j_2},\omega_{j_3}}=\frac{(2\pi)^{2}}{b_4^3 T^2}\sum_{k_1,k_2, k_3}  K\bigg(\frac{\omega_{j_1}-{\omega_{k_1}}}{b_4}\bigg) K\bigg(\frac{\omega_{j_2}-{\omega_{k_2}}}{b_4}\bigg) K\bigg(\frac{\omega_{j_3}-{\omega_{k_3}}}{b_4}\bigg) K\bigg(\frac{\omega_{j_4}+\tiny{\sum_{i=1}^{3}\omega_{k_i}}}{b_4}\bigg)\\& 
 \phantom{\frac{(2\pi)^{2}}{b_4^3 T^2}\sum_{k_1,k_2, k_3}}
\times \Phi(\boldsymbol{\omega_k})\Bigg[\frac{2\pi}{T^2} \sum_{t=1}^T{\F}_{t/T; \omega_{k_1},\omega_{k_2},\omega_{k_3}}+R_{4,T}
 \\&  \phantom{\frac{(2\pi)^{2}}{b_4^3 T^2}\sum_{k_1,k_2, k_3}}
+ \bigg(\frac{1}{T}  \sum_{t=1}^{T} \F_{t/T;\omega_{k_1}}e^{-\im t(\omega_{k_1}+\omega_{k_2})} +R_{2,T}\bigg)
\otimes \bigg( \frac{1}{T} \sum_{s=1}^{T} \F_{s/T;\omega_{k_3}}e^{\im s(\omega_{k_1}+\omega_{k_2})} +R_{2,T}\bigg)
 \\& \phantom{\frac{(2\pi)^{2}}{b_4^3 T^2}\sum_{k_1,k_2, k_3}}
  + \bigg(\frac{1}{T} \sum_{t=1}^{T} \F_{t/T;\omega_{k_1}}e^{-\im t(\omega_{k_1}+\omega_{k_3})} +R_{2,T}\bigg)
  \widetilde{\otimes} \bigg( \frac{1}{T} \sum_{s=1}^{T} \F_{s/T;\omega_{k_2}}e^{\im s(\omega_{k_1}+\omega_{k_3})}+R_{2,T} \bigg)\ \\&  \phantom{\frac{(2\pi)^{2}}{b_4^3 T^2}\sum_{k_1,k_2, k_3}}
+\bigg(  \frac{1}{T} \sum_{t=1}^{T} \F_{t/T;\omega_{k_1}}e^{\im t(\omega_{k_2}+\omega_{k_3})} +R_{2,T}\bigg)
\otimes _{\top} \bigg(\frac{1}{T} \sum_{s=1}^{T} \F_{s/T;\omega_{k_2}}e^{-\im s(\omega_{k_2}+\omega_{k_3})}+R_{2,T}\bigg)\Bigg].
\end{align*}
For the first term we note that a similar argument as in the stationary case yields
\begin{align*}
\biggsnorm{\E \frac{(2\pi)^2}{T^2}\sum_{j_1,j_2=1}^{T}\hat{\F}_{\omega_{j_1},-\omega_{j_1+h},-\omega_{j_2}} - \int \int  {G}_{\omega,-\omega +\omega_{h},-\omega^{\prime}}d\omega d\omega^{\prime}}_2 =O\bigg(\frac{1}{b_4 T}+b_4^2\bigg).
\end{align*}
We now turn to the three terms consisting of operators of second-order cumulant tensors of 
\[ \E \frac{(2\pi)^2}{T^2}\sum_{j_1,j_2=1}^{T}\hat{\F}_{\omega_{j_1},-\omega_{j_1+h},-\omega_{j_2}},\] which can be written as
 \begin{align*}
\frac{(2\pi)^2}{T^2}\sum_{j_1,j_2=1}^{T}  \frac{(2\pi)^{2}}{b_4^3 T^2}\sum_{k_1,k_2, k_3} & K(\frac{\omega_{j_1}-{\omega_{k_1}}}{b_4}) K(\frac{-\omega_{j_1+h}-{\omega_{k_2}}}{b_4}) K(\frac{-\omega_{j_2}-{\omega_{k_3}}}{b_4}) K(\frac{\omega_{j_2+h}+\tiny{\omega_{k_1}+\omega_{k_2}+\omega_{k_3}}}{b_4})
  \\& \times \Phi(\boldsymbol{\omega_k}) \Bigg[
 (\tilde{\F}_{k_1+k_2;\omega_{k_1}} +R_{2,T})\otimes  ( \tilde{\F}_{-k_1-k_2;\omega_{k_3}} +R_{2,T})
\\& \phantom{ \times \Phi(\boldsymbol{\omega_k}) \Bigg[}   + (\tilde{\F}_{k_1+k_3;\omega_{k_1}} +R_{2,T})\widetilde{\otimes} ( \tilde{\F}_{-k_1-k_3;\omega_{k_2}+\omega_{k_3}}+R_{2,T}) \\& \phantom{ \times \Phi(\boldsymbol{\omega_k}) \Bigg[}   + (\tilde{\F}_{-k_2-k_3;\omega_{k_1}} +R_{2,T})\otimes_{\top}  ( \tilde{\F}_{k_2+k_3;\omega_{k_2}}+R_{2,T})\Bigg].
\end{align*}
Let us focus on the first term. Using a change of variables and Corollary \ref{cumbounds} this becomes
\begin{align*}
&\frac{(2\pi)^2}{T^2}\sum_{j_1,j_2=1}^{T} \frac{(2\pi)^{2}}{b_4^3 T^2}\sum_{k_1,l, k_3}  K(\frac{\omega_{j_1}-{\omega_{k_1}}}{b_4}) K(\frac{-\omega_{j_1+h}-\omega_{l}+\omega_{k_1}}{b_4}) K(\frac{-\omega_{j_2}-{\omega_{k_3}}}{b_4}) K(\frac{\omega_{j_2+h}+\omega_{l}+\omega_{k_3}}{b_4})
  \\& \qquad \times\Phi(\omega_{k_1}, \omega_l-\omega_{k_1}, \omega_{k_3}, \omega_l+\omega_{k_3})\big(\tilde{\F}_{l;\omega_{k_1}} +R_{2,T}\big)\otimes  \big( \tilde{\F}_{-l;\omega_{k_3}} +R_{2,T}\big)
  \\& \le \frac{(2\pi)^2}{T^2}\sum_{j_1,j_2=1}^{T}  \Big| \frac{(2\pi)^{2}}{b_4^3 T^2}\sum_{k_1, k_3}  K(\frac{\omega_{j_1}-{\omega_{k_1}}}{b_4}) K(\frac{-\omega_{j_1+h}+\omega_{k_1}}{b_4}) K(\frac{-\omega_{j_2}-{\omega_{k_3}}}{b_4}) K(\frac{\omega_{j_2+h}+\omega_{k_3}}{b_4})\Big|\sup_{\omega}\sum_{l}\snorm{\tilde{\F}_{l;\omega}}^2_2
\\& \quad \quad 
 + O(\frac{1}{T^2} b^4 T^5 \frac{1}{b_4^3 T^2}\frac{1}{T})+O( \frac{1}{T^2}\frac{b_4^3 T^4}{b_4^4 T^3}) 
\\&  =O(1+b_4+\frac{1}{b_4 T}),
  \end{align*}
 where the second error term is a consequence of the remainders $R_{2,T}$ and where the third term follows from replacing the arguments in the second and fourth smoothing kernels. 
To prove (ii) note that we are also in this case looking at the indecomposable partitions of the arrary in \eqref{eq:decF4}. It is immediate that the terms of highest order are those with second-order cumulants and those with one fourth-order tensor and two second-order cumulant tensors. The latter is easily verified to impose two constraints on the summations by Corollary \ref{cumbounds} this will be of order $O(\frac{b_4^4 T^4}{b_4^6 T^4}\frac{1}{T}) = O(\frac{1}{b_4^2 T})$. The highest second-order partitions contain two sets with each one element from distinct rows. The remaining two sets are decomposable and take their elements in the same row. In particular, we find that Corollary \ref{cumbounds}(ii) implies in this case at least three bounded summations from which the result follows. \end{proof}

\section{Functional versus multivariate methods}
\label{func_not_mult}

Functional results are non-trivial extensions of their counterpart multivariate statements, even though the results are, as in this paper, often based on projections. This can be seen, for instance, through a simple example. Define the first-order functional autoregression $X_j=\Phi X_{j-1}+\varepsilon_j$ with
\[
\Phi(x)=a\big(\langle x,e_1\rangle+\langle x,e_2\rangle\big)e_1+a\langle x,e_1\rangle e_2,
\qquad x\in H,
\]
where $a\in(0,1)$ and $e_1,e_2\in H$ orthonormal. Assume $E[\langle\varepsilon_j,e_1\rangle^2]>0$ but $E[\langle\varepsilon_j,e_2\rangle^2]=0$. Then, the first fPC score series satisfies
\[
\langle X_j,e_1\rangle = a\langle X_{j-1},e_1\rangle +a^2\langle X_{j-2},e_1\rangle+\langle\varepsilon_j,e_1\rangle.
\]
It is seen that the projection of this FAR(1) process is an AR(2) process. So there is a complex interplay at work between functional time series and their projections onto finite-dimensional subspaces, even at the population level. The relationship becomes more intricate if population quantities are replaced by their sample counterparts. The extension to the functional level is therefore complicated, as the dynamics of a functional time series may not be captured by its finite-dimensional projections and further refinements and extensions of methods known for the latter case are needed.


\end{document}